\documentclass[11pt]{article}
\usepackage{graphicx} % Required for inserting images
\usepackage{mathtools}
\usepackage{enumitem}
\input{defs}
\usepackage{tikz}
\usetikzlibrary{patterns}
\input{qcircuit}
\usepackage[margin=1in]{geometry}
\allowdisplaybreaks
\usepackage{thm-restate}

\title{Fourier Spectrum of Noisy Quantum Algorithms}
\author{Uma Girish\thanks{  Columbia University. Email:  \href{mailto:ug2150@columbia.edu}{ug2150@columbia.edu}}} 
\date{}

\begin{document}

\maketitle

\begin{abstract}

Quantum computing promises exponential speedups for certain problems, yet fully universal quantum computers remain out of reach and near-term devices are inherently noisy. Motivated by this, we study noisy quantum algorithms and the landscape between $\BQP$ and $\BPP$. We build on a powerful technique to differentiate quantum and classical algorithms called the level-$\ell$ Fourier growth (the sum of absolute values of Fourier coefficients of sets of size $\ell$) and show that it can also be used to differentiate quantum algorithms based on the types of resources used. We show that noise acting on a quantum algorithm dampens its Fourier growth in ways intricately linked to the type of noise.

Concretely, we study noisy models of quantum computation where highly mixed states are prevalent, namely: $\DQC{k}$ algorithms, where $k$ qubits are clean and the rest are maximally mixed, and $\hBQP$ algorithms, where the initial state is maximally mixed, but the algorithm is given knowledge of the initial state at the end of the computation. We establish upper bounds on the Fourier growth of $\DQC{k}$, $\hBQP$ and $\BQP$ algorithms and leverage the differences between these bounds to derive oracle separations between these models. In particular, we show that \forrtwo~and \forrthree~require $N^{\Omega(1)}$ queries in the $\DQC{1}$ and $\hBQP$ models respectively. Our results are proved using a new matrix decomposition lemma that might be of independent interest.

%The study of separations between $\BQP$ and $\BPP$ is a central topic in complexity theory with a rich body of research. In recent years, there has been great interest in studying intermediate models of quantum computing and in understanding the space between $\BQP$ and $\BPP$. Motivated by this, we study noisy models of quantum computation where highly mixed states are prevalent, in particular:   $\DQC{k}$ algorithms, where $k$ qubits are clean and the rest are maximally mixed, and $\hBQP$ algorithms, where the initial state is maximally mixed, but the algorithm is given knowledge of the initial state at the end of the computation.
   
%In this work, we undertake a Fourier analytic study of the differences between $\DQC{k}$, $\hBQP$ and $\BQP$. We build on a powerful technique to differentiate quantum and classical algorithms called the level-$\ell$ Fourier growth, which captures the sum of absolute values of Fourier coefficients of sets of size $\ell$.  We establish upper bounds on the Fourier growth of $\DQC{k}$, $\hBQP$ and $\BQP$ algorithms and use them to derive oracle separations between these models. In particular, we show that \forrtwo~and \forrthree~require $N^{\Omega(1)}$ queries in the $\DQC{1}$ and $\hBQP$ models respectively. Our work shows that Fourier growth doesn't just distinguish quantum and classical algorithms, it can also differentiate between quantum algorithms based on the kind of resources used. 
\end{abstract}

\tableofcontents
\newpage

\section{Introduction} 

Quantum computing promises to solve certain problems exponentially faster than classical computers, as evidenced by numerous query complexity separations or oracle separations~\cite{DJ92,BV97,Sim97,Aar10,AA15}. Yet, we haven't been able to harness this, as we are far from being able to build fully universal quantum computers. While $\BQP$ algorithms generally assume noiseless computation, noise is arguably the most significant issue faced by near-term quantum computers and all current quantum devices are inherently noisy. 
To better understand what quantum resources are truly responsible for quantum advantage, researchers have proposed numerous intermediate models of quantum computing like $\IQP,\DQC{1},\NISQ$ and $\textsf{Boson Sampling}$~\cite{IQP,KL98,NISQ,Boson,ABKM16}. These models isolate specific quantum features -- such as having a few clean qubits or limited adaptivity -- and allow us to probe the  quantum landscape below $\BQP$. Although these models likely do not capture the full power of quantum computing, their precise relationship to $\BQP$ and to each other remains poorly understood. This raises a natural question: \begin{quote}
\centering
    \emph{What does the landscape of quantum computation below $\BQP$ look like? }
\end{quote}
In our work, we study this question from a Fourier analytic perspective. In particular, we study the level-$\ell$ Fourier growth of the acceptance probability of algorithms (\Cref{def:fourier_growth}). This is a measure of how well-spread the Fourier coefficients are. In our work, we show that Fourier growth is not just a tool for distinguishing quantum and classical models; it is a fine-grained tool capable of differentiating quantum models based on the kinds of quantum resources they utilize. We focus on noisy quantum algorithms and demonstrate that noise dampens the Fourier growth  in ways that are intricately linked to the type of noise present. 

In particular, we study noisy models like $\DQC{k}$, where $k$ qubits are clean and the rest are maximally mixed~\cite{KL98,morimae2014hardness}, and $\hBQP$, where the initial state is maximally mixed, i.e., a uniformly random computational basis state, but the algorithm is given knowledge of this initial state at the end of the computation ~\cite{ABKM16,JM24}. We prove Fourier growth bounds on the acceptance probability of $\DQC{k},\hBQP$ and $\BQP$ algorithms (\Cref{thm:main_theorem_dqck,thm:main_theorem_hbqp,thm:main_theorem_bqp}) and use the differences in these bounds to derive oracle separations between these models. In particular, we show that \forrtwo~and \forrthree, which can be solved with two queries in the $\hBQP$ and $\BQP$ models respectively, require $N^{\Omega(1)}$ queries in the $\DQC{1}$ and $\hBQP$ models respectively (\Cref{cor:main_theorem_1_1,cor:main_theorem_2}), resolving two conjectures from~\cite{JM24} and establishing the first oracle separation between $\hBQP$ and $\DQC{k}$, as well as a new oracle separation between $\BQP$ and $\hBQP$. 

We believe that the noise-induced dampening of Fourier growth is a more general phenomenon, and that the techniques developed here could shed light on other noisy models such as $\NISQ$. Our results are proved using a new matrix decomposition lemma that encodes information about indices in a matrix product that might be of independent interest. 

%there are few settings in which we can unconditionally establish separations between quantum and classical capabilities. One important example of such a setting is query complexity. Over the past several decades, numerous works~\cite{DJ92,BV97,Sim97,Aar10,AA15} have shown exponential query complexity separations (oracle separations) between quantum and classical algorithms. A powerful technique that has emerged in this context is a Fourier analytic technique known as \emph{Fourier growth bounds}, which we describe below.%This insight has led to several results in the field of quantum versus classical separations, including separations between $\BQP$ and $\mathsf{PH}$~\cite{RT22}, and optimal separations between $\BQP$ and $\BPP$~\cite{Aar10,AA15,Tal20,BS21,sherstov2023optimal}. 
\subsection{The Space Below $\BQP$}

%Over the past several decades, numerous works have shown examples of problems with quantum query protocols of small cost, but no classical query protocols of small cost~\cite{Sim97,DJ92,BV97,Aar10,AA15}. %In this context, the \emph{Forrelation problem}~\cite{AA15,Aar10,RT22} has emerged as a central problem. In the Forrelation Problem, the input is $x,y\in\bin^N$ for $N=2^n$ and the goal is to tell whether $\tfrac{1}{N}\abra{x|H_N|y}$ is close to zero or far from zero, where $H_N=\tfrac{1}{\sqrt{N}}\big[\begin{smallmatrix}1& 1\\ 1 & -1\end{smallmatrix}\big]^{\otimes n}$ is the unitary Hadamard matrix. A quantum algorithm can solve this problem \emph{one} query to the oracle $O_{x,y}$, however, classical algorithms requires $\tilde{\Omega}(N^{1/2})$ queries. This has led to various oracle separations between $\BQP$ and $\BPP$~\cite{Aar10,AA15,Tal20,sherstov2023optimal,BS21} and $\mathsf{PH}$~\cite{RT22}. %It has since been used in various other contexts, for example, a construction of a classical oracle relative to which $\mathsf{P}=\mathsf{NP}$ but $\mathsf{BQP} \neq \mathsf{QCMA}$ ~\cite{AIK22} and quantum cryptography exists~\cite{KQST23}, and separations between adaptive and non-adaptive quantum algorithms~\cite{GSTW24}. Most works on the Forrelation problem build on a Fourier analytic framework introduced by~\cite{RT22} involving Fourier growth bounds, which we describe below. While these results delineate the boundary between quantum and classical algorithms, 

The landscape of computational models between $\BQP$ and $\BPP$ is vast and intricate. There are numerous intermediate models of quantum computation like $\IQP,\DQC{1},\NISQ$ and $\textsf{Boson Sampling}$ \cite{IQP,KL98,NISQ,Boson} with constraints on the quantum resources. The study of such intermediate models serves two key purposes: (1) to systematically delineate the boundary between classical and quantum algorithms and pinpoint the minimal resources for quantum speedups, and (2) to model the physical constraints of near-term quantum devices and reason about them.

One important issue that affects near-term quantum computers is noise. Unlike classical systems, quantum computers are highly susceptible to various types of errors due to decoherence, imperfect gates, and environmental interactions. This prompts a natural question -- how much noise can quantum algorithms tolerate? How does noise change the quantum computational complexity landscape? This is challenging to answer in general, as there are many different kinds of noise that affect quantum algorithms. One way to simplify this challenge is to consider models that have an extreme amount of noise. In this endeavor, researchers have proposed highly noisy models of quantum computation like $\DQC{k}$ and $\hBQP$~\cite{KL98,morimae2014hardness,JM24} where all the noise is pushed onto the initial state -- the qubits start maximally or nearly maximally mixed, while the gates are noiseless. These models provide a framework for understanding the minimal number of clean qubits required to achieve quantum speedups. We describe these models below.

\begin{figure}
\centering
\mbox{ 
\Qcircuit @C=1em @R=.7em {
\lstick{} & &  &  \qw & \multigate{5}{U_1}  & \multigate{1}{O_x} & \multigate{5}{U_2}   & \multigate{1}{O_x}  &    \qw &  \cds{1}{\cdots\cdots}  & \multigate{5}{U_d} & \multigate{1}{O_x} & \multigate{5}{U_{d+1}} &\meter \\
\lstick{} & &  &  \qw  & \ghost{U_1}  & \ghost{O_x} & \ghost{U_2}   & \ghost{O_x} & \qw & \cds{1}{\cdots\cdots} & \ghost{U_d} & \ghost{O_x} &\ghost{U_{d+1}} &\meter  \inputgroupv{1}{2}{.8em}{.8em}{n} \\ 
\lstick{} &  &  &  \qw & \ghost{U_1} &\qw  &  \ghost{U_2}  & \qw &\qw & \cds{1}{\cdots\cdots} & \ghost{U_d} &\qw  &\ghost{U_{d+1}} &  \meter   \\
\lstick{} &  &  &  \qw & \ghost{U_1} &\qw  &  \ghost{U_2}  & \qw &\qw & \cds{1}{\cdots\cdots} & \ghost{U_d}  & \qw &\ghost{U_{d+1}} &  \meter  \inputgroupv{3}{4}{.8em}{.8em}{w} \\
\lstick{}  & \ket{0}   &  &  \qw  & \ghost{U_1}& \qw  &  \ghost{U_2}  &\qw  & \qw  &\cds{1}{\cdots\cdots} & \ghost{U_d}  & \qw &\ghost{U_{d+1}} &\meter \\
\lstick{} & \ket{0} & &  \qw  & \ghost{U_1}& \qw  & \ghost{U_2}   &  \qw  & \qw  &  \qw  &  \ghost{U_d} & \qw &\ghost{U_{d+1}}   & \meter \inputgroupv{5}{6}{.8em}{.8em}{k} \\
}
}
\caption{A $d$-query $\DQC{k}$ algorithm. The initial state on the first $n+w$ qubits is maximally mixed.}
\label{fig:dqc1}
\end{figure}

\paragraph*{$\DQC{k}$} Drawing inspiration from the NMR approach to quantum computing where mixed states are ubiquitous, Knill and Laflamme~\cite{KL98} introduced the one-clean qubit or $\DQC{1}$ model as an idealized version of a noisy quantum computer. In this model, one qubit is clean (noiseless) and the rest are maximally noisy, and the algorithm can apply (noiseless) unitary gates on these qubits and measure at the end. This model was later generalized  to $\DQC{k}$ to allow $k$ clean qubits~\cite{morimae2014hardness,fujii2015power}. This model does not seem to be universal for quantum computing since all qubits except a few are maximally noisy and many oracle problems like Simon's problem and order finding are not believed to be solvable in this model. Despite this, $\DQC{1}$ can solve problems that are believed to be classically hard, like estimating the trace and Pauli coefficients of a unitary matrix described by a quantum circuit~\cite{KL98,entanglement05}, Jones polynomials~\cite{jones08}, partition functions~\cite{partition21}. Under complexity theoretic assumptions, this model is not classically simulable~\cite{morimae2014hardness,fujii2018impossibility,morimae2017hardness}. There are exponential oracle separations between $\DQC{1}$ and $\BPP$~\cite{shepherd}. The communication version of the one clean qubit model provides exponential speedups over classical randomized communication~\cite{AGL24}. The fact that quantum speedups persist even under such extreme noise makes $\DQC{1}$ a particularly intriguing model for further study -- it challenges our understanding of what minimal quantum resources are required for speedups.

\begin{figure}
\centering
\mbox{ 
\Qcircuit @C=1em @R=.7em {
\lstick{} & &  &  \qw &\ctrl{3} & \qw& \qw &  \multigate{2}{U_1}   & \multigate{1}{O_x}  &    \qw &  \cds{1}{\cdots\cdots} & \multigate{2}{U_d}& \multigate{1}{O_x} & \multigate{2}{U_{d+1}}  & \meter \\
\lstick{} & &  &  \qw & \qw &\ctrl{3}& \qw & \ghost{U_1}  & \ghost{O_x} & \qw & \cds{1}{\cdots\cdots} &   \ghost{U_d}& \ghost{O_x} & \ghost{U_{d+1}} &\meter \\ 
\lstick{} &  &  &  \qw& \qw &\qw  & \ctrl{3} & \ghost{U_1} &\qw  &\qw & \cds{1}{\cdots\cdots}  & \ghost{U_d} & \qw & \ghost{U_{d+1}}  &  \meter     \inputgroupv{1}{3}{.8em}{1.8em}{n+w\hspace{2em}}\\
\lstick{} &\ket{0}  &  &  \qw &\targ& \qw   &\qw  &  \qw  & \qw &\qw & \cds{1}{\cdots\cdots} & \qw  & \qw & \qw &  \meter   \\
\lstick{}  & \ket{0}   &  &  \qw &\qw &\targ   & \qw  &  \qw   &\qw  & \qw  &\cds{1}{\cdots\cdots} & \qw& \qw  & \qw  &\meter \\
\lstick{} & \ket{0} & &  \qw &\qw & \qw &\targ     & \qw    &  \qw  & \qw  &  \qw& \qw  &  \qw & \qw   & \meter \inputgroupv{4}{6}{.8em}{1.8em}{n+w\hspace{2em}} \\
}
}
\caption{A $d$-query $\hBQP$ algorithm. The initial state on the first $n+w$ qubits can be thought of as maximally mixed, or as the pure state $2^{-(n+w)/2}\sum_{I\in\{0,1\}^{n+w}}\ket{I}$; the resulting circuits are equivalent.}
\label{fig:hbqp}
\end{figure}

\paragraph*{$\hBQP$.} The $\hBQP$ model was originally defined by~\cite{ABKM16} to capture the power of permutational computations on
special input states. This model was revisited by~\cite{JM24} in the context of delineating the boundary between $\BQP$ and $\DQC{1}$. In this model, the initial state is maximally mixed, i.e., a uniformly random computational basis state, but the algorithm learns this state at the end of the computation and decides whether to accept or reject. One can equivalently define this model as a quantum algorithm acting on one half of a maximally entangled EPR state and in the end, we measure both halves and do classical postprocessing on the measurement outcomes. This model is not believed to be universal for quantum computing as it allows a significant amount of noise, yet, this model encapsulates many known quantum speedups. It can solve the factoring problem and numerous oracle-based problems including Simon's problem, Deutsch-Jozsa, order finding, and the Forrelation problem and can simulate $\DQC{k}$ for any small $k$ as well as $\IQP$~\cite{JM24}. It appears to be the weakest quantum model that is unlikely to be universal and yet captures most known $\BQP$ speedups despite operating on maximally mixed states.

\paragraph*{}
A powerful and natural framework to study the differences between $\DQC{k},\hBQP$ and $\BQP$ is query complexity. In this setting, there is a boolean function $f:\bin^N\to \bin$ and the goal is to compute $f(x)$ for $x\in \bin^N$ while minimizing the number of queries to the oracle $O_x$. This model strips away implementation details and captures the essence of what makes different computational models powerful. The aforementioned quantum models can be formalized using this framework and are depicted in~\Cref{fig:dqc1,fig:hbqp}. (See~\Cref{def:dqck,def:hbqp_algorithm} for more details.) Query complexity has long been one of the most fruitful arenas for understanding the differences between quantum and classical computation and gives us  strong evidence for quantum advantage, including provable exponential oracle separations between $\BQP$ and $\BPP$. Over the years, the field has also developed an impressive arsenal of lower-bound techniques for both quantum and classical algorithms.
While these techniques are powerful for distinguishing quantum from classical, they are not designed to distinguish between quantum algorithms. Indeed, many of these methods -- including the polynomial method -- apply uniformly to all bounded low-degree polynomials and cannot capture the subtle differences between $\DQC{1},\hBQP$ and $\BQP$. This motivates the search for more fine-grained analytic techniques. 

The central contribution of this paper is to show that a Fourier analytic concept known as Fourier growth provides exactly such a tool. While Fourier growth was historically used to distinguish between quantum and classical algorithms, we demonstrate that it can also serve as a lens to separate quantum models from each other. We show that noise dampens the Fourier growth of quantum algorithms in ways that are intricately tied to the noise patterns. We now introduce Fourier growth, provide its historical context and describe its importance. 
%{\bf $\DQC{k}$ in query complexity.} A $d$-query $\DQC{k}$ algorithm starts with $k$ clean qubits and additional noisy qubits, applies a sequence of $d$ unitary operators and oracle $O_x$ operators interleaved, finally measures all the qubits and accepts a subset of measurement outcomes. This is depicted in~\Cref{fig:dqc1}.
%{\bf $\hBQP$ in query complexity.} A $d$-query $\hBQP$ algorithm starts with qubits in the maximally mixed state, applies a sequence of $d$ unitary operators and oracle gates interleaved, finally measures all the qubits. It then learns the initial state of the algorithm and can decide whether to accept or reject. This is depicted in~\Cref{fig:hbqp}.  

\begin{comment} The results of~\cite{RT22} imply that any family of functions $\cF$ that is closed under restrictions has advantage at most $ O\pbra{ L_{1,2}(\cF)}/\sqrt{K}$ in solving the Forrelation problem on inputs of length $K$. We then apply this to the family of all functions expressible as the acceptance probability of a $d$-query $\DQC{k}$ algorithm and take its closure under restrictions. Applying~\Cref{thm:main_theorem_dqck} immediately gives us the following corollary. 
\end{comment}

\subsection{Fourier Growth}

Fourier growth has emerged as a central concept that allows us to distinguish quantum and classical algorithms. To formally define Fourier growth, recall that every boolean function $f:\bin^N\to \R$ can be uniquely represented as a multi-linear polynomial
\[ f(x)=\sum_{S\subseteq[N]} \widehat{f}(S)\cdot \prod_{i\in S}x_i \]
where $\widehat{f}(S)$ are called the Fourier coefficients of $f$.

\begin{definition}[Signed Fourier Growth] \label{def:signed_fourier_growth}
For level $\ell\in \N$, and signs $\alpha_S\in[-1,1]$ for $S\subseteq[N]$ with $|S|=\ell$, define the $\alpha$-signed level-$\ell$ Fourier growth of $f$, denoted by $L_{1,\ell}^\alpha(f)$ as \[ L_{1,\ell}^\alpha(f):=\sum_{\substack{S\subseteq[N]\\|S|=\ell}} \alpha_S\cdot \widehat{f}(S),\] 
\end{definition}

\begin{definition}[Fourier Growth]
    \label{def:fourier_growth}
    For level $\ell\in\N$, the level-$\ell$ Fourier growth of $f$, denoted by $L_{1,\ell}(f)$, is the $\ell_1$-norm of the level-$\ell$ Fourier coefficients of $f$, 
\[ L_{1,\ell}(f) := \sum_{\substack{S\subseteq[N]\\|S|=\ell}}\abs{\widehat{f}(S)} = \max_{\alpha\in[-1,1]^{\binom{N}{\ell}}} L_{1,\ell}^\alpha(f).\] 
\end{definition}

%For $\{-1,1\}$-valued functions, the $\ell_2$-norm of the Fourier coefficients is one, and hence, the $\ell_1$-norm of the Fourier coefficients is a measure of how spread apart the Fourier coefficients are. 
 
%\input{plot_1} 

\label{sec:prior} Fourier growth bounds have been extensively studied and established for various classical models\footnote{By Fourier growth of a model, we refer to the Fourier growth of the acceptance probability of an algorithm in this model.}, including small-width DNFs/CNFs~\cite{Mansour95}, $\mathsf{AC}^0$ circuits~\cite{Tal17}, low-depth decision trees~\cite{Tal20,sherstov2023optimal}, low-degree $\mathsf{GF}(2)$ polynomials~\cite{CHLT19}, low-depth parity decision trees~\cite{GTW21}, low-degree bounded polynomials~\cite{IRR+21}, and more. Upper bounds on the Fourier growth, even for the first few levels, give rise to quantum versus classical separations. Intuitively, while both quantum and classical algorithms of
small query complexity can be represented by low-degree polynomials, the polynomials associated with quantum algorithms are a lot ``denser'' compared to their classical analogues, and this density is captured by Fourier growth. In particular, it was shown by~\cite{Tal20,sherstov2023optimal} that for $d$-query classical algorithms, $L_{1,\ell}(f)$ is at most $ \tilde{O}_\ell(d^{\ell/2})$; on the other hand, for $d$-query quantum algorithms, $L_{1,\ell}(f)$ is at most $O_\ell(d^\ell)\cdot N^{(\ell-1)/2}$~\cite{IRR+21} and this can be tight for certain algorithms. %See~\Cref{fig:quantum_classical_plot} for a depiction of this. 
A key problem that exploits this difference in the Fourier growth is the Forrelation problem. This was originally introduced by Aaronson and Ambainis~\cite{Aar10,AA15} to show an oracle separation between $\BQP$ and $\BPP$ and was subsequently used by Raz and Tal~\cite{RT22} in their breakthrough oracle separation of $\BQP$ and $\mathsf{PH}$. Building on this, \cite{BS21} generalized this to the $k$-Forrelation problem and used it to show optimal separations between $\BQP$ and $\BPP$. We describe this problem below.

\begin{definition}[$k$-Forrelation function] \label{def:forr_function} Let $N=2^n$. For  $x^{(1)},\ldots,x^{(k)}\in\bin^{N}$, define \[ \forr{k}(x^{(1)},\ldots,x^{(k)}):= \abra{ v\gap H_N \cdot O_{x^{(1)}} \cdot H_N\cdot O_{x^{(2)}}\cdots \cdots  H_N\cdot O_{x^{(k)}}\cdot  H_N \gap v}\]
where $H_N$ is the $N\times N$ unitary Hadamard matrix as in~\Cref{def:hadamard} and $\ket{v}=\ket{0\ldots 0}$.
\end{definition}

\begin{definition}[$k$-\textsc{Forrelation} problem with parameter $\varepsilon=\Theta(1/\log^k N)$]\label{def:forr_problem}  Given input $x\in\bin^{kN}$, return $-1$ if $\forr{k}(x)\ge 2\varepsilon$ and $1$ if $\forr{k}(x)\le \varepsilon$.
\end{definition}

Quantum algorithms in the $\BQP$ model can solve \forrk~using $\lceil k/2\rceil $ quantum queries. Furthermore, the results of~\cite{RT22,CHLT19,RT22,BS21} imply that any family of algorithms solving \forrk~must have large Fourier growth at levels $k,2k,\ldots,k (k-1)$ (see~\Cref{thm:rt22} and~\Cref{thm:bs21} for a precise statement). These results effectively reduce the task of proving lower bounds for the Forrelation problem to the task of establishing Fourier growth bounds. In particular, \forrtwo~involves level-2 bounds and \forrthree~involves level-3 and level-6 bounds. 
 %This has been used to show optimal oracle separations between $\BQP$ and $\BPP$~\cite{Tal20,sherstov2023optimal,BS21,BGGL22} and $\mathsf{PH}$~\cite{RT22} among other results. 
 Since classical algorithms have small Fourier growth at all levels, it follows from the aforementioned works that they cannot solve the Forrelation problem.

\subsection{Our Results}

In our work, we go beyond the idea of using Fourier growth to distinguish between quantum and classical algorithms. We show that although Fourier growth can be large for quantum algorithms, just how large it can be depends on the kind of quantum resources used and the types of noise present. In particular, we establish Fourier growth bounds for $\DQC{k},\hBQP$ and $\BQP$ algorithms. The bounds we obtain for $\ell=1,2,3$ are summarized in~\Cref{tab:tab1} and depicted in~\Cref{fig:fig4}.

 \begin{table}[] 
 \setlength{\tabcolsep}{10pt} % Default value: 6pt
\renewcommand{\arraystretch}{1.25} % Default value: 1
 \centering
\begin{tabular}{|c|c|c|c|}
\hline
\begin{tabular}[c]{@{}c@{}}Fourier Growth of \\ $d$-query algorithms\end{tabular} & $\ell=1$ & $\ell=2$& $\ell=3$       \\ \hline
$\BQP$ & $d$  & $\tilde{O}(d\sqrt{N})$   & $\tilde{O}(dN)$ \\ 
~\cite{IRR+21},~\Cref{thm:main_theorem_bqp} & & &    \\ 
\hline
$\hBQP$   & $O(d)$ & $\tilde{O}(d\sqrt{N})$   & $O(d^7\sqrt{N})$   \\ 
\Cref{thm:main_theorem_hbqp} & & &  \\ \hline
$\DQC{1}$   & $O(d)$  & $O(d^3)$  & $\tilde{O}(d^3\sqrt{N})$  \\  
\Cref{thm:main_theorem_dqck}& & &  \\ \hline
$\BPP$~\cite{Tal20,sherstov2023optimal}  & $O(\sqrt{d})$   & $O(d\sqrt{\log N})$ & $O(\sqrt{d^3}\log N)$ \\ \hline 
\end{tabular}
\caption{Upper Bounds on the Fourier growth of the acceptance probability of various $d$-query algorithms.}
\label{tab:tab1}
\end{table}

\paragraph*{$\DQC{k}$ algorithms.} 
\begin{theorem}\label{thm:main_theorem_dqck} Let $f(x)$ be the acceptance probability of a $d$-query $\DQC{k}$ algorithm and $\rho\in \{-1,1,\star\}^N$ be any restriction. Then, for all $\ell\ge 2$, we have \[ L_{1,\ell}(f|_\rho)\le c^\ell\cdot d^3\cdot \sqrt{K}\cdot N^{(\ell-2)/2}\cdot \log^{\ell-2}(N)\cdot \sqrt{\ell!}\] 
for some constant $c>0$ where $K=2^k$.
\end{theorem} 

We prove the $\ell=2$ version of this in~\Cref{sec:proof_dqck} and higher levels in~\Cref{sec:bootstrapping} and show that the dependence on $k$ and $N$ are individually optimal in~\Cref{sec:tightness_DQC}. Here, the dependence on $N$ is particularly interesting. As we will see in~\Cref{thm:main_theorem_bqp}, the Fourier growth of $\DQC{1}$ algorithms falls short of the growth of general $\BQP$ algorithms by a factor of $\sqrt{N}$ at each level.  

\paragraph*{$\hBQP$ algorithms.} 
For the $\hBQP$ model, we are unable to prove $L_{1,3}$ and $L_{1,6}$ bounds. Currently, we do not have any upper bounds on $L_{1,3},L_{1,6}$ that are stronger than the ones for general $\BQP$ algorithms. Nevertheless, for our applications to Forrelation lower bounds, it turns out that we only need to deal with a certain family of signs, which we are able to do~\footnote[1]{\label{note1} We observe that~\cite{BS21} show that to establish lower bounds for \forrthree, one only needs to prove signed-Fourier growth bounds for a particular family of signs. (See~\Cref{def:alpha} and~\Cref{thm:bs21} for more details.) When we refer to the Fourier growth of $\hBQP$ algorithms, we typically mean signed-Fourier growth for signs as in~\Cref{thm:bs21,def:alpha}.}. 

\begin{theorem}\label{thm:main_theorem_hbqp} Let $f(x)$ be the acceptance probability of a $d$-query $\hBQP$ algorithm and $\rho \in\{-1,1,*\}^{3N}$ be any restriction. Let $\gamma,\gamma'\in[-1,1]^{3N}$ and $\alpha(\gamma)\in[-1,1]^{\binom{3N}{3}},\beta(\gamma,\gamma')\in[-1,1]^{\binom{3N}{6}}$ be signs as in~\Cref{def:alpha}. Then,
\[ L_{1,3}^{\alpha(\gamma)}(f|_\rho) \le O(d^7)\cdot \sqrt{N},\]
\[ L_{1,6}^{\beta(\gamma,\gamma')}(f|_\rho) \le O(d^{10})\cdot \sqrt{N^3} .\]
\end{theorem}

We prove the level-3 version in~\Cref{sec:hbqp} and the higher levels in~\Cref{sec:improved_bootstrapping}. We are unaware if this bound is tight, or if one can derive a similar bound for all families of signs (see~\Cref{sec:future})\footnote{We remark for this family of signs, $\BQP$ algorithms can already achieve a significantly larger Fourier growth. In particular, consider the acceptance probability $f(x)$ of the two-query $\BQP$ algorithm that solves \forrthree. For $\gamma=(1,\ldots,1)$, one can show that $L_{1,3}^{\alpha(\gamma)}(f)=\Omega(N)$.}.

\paragraph*{$\BQP$ algorithms.}  
\begin{theorem}\label{thm:main_theorem_bqp}
Let $f(x)$ be the acceptance probability of a $d$-query $\BQP$ algorithm and $\rho\in\{-1,1,\star\}^{N}$ be any restriction. Then,
\[L_{1,\ell}(f|_\rho)\le c^\ell\cdot d\cdot N^{(\ell-1)/2}\cdot \log^{\ell-1}(N)\cdot \sqrt{\ell!}\]
\end{theorem}
We prove this in~\Cref{sec:bootstrapping}. The dependence on $N$ is tight due to the $k$-\textsc{Forrelation} problem. The best-known bound prior to this work is an upper bound of $d^\ell\cdot \exp\pbra{{\ell+1\choose 2}}\cdot N^{(\ell-1)/2}$ for bounded degree-$d$ polynomials due to~\cite{IRR+21}. We see in this expression that the dependence on $d,\ell$ is of the form $d^\ell\cdot \exp\pbra{\ell^2/2}$ , which is quite large for $\ell\gtrapprox \sqrt{d}$, in contrast to our dependence, which is at most $d\cdot \sqrt{\ell!}$. In particular, in the regime where $d\ge \Omega(\sqrt{\log(N)})$ and $\ell$ is constant, our bound is an improvement. We are not aware if this dependence is tight and leave this for future work (see~\Cref{sec:future}).

%Combining this with the results of~\cite{BS21} (see~\Cref{thm:bs21}) implies that the maximum advantage that a $d$-query $\hBQP$ algorithm has in solving 3-forrelation is at most $O(d^6)/\sqrt{N}$.  

\input{plot_2}

We remark that variants of~\Cref{thm:main_theorem_dqck} and~\Cref{thm:main_theorem_bqp} hold even with classical pre-processing. The proof of this is quite simple and similar to ideas in~\cite{GSTW24} and is deferred to~\Cref{sec:hybrid}.

\paragraph*{Comparison to Prior Works.}

While Fourier growth has been extensively studied for classical algorithms, we are aware of only a few works that explicitly study the Fourier growth of quantum algorithms~\cite{AG23,GSTW24,IRR+21}. Among these,~\cite{IRR+21} and~\cite{GSTW24} are closely related to our work. As mentioned before,~\cite{IRR+21} establishes bounds on the Fourier growth of $\BQP$ algorithms that is slightly weaker than ours; furthermore, their bounds apply to all bounded low-degree polynomials and consequently cannot be used to distinguish between $\BQP,\hBQP$ and $\DQC{1}$. 
 
The work of~\cite{GSTW24} is especially closely related to our work. They study quantum algorithms with $k$ rounds of parallel queries and show that reducing the number of rounds even by one can cause a large blowup in the quantum query complexity. They achieve this by showing Fourier growth bounds for $k$-round quantum algorithms and leveraging the differences between the bounds for different $k$. Our work shares some conceptual similarities with their work, particularly in leveraging Fourier growth bounds to distinguish between quantum models, and also in using similar techniques for storing information about parities within matrix products. However, the models of quantum computation we consider are completely different. In~\cite{GSTW24}, the number of rounds is constrained, while the number of clean qubits is unlimited and the initial state is the all-zeroes state. In contrast, in our setting, the number of clean qubits is constrained and the initial state is forced to be highly mixed, while the number of rounds is allowed to be large. These differing constraints lead to fundamentally different behaviors. Consequently, our techniques diverge from theirs and we require distinct ideas and develop new techniques. 

It is worth emphasizing that the idea of using Fourier growth to distinguish between low-degree polynomials arising from different types of algorithms traces back to the landmark oracle separation of $\BQP$ and $\PH$~\cite{RT22}. The core challenge in that setting was that both models admit low-degree polynomial approximations, and Fourier growth was used precisely to tell these polynomials apart. 

\subsection{Applications}
We study the complexity of the Forrelation problem and its variants in the $\DQC{k}$ and $\hBQP$ models. Combining our Fourier growth bound (\Cref{thm:main_theorem_dqck}) with the results of~\cite{RT22,CHLT19} (see~\Cref{thm:rt22}) and the upper bounds on \forrtwo~from~\cite{Aar10,AA15}, we immediately obtain the following corollary. 
\begin{corollary}
For any $k\in \N$, the \forrtwo~problem on $2^k$-bit inputs can be solved by a $\DQC{k}$ algorithm  with success probability at least $2/3$ by making one quantum query, however, any $\DQC{k-t}$ algorithm that makes $d$ quantum queries has success probability at most $\frac{1}{2}+\tilde{O}\pbra{d^3}\cdot 2^{-t/2}$.  \label{cor:main_theorem_1}
\end{corollary}

In particular, any $\DQC{k-t}$ algorithm that succeeds with probability at least $2/3$ must make at least $\tilde{\Omega}(2^{t/6})$ queries. Setting $k=\log N$, we obtain the following corollary. 

\begin{corollary}\label{cor:main_theorem_1_1}
The \forrtwo~problem on $N$-bit inputs, which can be solved with $\log N$ clean qubits and one quantum query, requires $\tilde{\Omega}(N^{c/6})$ queries in the $\DQC{(1-c)\log N}$ model for all constants $c<1$. In particular, any $\DQC{1}$ algorithm for \forrtwo~must make $\tilde{\Omega}(N^{1/6})$ queries.
\end{corollary}

\medskip

We remark that~\Cref{cor:main_theorem_1_1} holds even if the algorithm is allowed to make $\tilde{\Omega}(N^{c/6})$  classical pre-processing queries in advance (using clean bits). We derive the following implications of~\Cref{cor:main_theorem_1_1}.

\paragraph*{A Hierarchy Theorem for $\DQC{k}$.}
In this work, we quantify the power that each additional clean qubit gives to quantum algorithms. It is not too difficult to show that any $\DQC{k}$ algorithm can be simulated by a $\DQC{k-t}$ algorithm without additional queries but with a loss of $2^{\Theta(t)}$ in the advantage (\Cref{claim:simulation}). \Cref{cor:main_theorem_1} shows that this is tight, up to a constant in the exponent. This shows that the number of clean qubits in a quantum algorithm cannot be efficiently reduced, even with a large amount of classical pre-processing on clean bits.  

\paragraph*{The First Oracle separation between $\hBQP$ and $\DQC{1}$.} We give the first oracle separation between $\hBQP$ and $\DQC{1}$, resolving a conjecture of~\cite{JM24}. In particular, Jacobs and Mehraban showed that \forrtwo~on $N$-bit inputs is solvable in the $\hBQP$ model with two quantum queries and conjectured that it requires $N^{\Omega(1)}$ queries in the $\DQC{1}$ model (see open question \#1 on page 8~\cite{JM24}). Our work (\Cref{cor:main_theorem_1}) proves this conjecture.   

\medskip

\paragraph*{A New Oracle separation between $\BQP$ and $\hBQP$.} 
Jacobs and Mehraban conjectured (see open question \#5 on page 8~\cite{JM24})  that \forrthree~is not in $\hBQP$  and our work (\Cref{cor:main_theorem_2}) resolves this. By combining our Fourier growth bound (\Cref{thm:main_theorem_hbqp}) with the results of~\cite{BS21} (\Cref{thm:bs21}), we immediately obtain the following corollary. 
\begin{corollary} The \forrthree~problem on $3N$-bit inputs, which can be solved by a $\BQP$ algorithm with two quantum queries, requires $\tilde{\Omega}(N^{\Omega(1)})$ queries in the $\hBQP$ model. \label{cor:main_theorem_2}
\end{corollary}

%We also remark that~\Cref{cor:main_theorem_2} holds even if the algorithm is allowed to make $\tilde{\Omega}(N^{1/12})$ classical pre-processing queries in advance. %While this is not the first oracle separation between $\BQP$ and $\hBQP$, we remark that there is an advantage to this separation over prior ones: our lower bound holds even if the algorithm is given query access to any sub-string of the input $x\in\{0,1\}^N$, while prior separations are broken with such oracles. See~\Cref{sec:appendix_sub_string} for more details.
We remark that while~\Cref{cor:main_theorem_2} is not the first oracle separation between $\BQP$ and $\hBQP$, there are some advantages to this new separation. The prior separation (in~\cite{JM24}) is as follows: given any oracle $O$ of length $2^n$ separating $\BQP$ and $\BPP$, we can embed it into a larger oracle $O'$ of length $2^{2n}$ whose first diagonal block is $O$ and all other diagonal entries are 1. It is not too difficult to show that $O'$ separates $\BQP$ and models like $\hBQP,\DQC{1},\NISQ$, and this was formalized in~\cite{NISQ,JM24}. The key intuition is that these intermediate models operate on highly mixed states and therefore assign only a vanishingly small weight to the relevant part of $O'$, namely $O$ itself. Embeddings of the Forrelation problem can thus separate $\BQP$ and $\hBQP$, but such separations are somewhat unsatisfactory since they do not establish the hardness of the original problem and apply uniformly to all models like $\hBQP,\DQC{1},\NISQ$. Our lower bound technique circumvents this limitation by directly proving a lower bound for the original \forrthree~problem in the $\hBQP$ model.

\subsection{Technical Highlight: Matrix Decomposition Lemma}

The main recurring technique in our paper is the use of a matrix decomposition lemma (see \Cref{lem:main_lemma}). This lemma offers a way to encode information about the indices involved in a matrix product and arises naturally in the context of quantum algorithms, as it allows us to encode information about the Fourier coefficients within a sequence of matrix products. We think it might be of independent interest. 

Firstly, we observe that the acceptance probability of quantum algorithms can be expressed as a product of matrices with bounded operator norms. To give some intuition, fix $i_{1},i_{d+1}\in[N]$. Consider a sequence of unitary matrices $U_1,\ldots, U_d$ and let $U[i\gap j]$ denote the $(i,j)^{\mathrm{th}}$-entry of $U$.  Consider a $\BQP$ algorithm that starts with the initial state $\ket{i_{1}}$, evolves it according to the unitary operators $U_1,\ldots,U_d$, interleaved with phase oracles $O_x$ and finally measures the qubits and accepts if the outcome is $\ket{i_{d+1}}$. The acceptance probability of this algorithm is given by $\abs{f(x)}^2$ where 
\begin{align*}f(x)&:= \bra{i_1} U_1\cdot O_x\cdot U_2\cdot O_x\cdots O_x\cdot U_d\ket{i_{d+1}}\\
&=\sum_{i_2,\ldots,i_d}\pbra{\prod_{t\in[d]} U_t[i_t\gap i_{t+1}]} \cdot \pbra{\prod_{t\in[d]\setminus \{1\}}x_{i_t}}\end{align*}  
More generally, by allowing the matrices $U_1,\ldots,U_d$ to be arbitrary matrices with spectral norm at most 1 and by adding workspace, we can produce a similar expression for $f(x)$ which \emph{equals} the acceptance probability of an arbitrary $\lfloor d/2\rfloor$-query $\BQP$ algorithm (see~\Cref{claim:acceptance_probability_bqp}). There are other expressions for capturing the acceptance probability of $\DQC{k}$ and $\hBQP$ algorithms using matrix products (see~\Cref{claim:dqck_acceptance_probability} and~\Cref{claim:acceptance_probability_hbqp}). Now that we have an expression for the acceptance probability, we need to compute the Fourier coefficients. Observe that for all $S\subseteq[N]$,
\[ \widehat{f}(S) =\sum_{i_{2},\ldots,i_{d}}\prod_{t\in[d]} U_t[i_t\gap i_{t+1}]\cdot \indi\sbra{S=\{i_{2}\}\oplus \ldots \oplus \{i_{d}\}}.\]

Our main idea is to try and encode information about the Fourier coefficients inside a product of matrices with bounded norms. The hope is that since $f(x)$ itself is a product of matrices with bounded norms, so are its Fourier coefficients.
To illuminate the main idea, say we wish to multiply the matrices $U_1,\ldots,U_d$ to get a matrix $U$ where
\[U[i_1\gap i_{d+1}] = \sum_{i_2,\ldots,i_d}\prod_{t\in[d]} U_t[i_t\gap i_{t+1}] \] 
but additionally, we wish to retain information about the symmetric difference of the intermediate indices $\{i_2\},\ldots,\{i_d\}$ until the very end. More formally, we wish to design a matrix $\tilU$ whose rows are indexed by $i_1$ and columns by $i_{d+1}S_{d+1}$ such that 
\[\tilU[i_1\gap i_{d+1}S_{d+1}] = \sum_{i_{2},\ldots,i_{d}}\prod_{t\in[d]} U_t[i_t\gap i_{t+1}]\cdot \indi\sbra{S_{d+1}=\{i_{2}\}\oplus \ldots \oplus \{i_{d}\}}.\] 
Here, the indicator function ensures that for each $S_{d+1}$, the corresponding entry of the final matrix only involves contributions from indices that satisfy the parity condition with respect to $S_{d+1}$. The reason we want to do this is clear; the entry $\tilU[i_{1}\gap i_{d+1}S]$ precisely equals the Fourier coefficient $\hat{f}(S)$. Thus, by reading off the entries of matrix $\tilU$ restricted to rows corresponding to $i_1$ and columns corresponding to $i_{d+1}$, we would obtain the list of all Fourier coefficients. The challenge lies in constructing such a matrix $\tilU$ with bounded norms and this is precisely achieved by~\Cref{lem:main_lemma}. It embeds the required combinatorial information about the indices within a matrix product while maintaining control over the norms of $\tilU$. We also show an improved matrix decomposition lemma (\Cref{lem:main_lemma_2}) that allows slightly more complex predicates of the indices being summed over -- in particular, we allow the imposition of inequality constraints over indices being summed over.

We remark that~\cite{GSTW24} implicitly proves another matrix decomposition lemma with a few key differences -- their bounds are for algorithms with a small fixed number of rounds but can handle parallel queries, and they only require bounds on the spectral norms of the underlying matrices. In our work, to handle $\DQC{k}$ and $\hBQP$ algorithms that can have a large number of rounds, we need a different kind of matrix decomposition and crucially, we require bounds on the Frobenius norms of the matrices in the decomposition, as well as the ability to impose inequality constraints over indices being summed over. This part is fundamentally new and requires additional work to prove.

\subsection{Proof Sketch}

 In general, proving Fourier growth bounds is quite challenging and technically involved. A major challenge arises from the need to incorporate the signs $\alpha_S\in[-1,1]$ into the matrix product given by the matrix decomposition lemma, and also from the need to sum over all sets $S$ of size $\ell$. Introducing the signs in a naive fashion often blows up the operator norms of the underlying matrices, making it difficult to maintain control over the Fourier growth. The heart of our proof involves techniques to incorporate these signs while keeping the operator norms bounded. This step turns out to be especially challenging for $\hBQP$ algorithms and we are unable to handle arbitrary signs $\alpha_S$. However, we are able to successfully encode the signs that arise from the \forrthree~problem.

\label{sec:proof_dqck_overview}

\begin{figure}
\centering
\mbox{ 
\Qcircuit @C=1em @R=.7em {
\lstick{} & &  &  \qw & \multigate{1}{V_d}  & \multigate{1}{O_x} & \multigate{1}{V_{d-1}}   & \multigate{1}{O_x}  &    \qw &  \cds{1}{\cdots\cdots}  & \multigate{1}{V_1} & \multigate{1}{O_x} & \qw \\
\lstick{}  & &  &  \qw  & \ghost{V_d}  & \ghost{O_x} & \ghost{V_{d-1}}   & \ghost{O_x} & \qw & \cds{1}{\cdots\cdots} & \ghost{V_1} & \ghost{O_x} &  \qw  \inputgroupv{1}{2}{.8em}{.8em}{n} \\ 
\lstick{}  & \ket{0}   &  &  \gate{H} & \ctrl{-1} \qw & \ctrl{-1} \qw  & \ctrl{-1}\qw  &\ctrl{-1} \qw  & \qw  &\qw & \ctrl{-1} \qw   &\ctrl{-1} \qw & \gate{H}  &\meter & & \text{output}\\ \\
}
}
\caption{A simple example of a $d$-query $\DQC{1}$ algorithm. The initial state on the first $n$ qubits is maximally mixed.}
\label{fig:proof_overview_dqc1}
\end{figure}

In this section, we present the simplest part of our proof: using the matrix decomposition lemma (\Cref{lem:main_lemma}) to establish Fourier growth bounds for $\DQC{1}$ algorithms. We will make some simplifications: we only focus on level $\ell=2$; we will assume that there is no restriction $\rho$ on the inputs; and we will only consider algorithms with one clean qubit of a special form in~\Cref{fig:proof_overview_dqc1}. These simplifications are only for the proof sketch and still give enough intuition for the general case. 

Firstly, it is not too difficult to derive an expression for acceptance probability of the algorithm in~\Cref{fig:proof_overview_dqc1}. 
%The initial state of the noisy qubits is a uniform mixture of $\ket{i}\bra{i}$ for $i\in[N]$. With probability $1/N$ this is a pure state $\ket{i}$, and conditioned on this initial state, the final state of the algorithm just before applying the last Hadamard gate is a pure state given by \[ \frac{1}{\sqrt{2}} (O_x\cdot U_d\cdots O_x\cdot U_1)\ket{i}\ket{1}+\frac{1}{\sqrt{2}}\ket{i}\ket{0}\] Thus, measuring the final qubit in the Hadamard basis results in a random bit that is 1 with probability $\tfrac{1}{2}- \tfrac{1}{2} \bra{i} O_x\cdot U_d\cdots O_x\cdot U_1\ket{i}.$ Since the initial state on the noisy qubits is maximally mixed, the overall acceptance probability of the algorithm is given by $\tfrac{1}{2}-\tfrac{1}{2}f(x)$ where \[ f(x):=\frac{1}{N}\sum_{i\in[N]}  \tfrac{1}{2} \bra{i} O_x\cdot U_d\cdot \cdots O_x\cdot U_1\ket{i}=\tfrac{1}{N}\Tr\pbra{O_x\cdot U_d\cdot \cdots O_x\cdot U_1}.\]
This is given by $\frac{1}{2}+\frac{1}{2}f(x)$ where 
\begin{align}\label{eq:proof_overview_1}\begin{split}
f(x)&:=\tfrac{1}{N}\Tr\pbra{O_x\cdot V_1\cdot O_x\cdot V_2\cdots O_x\cdot V_d}\\
&=\frac{1}{N}\sum_{i_1,\ldots,i_{d}\in[N]} \pbra{\prod_{t\in[d]}V_t[i_t\gap i_{t+1}]}\cdot \pbra{\prod_{t\in[d]}x_{i_t}}\end{split}
\end{align}
where $V_1\ldots,V_d$ are the $N\times N$ unitary matrices applied by the algorithm and we use the convention that $i_{d+1}=i_1$. One can derive a similar expression for the acceptance probability of an arbitrary $\DQC{k}$ algorithm (see~\Cref{claim:dqck_acceptance_probability} for more details). Let us now compute the Fourier coefficients of the acceptance probability, which equals (up to a factor of $1/2$) the Fourier coefficients of $f(x)$, which are easy to read off of~\Cref{eq:proof_overview_1}. For any $S\subseteq[N]$, the $S$-th Fourier coefficient of $f$ is given by
\begin{equation}\label{eq:proof_overview_2} \widehat{f}(S)= \frac{1}{N}\sum_{i_1,\ldots,i_d\in[N]}\pbra{\prod_{t\in[d]}V_t[i_t\gap i_{t+1}]}\cdot \indi\sbra{\{i_1\}\oplus \ldots\oplus \{i_d\}=S }.\end{equation}

The quantity we wish to bound is the level-2 Fourier growth of $f$, i.e., $L_{1,2}(f)=\max_\alpha L_{1,2}^\alpha(f)$, where \begin{equation}\label{eq:proof_overview_2.5}L_{1,2}^\alpha(f)\triangleq \sum_{|S|=2}\alpha_S\cdot \widehat{f}(S)\end{equation}
for signs $\alpha_S\in[-1,1]$ for $S\subseteq [N]$ of size 2. Fix any such signs $\alpha$. Substituting the expression for Fourier coefficients $\hat{f}(S)$ (\Cref{eq:proof_overview_2}) in the expression for $L_{1,2}^\alpha(f)$ (\Cref{eq:proof_overview_2.5}), we see that our goal is to upper bound
\begin{align}\label{eq:proof_overview_3} L_{1,2}^\alpha(f)=\sum_{|S|=2}\alpha_S\cdot \frac{1}{N}\sum_{i_1,\ldots,i_d\in[N]}\pbra{\prod_{t\in[d]}V_t[i_t\gap i_{t+1}]} \cdot \indi\sbra{\{i_1\}\oplus \ldots\oplus \{i_d\}=S }.\end{align}  

\paragraph*{Decomposing $L_{1,2}^\alpha$ into a few terms.} First, we will group the terms in~\Cref{eq:proof_overview_3} into a few terms. We will express $L_{1,2}(f)^\alpha$ as a sum over pairs $(t_1,t_2)$ such that $t_1\neq t_2\in[d]$ of a quantity $\Delta_{t_1,t_2}^\alpha$. We describe this below.

Observe that for a term to contribute to~\Cref{eq:proof_overview_3}, the symmetric difference of $i_1,\ldots,i_d$ has size 2. In this case, there must exist a pair of indices $t_1<t_2\in[d]$ such that $i_{t_1}$ and $i_{t_2}$ are distinct and the symmetric difference of the rest of the $i_t$ is the empty set. More precisely, if $\{i_1\}\oplus \ldots\oplus \{i_d\}=S$ for a set $S$ of size 2, then
\[\exists t_1<t_2\in[d]\text{ such that }\{i_{t_1},i_{t_2}\}=S\text{ and }\oplus_{t\in[d]\setminus\{t_1,t_2\}}\{i_t\}=\emptyset.\] 
Conversely, any such $t_1,t_2\in[d]$ and $i_1,\ldots,i_d$ satisfying the above equation defines a unique $S=\{i_{t_1},i_{t_2}\}$.
For any pair of indices $t_1<t_2\in[d]$, let $\Delta_{t_1,t_2}^\alpha$ be the contribution of the corresponding terms to $L_{1,2}^\alpha(f)$, i.e.,
 \begin{align} \Delta_{t_1,t_2}^\alpha:=\frac{1}{N}\sum_{i_{t_1}\neq i_{t_2}\in[N]} \alpha_{\{i_{t_1},i_{t_2}\}}\cdot \sum_{\substack{i_{t_1+1},\ldots,i_{t_2-1}\in[N]\\ i_{t_2+1},\ldots,i_{t_1-1}\in[N]}}   \pbra{\prod_{t\in[d]}V_t[i_t\gap i_{t+1}]} \cdot \indi\sbra{\bigoplus_{t\in[d]\setminus\{t_1,t_2\}}\{i_t\}=\emptyset } . \label{eq:dqc1_delta}\end{align}
Ideally, we would like to say that  $L_{1,2}^\alpha(f)=\sum_{t_1<t_2\in[d]}\Delta^\alpha_{t_1,t_2}$ and to bound the latter quantity, we observe that there are $O(d^2)$ choices of $t_1< t_2\in[d]$ and for any such choice, we will show in the second step that $\Delta_{t_1,t_2}^\alpha\le 1$, obtaining $L_{1,2}^\alpha(f)\le O(d^2)$ as desired. Unfortunately, it is not true that $L_{1,2}^\alpha(f)=\sum_{t_1<t_2\in[d]}\Delta^\alpha_{t_1,t_2}$, as $t_1,t_2$ are not uniquely defined for a given set of indices $i_1,\ldots,i_d$.\footnote{We thank Francisco Escudero Gutierrez
and Miquel Saucedo Cuesta for
pointing this out.} In order to address this, we take two different approaches for $\DQC{1}$ and $\hBQP$ algorithms. In the former case, we let $t_1$ denote the first time an element of $S$ is seen, and let $t_2$ denote the next time an element of $S$ is seen. As a result, we will have to modify the definition of $\Delta_{t_1,t_2}^\alpha$ as follows.
\begin{align*} \Delta_{t_1,t_2}^\alpha&:=\frac{1}{N}\sum_{i_{t_1}\neq i_{t_2}\in[N]} \alpha_{\{i_{t_1},i_{t_2}\}}\cdot \sum_{\substack{i_{t_1+1},\ldots,i_{t_2-1}\in[N]\\ i_{t_2+1},\ldots,i_{t_1-1}\in[N]}}   \pbra{\prod_{t\in[d]}V_t[i_t\gap i_{t+1}]} \cdot \indi\sbra{\bigoplus_{t\in[d]\setminus\{t_1,t_2\}}\{i_t\}=\emptyset } \\
&\cdot \prod_{t\in[1,t_1)}\indi\sbra{i_t\neq i_{t_1}}\cdot \prod_{t\in[1,t_2)} \indi\sbra{i_t\neq i_{t_2}}.\end{align*}
For the $\hBQP$ bound on the other hand, we take an alternate approach where we don't uniquely identify $t_1,t_2$, instead, we indeed sum over all $\Delta_{t_1,t_2}$, each occurring with a coefficient that precisely cancels out to give an expression for the $L_{1,\ell}^{\alpha}(f)$. For the rest of the proof overview, we ignore these subtleties and imagine for now that $L_{1,2}^\alpha(f)=\sum_{t_1<t_2\in[d]}\Delta^{\alpha}_{t_1,t_2}$ for $\Delta^{\alpha}_{t_1,t_2}$ as in~\Cref{eq:dqc1_delta} and proceed.

\begin{comment}
  Fix $\ell=2$ and any $\alpha_S\in [-1,1]$ for each $S\in{\binom{[\tilN]}{2}}$. From \Cref{eq:dqck_proof_1} and \Cref{claim:dqck_fourier_coefficients}, we see that our goal is to upper bound 
\begin{align} \label{eq:dqc1_6}
    L_{1,\ell}^\alpha(f|_\rho)= \sum_{\substack{S\subseteq[\tilN]\\|S|=2}} (NW)^{-1}\sum_{\iwk{1},\ldots,\iwk{d'}\in[M]} \pbra{\prod_{t\in[d']} V^\rho_t[\iwk{t}\gap\iwk{t+1}]} \cdot \indi\sbra{\bigoplus_{\substack{t\in[d']\\\text{with }i_t\le \tilN}} \{i_t\}=S}\cdot \alpha_S

\end{align}

Define
\begin{align}\label{eq:dqc1_delta}\begin{split}
\Delta_{t_1,t_2}&:=\sum_{\substack{\iwk{t_1},\iwk{t_2}\in[M]\\\text{with }i_{t_1}\neq i_{t_2}\in[\tilN]}} \sum_{\substack{\iwk{t_1+1},\ldots,\iwk{t_2-1}\in[M]\\\iwk{t_2+1},\ldots,\iwk{t_1-1}\in[M] }}  \pbra{\prod_{t\in [t_1,t_2)} V^\rho_t[\iwk{t}\gap\iwk{t+1}]} \cdot  \pbra{\prod_{t\in [t_2,t_1)} V^\rho_t[\iwk{t}\gap\iwk{t+1}]}\\
	&\cdot  \indi\sbra{\bigoplus_{\substack{t\in[d']\setminus\{t_1,t_2\}\\\text{with }i_t\le \tilN}} \{i_t\}=\emptyset}\cdot \alpha_{\{i_{t_1},i_{t_2}\}}.\end{split}\end{align}
By combining the above with~\Cref{eq:dqc1_6} and the fact that $d'=O(d)$, we see that
\[L_{1,2}^\alpha(f|_\rho)=  (NW)^{-1}  \sum_{t_1< t_2\in[d']} \Delta_{t_1,t_2} \le O(d^2)\cdot (NW)^{-1}\cdot \max_{t_1<t_2\in[d']} \Delta_{t_1,t_2}.\]
We will now show that $\Delta_{t_1,t_2}\le NW\cdot \sqrt{K}$ for any $t_1\neq t_2\in[d']$ and this would show that $L_{1,2}^\alpha(f|_\rho)\le O(d^2)\cdot \sqrt{K}$ as desired in~\Cref{thm:main_theorem_dqck}. 

\end{comment}

\paragraph*{Showing that $\Delta_{t_1,t_2}^\alpha\le 1$.} This is where we will use the matrix decomposition lemma (\Cref{lem:main_lemma}). We will group the terms $t\in[d]$ into circular intervals $[t_1,t_2)$ and $[t_2,t_1)$\footnote{We arrange $1,\ldots,d$ in a clock-wise circle and define the intervals clock-wise. For instance, the interval $[d-2,2]$ refers to the set $\{d-2,d-1,d,1,2\}$. The intervals $(t_1,t_2)$ and $(t_2,t_1)$ are well-defined but would be empty if $t_2=t_1\pm 1$ modulo $d$. In each of these cases, it is understood that the summation over $i_{t_1+1},\ldots,i_{t_2-1}$ and $i_{t_2+1},\ldots,i_{t_1-1}$ respectively is to be ignored.}. We will apply the matrix decomposition lemma on $V_{t_1},\ldots,V_{t_2-1}$ to remember the symmetric difference of $\{i_t\}$ for $t\in(t_1,t_2)$ and similarly on the matrices $V_{t_2},\ldots,V_{t_1-1}$ to remember the symmetric difference of $\{i_t\}$ for $t\in(t_2,t_1)$ and then enforce equality between these sets. More precisely, apply~\Cref{lem:main_lemma} (with $T=\emptyset$) on the matrices $V_{t_1}^\rho,\ldots,V_{t_2-1}^\rho$ to obtain $\tilV_{[t_1,t_2)}$ and to $V_{t_2}^\rho,\ldots,V_{t_1-1}^\rho$ backwards to obtain $\tilV_{[t_2,t_1)}$ such that for all $i_{t_1},i_{t_2}\in[N],S_{t_2}\subseteq[N]$, 
\begin{equation}\label{eq:dqc1_7} \tilV_{[t_1,t_2)}[i_{t_1}\gap i_{t_2}S_{t_2}]=\sum_{i_t\in[N]\text{ for }t\in(t_1,t_2)}  \pbra{ \prod_{ t\in [t_1,t_2)}V_t[i_{t}\gap i_{t+1}]}\cdot \indi\sbra{\bigoplus_{t\in (t_1,t_2)} \{i_t\}=S_{t_2}}, \end{equation}
\begin{equation}\label{eq:dqc1_100} \tilV_{[t_2,t_1)} [i_{t_1}\gap i_{t_2}S_{t_2}]=\sum_{ i_t\in[N]\text{ for }t\in(t_2,t_1)}  \pbra{ \prod_{ t\in [t_2,t_1)}V_t[i_{t}\gap i_{t+1}]}\cdot \indi\sbra{\bigoplus_{t\in (t_2,t_1)}\{i_t\}=S_{t_2}}. \end{equation}
Substituting~\Cref{eq:dqc1_7,eq:dqc1_100}  in~\Cref{eq:dqc1_delta}, we see that
\begin{align*} \Delta_{t_1,t_2}^\alpha&\triangleq \frac{1}{N}\sum_{i_{t_1}\neq i_{t_2}\in[N]}\alpha_{\{i_{t_1},i_{t_2}\}} \sum_{S_{t_2}\subseteq[N]} \tilV_{[t_1,t_2)}[i_{t_1}\gap i_{t_2}S_{t_2}] \cdot \tilV_{[t_2,t_1)}[i_{t_1}\gap i_{t_2}S_{t_2}]   \\
&\le \frac{1}{N}\sum_{i_{t_1}\neq i_{t_2}\in[N]}\sum_{S_{t_2}\subseteq[N]} \abs{\tilV_{[t_1,t_2)}[i_{t_1}\gap i_{t_2}S_{t_2}]} \cdot \abs{\tilV_{[t_2,t_1)}[i_{t_1}\gap i_{t_2}S_{t_2}]}  \tag{since $\alpha_{\{i_{t_1},i_{t_2}\}}\in[-1,1]$} \\ 
&\le  \frac{1}{N}\cdot \vabs{\tilV_{[t_1,t_2)}}_\frob\cdot \vabs{\tilV_{[t_2,t_1)}}_\frob \tag{\Cref{fact:frob_product}}\end{align*}
Firstly, observe that \[\max\pbra{\|\tilV_{[t_1,t_2)}\|_\frob,\|\tilV_{[t_2,t_1)} \|_\frob} \le \sqrt{N}.\] 
This is because both matrices have operator norm at most one and either have at most $N$ rows or $N$ columns. This implies that $\Delta_{t_1,t_2}^\alpha \le N^{-1}\cdot N\le 1$. This completes the proof sketch.

\paragraph*{}
We now describe some of the additional ideas involved in generalizing this proof.

\paragraph*{Generalizing to higher levels.} We provide a general bootstrapping argument that proves higher level Fourier growth bounds assuming bounds for lower levels. This is applicable to any restriction-closed family of boolean functions and allows us to prove improved bounds on the Fourier growth of $\BQP$ algorithms, as well as higher-level and level-6 bounds for $\DQC{k}$ and $\hBQP$ algorithms respectively.

\paragraph*{$\hBQP$ algorithms.} It is not too hard to show that the expression for the acceptance probability of a $d$-query $\hBQP$ algorithm is quite similar to~\Cref{eq:proof_overview_1}, except, there are $2d+2$ matrices $V_1,\ldots,V_{2d+2}$, and more importantly, there is an extra term of the form $F_{i_1,i_{d+1}}\in\{0,1\}$ inside the summation, which corresponds to the post-processing of the measurement outcomes of the initial and final states. (See~\Cref{eq:hbqp_1} and~\Cref{claim:acceptance_probability_hbqp} for a formal expression.) This additional term $F_{i_1,i_{d+1}}$ is challenging to incorporate while keeping the norms bounded. As a result, proving bounds for $\hBQP$ algorithms turns out to be more technically involved. We need to use an improved matrix decomposition lemma (\Cref{lem:main_lemma_4}).

Furthermore, we are only able to prove level-3 and level-6 Fourier growth bounds for a particular family of signs as in~\Cref{def:alpha}. The reason why the signs $\alpha(\gamma)$ and $\beta(\gamma)$ in~\Cref{def:alpha} are easier to deal with than general signs, is that once we fix $i_2$, $\alpha(\gamma)_{i_1,i_2,i_3}$ becomes a product of three terms, the first depending only on $i_1$, the second on $i_3$ and the third on $\gamma$ in a product fashion. Similarly, once we fix $i_2,i_5$, then $\beta(\gamma)_{i_1,\ldots,i_6}$ becomes a product of five terms, the first depending only on $i_1$, the second on $i_4$, the third on $i_3$, the fourth on $i_6$, and the fifth on $\gamma$ in a product fashion. These kinds of signs that are products across the indices are much easier to handle than general families of signs and often exhibit a Fourier growth that is much smaller than the Fourier growth for arbitrary signs\footnote{Indeed, for general bounded degree-$d$ polynomials, the level-$\ell$ Fourier growth with arbitrary signs can be as large as $N^{\Omega(\ell)}$, whereas for signs that are a product across the indices, the Fourier growth is at most $d^{O(\ell)}$~\cite{IRR+21}.}. We then show that summing over the $i_2$, or over the $i_2,i_5$ doesn't blow up the Fourier growth by much. (See~\Cref{sec:hbqp} for more details.)

\subsection{Outlook \& Future Directions}
\label{sec:future}

Broadly, our results suggest that Fourier growth provides a powerful analytic lens to separate models of quantum computation. Several natural next steps emerge in this direction and we highlight some open questions in this section.

\begin{enumerate} 
\item \textbf{Fourier Growth of $\NISQ$.}
Researchers have attempted to model $\NISQ$ (noisy intermediate scale quantum) algorithms through the lens of query complexity, in the hopes of understanding the computational power of near-term quantum devices~\cite{NISQ,chia2024}. %In this model, there is depolarizing noise acting on all qubits initially and after every round of the algorithm. 
There has been recent interest in using \forrtwo~to show quantum advantages in near-term experiments~\cite{Geo25,Shu25} and this prompts the natural question, can we solve \forrtwo~in $\NISQ$? If not, can we prove bounds on the Fourier growth of $\NISQ$? %There are interesting connections between this model and the ones we study; for instance, a variant of $\NISQ$ where the noise acts only on the initial state contains $\hBQP$~\cite{dale}, leading us to believe that the techniques in our paper might shed light on this question.

\item \textbf{The Power of $\DQC{1}$.} Where does $\DQC{1}$ fit within the landscape of classical complexity, and in particular, is it contained in $\PH$? 
%Despite numerous tasks solvable in $\DQC{1}$ that appear to be classically hard, a classical oracle separation between $\DQC{1}$ and $\PH$ remains open, to the best of our knowledge. 
The differences between the Fourier growth of $\DQC{1}$ and $\PH$ are quite stark, %~\Cref{sec:tightness_DQC} shows that the level-$\ell$ growth of $\DQC{1}$ algorithms can be as large as $N^{(\ell-2)}$, whereas for depth-$d$ size-$s$ $\mathsf{AC}0$ circuits (the query complexity analogue of $\PH$), it is at most $\poly(\log^d (s))^{\ell}$~\cite{Tal20,sherstov2023optimal}. However, 
but it is not clear how to leverage this into an oracle separation, as existing approaches rely on the Forrelation problem, which is hard for $\DQC{1}$. Developing new techniques here would not only clarify the power of $\DQC{1}$, but also expand the toolkit for proving lower bounds on classical computation. %A key problem that exploits this difference in Fourier growth is the Forrelation problem, but it is not solvable in $\DQC{1}$, and since all known oracle separations between $\BQP$ and $\PH$ are based on variants of the Forrelation problem, we would need to come up with a fundamentally new candidate to show such a separation. We remark that an oracle separation between $\DQC{1}$ and $\BPP$ is known~\cite{shepherd}.

\item \textbf{The Power of $\IQP$.} Another intriguing intermediate model is $\IQP$, whose power derives from its ability to perform Fourier sampling. How does this model compare to $\DQC{1}$ and $\hBQP$? Understanding the relationship between these models would help chart the intermediate landscape between $\BPP$ and $\BQP$ and reveal the relative power of various quantum capabilities like Fourier sampling and trace estimation. It was shown by~\cite{JM24} that $\IQP$ can be simulated by $\hBQP$ and they conjectured that this containment is strict. Is \forrtwo~solvable in $\IQP$~\footnote{This question has since been resolved affirmatively by~\cite{IQPinFORR}.} and if not, can we prove Fourier growth bounds?
 
\item \textbf{Tight Bounds on the Fourier Growth of Quantum Algorithms.}
Finally, many of our upper bounds on the Fourier growth are not known to be tight. Are the dependencies on $d$ and $\ell$ tight in~\Cref{thm:main_theorem_dqck,thm:main_theorem_hbqp,thm:main_theorem_bqp}? What is the Fourier growth of $\hBQP$ with respect to arbitrary families of signs? Tight bounds on Fourier growth could provide a precise handle for quantum computational power, and help map the landscape between classical, intermediate, and fully quantum models.
\end{enumerate}

\subsection{Organization.}

\Cref{sec:prelims} consists of preliminaries, where we formally describe the various models of computation and state the results we need from prior works on Forrelation. In~\Cref{sec:decomposition}, we describe and prove some of the basic matrix decomposition lemmas (\Cref{lem:main_lemma,lem:main_lemma_2}). We prove our Fourier growth bounds for $\DQC{k}$ in~\Cref{sec:proof_dqck} (proof of \Cref{thm:main_theorem_dqck}), $\hBQP$ in~\Cref{sec:hbqp} (proof of \Cref{thm:main_theorem_hbqp}) and $\BQP$ in~\Cref{sec:bootstrapping} (proof of \Cref{thm:main_theorem_bqp}).

%\paragraph*{Remark.} It is an intriguing open question whether we can achieve an oracle separation between $\BQP$ and $\hBQP$ where the  $\BQP$ algorithm uses only one query. The naive approach of using level-two Fourier growth bounds does not work, since the Forrelation problem is in $\hBQP$ and hence, $\hBQP$ algorithms can indeed achieve large level-2 Fourier growth. 

\section{Preliminaries \& Notation}

\paragraph*{Restrictions.} For a restriction $\rho\in\{-1,1,\star\}^N$ and a vector $x\in\{-1,1\}^N$, the $i$-th coordinate of the restricted vector $\rho(x)\in\{-1,1\}^N$ is $\rho_i$ if $\rho_i\in\{-1,1\}$ and $x_i$ if $\rho_i=\star$ for $i\in[N]$. For a boolean function $f:\{-1,1\}^N\to \mathbb{R}$, and a restriction $\rho\in\{-1,1,\star\}^N,$ we use $f|_\rho$ to denote the restricted function which maps $x$ to $f(\rho(x))$ for $x\in\{-1,1\}^N$.

\label{sec:prelims}
\paragraph*{Sets.} For $x\in \R^N$ and $S\subseteq[N]$, we use $\chi_S(x)$ or $x_S$ to denote $\prod_{i\in S}x_i$. For indices $i_1,\ldots,i_k\in [N]$, we use $\{i_1\}\oplus\ldots\oplus \{i_k\}$ to denote the symmetric difference $\oplus_{t\in [k]}\{i_t\}$ and similarly $S_1\oplus S_2$ denotes the symmetric difference of the sets $S_1$ and $S_2$.

We will often use uppercase letters to denote $2$ to the power of lowercase letters, in particular, $N=2^n,W=2^w,K=2^k$ and $M=2^m$.

\paragraph*{Circular Intervals.} For $i,j\in[n]$, we use $[i,j]$ to denote the clockwise sequence of points from $i$ to $j$ when $1,\ldots,n$ are arranged clock-wise in a circle. For example, $[n,2]=\{n,1,2\}$ and $[1,3]=\{1,2,3\}$. We use $(,]$ and $[,)$ and $(,)$ to denote half-open or open intervals.

 \paragraph*{Vectors and Inner Products.} We identify the space $\{0,1\}^n$ with $[N]$ under the natural correspondence $(a_1,\ldots,a_n)\to 1+\sum_i a_i2^{i-1}$. We also identify $\{0,1\}^n$ with $\{-1,1\}^n$ under the correspondence that maps $0$ to $1$ and $1$ to $-1$. For $u,v\in[N]$, we use $\langle u,v\rangle_2:=\sum_{i\in[n]} u_i v_i \mod 2$ to denote the inner product over $\F_2$ under the aforementioned correspondence. For $u\in\C^N$ and $U\in \C^{N\times N}$, we use $u^\dagger, U^\dagger$ to denote the conjugate-transpose. For complex vectors $u,v\in\C^N$, we use $\langle u\mid v \rangle$, $v^\dagger u$, and $\langle u, v \rangle$ to denote $\sum_i u_i\overline{v}_i$, the complex inner product.

\paragraph*{Matrices.} We use $\bI$ to denote the identity matrix, where the dimensions are clear from context. 
We will often encounter matrices whose rows and columns are indexed by $(i,w)$ for $i\in[N],w\in[W]$, or by $(i,w,k)$ for $i\in[N],w\in[W],k\in[K]$. For ease of notation, we use $\iw{}$ as a shorthand for $(i,w)$ or $(i,w,k)$, where the distinction will be clear from the context. For $\iw{t},\iw{t+1}\in [M]$, we use either $U_t[\iw{t}\gap \iw{t+1}]$ or  $U_t[\iw{t}, \iw{t+1}]$ to denote the $(\iw{t},\iw{t+1})$-the entry of $U_t$. For matrices $U_1,\ldots,U_d$, we use $U_{[t_1,t_2]}$ to denote the product $\prod_{t\in[t_1,t_2]}U_t=U_{t_1}\cdots U_{t_2}$ of the matrices in the circular interval $[t_1,t_2]$ in clockwise order. We define $U_{[t_1,t_2)},U_{(t_1,t_2]},U_{(t_1,t_2)}$ analogously.

\begin{definition}[Hadamard Matrix] \label{def:hadamard}
For $N=2^n$, the Hadamard matrix $H_N$ is defined to be
\[ H_N=\frac{1}{\sqrt{N}} \begin{bmatrix}1 & 1 \\ 1  & -1\end{bmatrix}^{\otimes n}.\]
\end{definition}

\paragraph*{Matrix Norms \& Inequalities.} Let $\|\cdot \|_\op$ and $\|\cdot \|_\frob$ denote the spectral and Frobenius norm, or equivalently, the Schatten-$\infty$ and Schatten-2 norms. The following basic fact follows from Holder's Inequality for Schatten norms.

\begin{fact} \label{fact:frob_op}
Let $A,B,C$ be rectangular matrices with $A=BC$. Then, 
\[\|A\|_\frob \le \min(\|B\|_\op\cdot \|C\|_\frob,\|B\|_\frob\cdot \|C\|_\op)\]  \end{fact} 
\begin{proof}
 $\|A\|_\frob^2\triangleq \Tr\pbra{ AA^\dagger}=\Tr(BC\cdot C^\dagger B^\dagger)=\Tr(B^\dagger B\cdot CC^\dagger)\le \|C\|_\op^2 \cdot \Tr(B^\dagger B\cdot \id)=\|C\|_\op^2\cdot \|B\|_\frob^2$. Applying the same argument on $A^{\dagger}$ gives the other inequality.
\end{proof}

%\begin{corollary}\label{fact:frob_op_new} Let $A,B$ be rectangular matrices. Then, $\Tr(A^\dagger AB^\dagger B)\le \Tr(A^\dagger A)$ if $B^\dagger B\preceq \id$.\end{corollary}

The Cauchy-Schwarz inequality implies the following fact.
\begin{fact} \label{fact:frob_product}
For rectangular matrices $A,B$, and any subset $T$ of indices, we have 
\[ \sum_{(i,j)\in T}|A[i\gap j]|\cdot |B[i\gap j]| \le \|A\|_\frob\cdot \|B\|_\frob.\]  \end{fact} 

%We use the following basic facts about the spectral norms of matrices.

%\begin{fact} \label{fact:op_submatrix} For any submatrix $B$ of $A$, we have $\|B\|_\op\le \|A\|_\op$. \end{fact}

%\begin{fact} \label{fact:op_block}For any block diagonal matrix $A$ consisting of blocks $A_1,\ldots,A_t,$ we have $\|A\|_\op \le \max_{i\in[t]}\|A_t\|_\op$.\end{fact}

\subsection{Fourier Growth}
Recall the definition of the Fourier growth as in~\Cref{def:signed_fourier_growth} and~\Cref{def:fourier_growth}. For a family of functions $\cF$, we use $L_{1,\ell}(\cF)$ to denote $\max_{f\in\cF}L_{1,\ell}(f)$.

\begin{comment}
    
For level $\ell\in \N$, and signs $\alpha_S\in[-1,1]$ for $S\subseteq[N]$ with $|S|=\ell$, define the $\alpha$-signed level-$\ell$ Fourier growth of $f$, denoted by $L_{1,\ell}^\alpha(f)$ as \[ L_{1,\ell}^\alpha(f):=\sum_{\substack{S\subseteq[N]\\|S|=\ell}} \alpha_S\cdot \widehat{f}(S),\]

Define the level-$\ell$ Fourier growth of $f$, denoted by $L_{1,\ell}(f)$ to be the largest possible value of $L_{1,\ell}^\alpha(f)$ over signs $\alpha\in[-1,1]^{\binom{n}{\ell}}$.
\[ L_{1,\ell}^\alpha(f):=\max_{\alpha\in[-1,1]^{\binom{N}{\ell}}} L_{1,\ell}^{\alpha}(f)\]
Observe that the maximum in the R.H.S. is achieved when $\alpha_S=\widehat{f}(S)/|\widehat{f}(S)|$ for all $|S|=\ell$. Thus, one can equivalently define $L_{1,\ell}(f)$ to be 
\[ L_{1,\ell}(f)=\sum_{S\subseteq[N],|S|=\ell}\abs{\widehat{f}(S)},\]
as in~\Cref{def:fourier_growth}.
\end{comment}

\paragraph*{Lower Bounds for Forrelation from Fourier Growth.}

The results of~\cite{RT22,CHLT19} imply that to show lower bounds on the \forrtwo~problem, it suffices to prove Fourier growth bounds for level 2. 
\begin{theorem}[\cite{RT22,CHLT19}]
     Let $\cF$ be any family of $2N$-variate boolean functions closed under restrictions. Then, the maximum advantage with which $\cF$ solves \forrtwo~is at most
\[ O\pbra{\frac{L_{1,2}(\cF)}{\sqrt{N}}}. \]\label{thm:rt22}
\end{theorem}

The results of~\cite{BS21} imply that to show lower bounds on the \forrthree~problem, it suffices to prove signed-Fourier growth bounds for level 3 and 6, for the following family of signs. 
 
\medskip
\noindent

\begin{definition}
Partition $[3N]$ into $A:=[N],B:=(N,2N],C:=(2N,3N]$ and fix any ordering $<$ of the elements in $A,B,C$. There is a natural correspondence $B\leftrightarrow[N]$ given by $b\leftrightarrow b-N$ for all $b\in B$ and a similar correspondence $C\leftrightarrow[N]$ given by $c\leftrightarrow c-2N$ for all $c\in C$. Let $\gamma\in[-1,1]^{3N},\gamma'\in[-1,1]^{3N}$. Define $\alpha(\gamma)\in[-1,1]^{\binom{3N}{3}}$ and  $\beta(\gamma,\gamma')\in[-1,1]^{\binom{3N}{6}}$ as follows. Let $\tilH\in\{-1,1\}^{N\times N}$ be the matrix whose $(i,j)$-th entry is $(-1)^{\abra{i,j}_2}=\mathrm{sign}(H_N[i\gap j])$ for $i,j\in[N]$.
 For $i_1,i_2,i_3\in[3N]$, let
\[ \alpha(\gamma)_{i_1,i_2,i_3}:=\begin{cases} \tilH(i_2, i_1)\cdot \tilH(i_2, i_3)\cdot\pbra{ \prod_{t\in[3]}\gamma_{i_t}}  &\text{if } i_1\in A,i_2\in B,i_3\in C \\  
0 &\text{otherwise.}
\end{cases} \]
For $i_1,\ldots,i_6\in[3N]$, let
\[ \beta(\gamma,\gamma')_{i_1,\ldots,i_6}:=\begin{cases} \alpha(\gamma)_{i_1,i_2,i_3}  \cdot   \alpha(\gamma')_{i_4,i_5,i_6}  &\text{if } i_1<i_4\in A,i_2< i_5\in B,i_3< i_6\in C \\  
0 &\text{otherwise.}
\end{cases}\]
\label{def:alpha}
\end{definition} 

The following theorem is implicit in~\cite{BS21}.\footnote{In particular, see equation (5.7) and the equation above in~\cite{BS21} for the level-3 contribution and equation (5.13) and the preceding paragraph for the level-6 contribution.}

\begin{theorem}[Implicit in~\cite{BS21}] Let $\cF$ be any family of $3N$-variate boolean functions that is closed under restrictions. Let $\gamma,\gamma'\in[-1,1]^{3N}$ and $\alpha(\gamma)\in[-1,1]^{\binom{3N}{3}},\beta(\gamma,\gamma')\in[-1,1]^{\binom{3N}{6}}$ be as in~\Cref{def:alpha}. Then, the maximum advantage with which $\cF$ solves \forrthree~is at most
\[ \max_{\gamma,\gamma'\in[-1,1]^{3N}}O\pbra{\frac{L_{1,3}^{\alpha(\gamma)}(\cF)}{N} + \frac{L_{1,6}^{\beta(\gamma,\gamma')}(\cF)}{N^2}}. \]\label{thm:bs21}
\end{theorem}

\subsection{Quantum Query Complexity}

In the setting of quantum query complexity, the input is accessed by an oracle. This oracle is typically an operator $\widetilde{O}_x$ for $x\in\{0,1\}^N$ which maps $\ket{b}\ket{i}\to \ket{b\oplus x_i}\ket{i}$ for $b\in\{0,1\},i\in[N]$.
One can alternatively define an oracle $O_x$ for $x\in\{-1,1\}^N$ which maps $\ket{b}\ket{i}$ to itself if $b=0$ and to $\ket{b}\ket{i}x_i$ if $b=1$ and $i\in[N]$. It is not too difficult to show that these two definitions are equivalent, up to a Hadamard gate on the first qubit. We will work with the oracle $O_x$ and later introduce some additional simplifications.

The most general model of a quantum query algorithm is the $\BQP$ model defined below. For the following definition, we interpret $n$ as the number of qubits on which the oracle acts and $w$ as the number of qubits of extra workspace. As mentioned before, we use $\bI$ to denote the identity matrix, where the dimension is implicit.

\begin{definition}[$\BQP$ Algorithm with $d$ Queries]\label{def:bqp_algorithm}
    Let $n,w\in \N$, $N=2^n,W=2^w$ and $M=N W$. A $\BQP$ algorithm acts on $n+w$ qubits initialized to $\ket{0,\ldots,0}$. Let $U_0,U_1,\ldots,U_{d}\in \C^{M\times M}$ be $M\times M$ unitary matrices. The algorithm applies the unitary operators $U_0,\ldots,U_{d}$ interleaved with the oracle $O_x\otimes \bI$ and measures all the qubits at the end to obtain an outcome $\iw{d+1}$. The algorithm accepts iff $\iw{d+1}\in \cF$ where $\final \subseteq [M]$ is a subset. (See~\Cref{fig:bqp} for a depiction.)
\end{definition}

The following claim expresses the acceptance probability of a $d$-query $\BQP$ algorithm and is not too difficult to prove. 
\begin{claim} The acceptance probability of a $d$-query $\BQP$ algorithm can be expressed as 
\[  f(x):= \abra{ v\gap   O\cdot V_1\cdots   V_{2d-1}\cdot  O\gap v}\]
where $O=O_x\otimes \bI$, $V_1,\ldots,V_{2d-1}\in \mathbb{C}^{M\times M}$ are matrices with $\|V_t\|_\op \le 1$ for all $t\in[2d-1]$ and $v\in \C^M$ is a unit vector.\label{claim:acceptance_probability_bqp}
\end{claim}

\begin{figure}
\centering
\mbox{ 
\Qcircuit @C=1em @R=.7em {
\lstick{} & \ket{0} &  &  \qw & \multigate{3}{U_0}  & \multigate{1}{O_x} & \multigate{3}{U_1}   & \multigate{1}{O_x}  &    \qw &  \cds{1}{\cdots\cdots}  & \multigate{3}{U_{d-1}} & \multigate{1}{O_x} & \multigate{3}{U_{d}} &\meter \\
\lstick{} & \ket{0} &  &  \qw  & \ghost{U_1}  & \ghost{O_x} & \ghost{U_2}   & \ghost{O_x} & \qw & \cds{1}{\cdots\cdots} & \ghost{U_{d-1}} & \ghost{O_x} &\ghost{U_{d}} &\meter  \inputgroupv{1}{2}{.8em}{1.8em}{n\hspace{2em}} \\ 
\lstick{} & \ket{0} &  &  \qw & \ghost{U_1} &\qw  &  \ghost{U_2}  & \qw &\qw & \qw & \ghost{U_{d-1}} &\qw  &\ghost{U_{d}} &  \meter    \\ 
\lstick{} & \ket{0} &  &  \qw & \ghost{U_1} &\qw  &  \ghost{U_2}  & \qw &\qw & \qw & \ghost{U_{d-1}} &\qw  &\ghost{U_{d}} &  \meter  \inputgroupv{3}{4}{.8em}{1.8em}{w\hspace{2em}}  \\ 
}
}
\caption{A $d$-query $\BQP$ algorithm. }
\label{fig:bqp}
\end{figure}

In the following sections, we will define $\DQC{k}$ and $\hBQP$ algorithms.

\subsection{$\DQC{k}$ algorithms} 

We interpret $n$ as the number of qubits on which the oracle acts, $k$ as the number of clean qubits, and $w$ as the number of qubits of extra workspace.

\begin{definition}[$\DQC{k}$ Algorithm with $d$ Queries] \label{def:dqck}
Let $n,w,k\in \N$ and $N=2^n,W=2^w,K=2^k$ and $M=NWK$. A $\DQC{k}$ algorithm acts on $k$ clean qubits initialized to the $\ket{0\ldots 0}$ state and $n+w$ maximally noisy qubits which consist of $n$ qubits on which the oracle acts and $w$ qubits of workspace. Let $U_1,\ldots, U_{d+1}\in\C^{M\times M}$ be $M\times M$ unitary matrices. Let $\start=[NW]\times\{1\}$ be the set of all possible starting basis states of the algorithm and $\final\subseteq[NWK]$ be the subset of final basis states that is accepted by the algorithm. The algorithm starts with a uniformly random basis state sampled from $\start$, applies the unitary operators $U_1,\ldots,U_{d+1}$, interleaved with the oracle $O_x\otimes \bI$, measures all the qubits at the end and accepts if the outcome is in $\final$. (See~\Cref{fig:dqc1} for a depiction.)
\end{definition}

\paragraph*{Remark.} In our model, the oracles are not allowed to directly act on the clean qubits, nevertheless, we can effectively implement this type of operation by swapping the clean qubits with the noisy qubits, applying the oracle on those noisy qubits and swapping them back with the clean qubits. While this transformation does require the use of $k$ extra (potentially noisy) qubits to do the swap operation, our formalism has the advantage that we can talk about oracle separations where $k$, the number of clean qubits is significantly smaller than $n$, where the length of the input is $2^n$. This is important, since when $k\gg n$, many problems become solvable with a few quantum queries with $O(k)$ clean qubits.

 We will now provide an expression for the acceptance probability of a $\DQC{k}$ algorithm, which we will prove in the appendix (\Cref{sec:appendix_acceptance_probability}). As mentioned before, estimating the trace of a unitary matrix described by a quantum circuit is known to be complete for the class $\DQC{1}$~\cite{KL98} and a similar statement is true in query complexity as well.

\begin{claim}\label{claim:dqck_acceptance_probability}
The acceptance probability of a $d$-query $\DQC{k}$ algorithm can be expressed as 
\begin{align*} f(x)&=(NW)^{-1}\cdot \Tr\pbra{ O \cdot V_1 \cdots   O\cdot  V_{2d}}\end{align*}
where $O=O_x\otimes \bI$, $V_1,\ldots,V_{2d}\in\C^{M\times M}$ satisfy $\|V_t\|_\op\le 1$ for $t\in[2d]$, furthermore, $\|V_1\|_\frob \le \sqrt{NW}$. 
\end{claim}

\subsection{$\hBQP$ algorithms}

We interpret $n$ as the number of qubits on which the oracle acts and $w$ as the number of qubits of extra workspace.

\begin{definition}[$\hBQP$ Algorithm with $d$ Queries]\label{def:hbqp_algorithm}
    Let $n,w\in \N$, $N=2^n,W=2^w$ and $M=N W$. A $\hBQP$ algorithm acts on $n+w$ qubits initialized to $\ket{\iw{1}}$ for a uniformly random $\iw{1}\sim [M]$. The algorithm does not have knowledge of $\iw{1}$.  Let $U_1,\ldots,U_{d+1}\in \C^{M\times M}$ be $M\times M$ unitary matrices. The algorithm applies the unitary operators $U_1,\ldots,U_{d+1}$ interleaved with the oracle $O_x\otimes \bI$ and measures all the qubits at the end to obtain an outcome $\iw{d+2}$. Finally, the algorithm then learns $\iw{1}$. The algorithm accepts iff $(\iw{1},\iw{d+2})\in \cF$ where $\final\subseteq[M]\times[M]$ is a subset. (See~\Cref{fig:hbqp} for a depiction.)
\end{definition}

We provide an expression for the acceptance probability of a $d$-query $\hBQP$ algorithm, which is proved in~\Cref{sec:appendix_acceptance_probability}. 

\begin{claim} The acceptance probability of a $d$-query $\hBQP$ algorithm can be expressed as 
\[f(x):=M^{-1}\sum_{\iw{1},\iw{d+2}\in[M]}F_{\iw{1},\iw{d+2}}\cdot \bra{\iwk{1}} U_1^\dagger\cdot O \cdots O\cdot  U_{d+1}^\dagger  \ket{\iw{d+2}}\bra{\iwk{d+2}}  U_{d+1}\cdot O\cdots O\cdot  U_1 \ket{\iwk{1}}  \]
where $O=O_x\otimes \bI$,  and $U_1,\ldots,U_{d+1}\in \mathbb{C}^{M\times M}$ are matrices with $\|U_t\|_\op \le 1$ for all $t\in[d+1]$.
\label{claim:acceptance_probability_hbqp}
\end{claim}

\paragraph*{Some Remarks.}
 
\begin{itemize}
\item While our way of defining $\DQC{k}$ and $\hBQP$ doesn't clearly subsume $\BPP$, there is a simple way to fix this. We can define variants of these models where the algorithm is allowed to make up to $d$ classical pre-processing queries on clean bits, and based on the query outcomes, choose a $d$-query quantum algorithm to run. When defined this way, these models immediately subsume $\BPP$, since we can implement any $\BPP$ algorithm in the pre-processing part. Interestingly, many of the results in our paper, especially the lower bounds hold even for algorithms with a large amount of classical pre-processing. See~\Cref{sec:hybrid} for more details.
\item Unlike~\cite{GSTW24}, our model does not allow parallel queries. This is without loss of generality, as our model has unrestricted depth and we can simulate $k$ parallel queries by $k$ adaptive queries. If we allow parallel queries but limit the depth, we suspect that it might lead improved Fourier growth bounds in terms of the depth of the algorithm, but we leave this to future work. 
\item In the rest of this paper, we will work with the oracle $O'_x$ which maps $\ket{i}$ to $\ket{i}x_i$ for all $i\in[N]$ where $x$ is of length $N$. Note that the aforementioned oracle $O_x$ is the controlled version of $O'_x$ and generally offers more functionality than $O'_x$. However, in all our proofs, it suffices to work with the oracle $O'_x$ since we allow restrictions $\rho\in\{-1,1,\star\}^{N}$ to act on our input. In particular, if we consider $O'_x$ for bit-strings of length $2N$ and apply the restriction which fixes the first $N$ coordinates to 1, we obtain the oracle $O_x$ on bit-strings of length $N$ as desired. Since all our Fourier growth bounds work even under restrictions of the input, it suffices to work with oracles of the form $O'_x$ and all our Fourier growth bounds will carry over to oracles of the form $O_x$ if $N$ is replaced by $2N$. Henceforth, we will refer to the oracle $O'_x$ as $O_x$ and work with this oracle.
\end{itemize}

\section{Main Technical Tool: Matrix Decomposition Lemma}
\label{sec:decomposition}

The following matrix decomposition lemma is a  recurring tool in this paper. It allows us to encode information about the indices in a matrix multiplication by embedding them inside a larger matrix multiplication. In this lemma, we have matrices $U_1,\ldots,U_{d}$ where the rows and columns of $U_t$ are indexed by $\iw{t}$ and $\iw{t+1}$ respectively. Here, $\iw{}$ is a shorthand for either $(i,w,k)$ or $(i,w)$ where $i\in[N]$ corresponds to indices we want to remember information about and $w\in[W],k\in[K]$ corresponds to auxiliary workspace indices. The set $T$ corresponds to the complement of matrices whose index information we want to retain, i.e., we don't care about the matrices in $T$. The set $L$ indicates that we do not store parity information for indices $i_t$ with $i_t\notin L$ and the set $S_{d+1}$ corresponds to the information aggregated after multiplying the matrices.

\begin{lemma} Let $U_1,\ldots,U_{d}$ be $M\times M$ matrices with $\|U_t\|_\op \le 1$ for $t\in[d]$ and let $T\subseteq[d]$ and $L\subseteq [N]$. Define a matrix $\tilU$ such that for all $\iwk{1},\iwk{d+1}\in[M],S_{d+1}\subseteq[N],$
\begin{align*} \tilU[\iwk{1}\gap \iwks{d+1}] 
&=\sum_{\iwk{2},\ldots,\iwk{d}\in[M]} \pbra{\prod_{t\in[1,d]} U_{t}[\iwk{t}\gap\iwk{t+1}]}
\cdot \indi\sbra{S_{d+1}=\bigoplus_{\substack{t\in [1,d]\setminus T\\i_t\in L}}\{i_{t}\} }.\end{align*}
Then, $\|\tilU\|_\frob \le \min_{t\in[d]}\|U_t\|_\frob$.
\label{lem:main_lemma}
\end{lemma}

\begin{proof}[Proof of~\Cref{lem:main_lemma}]
Define a diagonal matrix $O_x$ whose $i$-th entry is a variable $x_i$ if $i\in L$ and is 1 otherwise. Define matrices $V_i(x)=U_i$ if $i\in T$ and $V_i(x)=O_x\cdot U_i$ otherwise. Observe that $\|V_i\|_\op \le 1$ for all $i\in[d]$. Consider the matrix-valued function
$M(x):=V_1(x)\cdots V_d(x)$. Observe that
\[ M(x)[\iw{1}\gap \iw{d+1}] 
=\sum_{\iwk{2},\ldots,\iwk{d}\in[M]} \pbra{\prod_{t\in[1,d]} U_{t}[\iwk{t}\gap\iwk{t+1}]}
\cdot \prod_{\substack{t\in[1,d]\setminus T\\ i_t\in L}}x_{i_t}\]
Hence, for all $S\subseteq[N]$, we have $\widehat{M}(S)[\iw{1}\gap\iw{d+1}]= \tilU[\iw{1}\gap\iws{d+1}]$. This implies that \begin{align*}\|\tilU\|_\frob^2 &= \sum_{S\subseteq[N]} \sum_{\iw{1},\iw{d+1}\in[M]}\widehat{M}(S)[\iw{1}\gap\iw{d+1}]^2\\
&=\sum_{S\subseteq[N]} \|\widehat{M}(S)\|_\frob^2 =\E_{x}\sbra{\|M(x)\|_\frob^2} \tag{by Parseval's.}
\end{align*}
Finally, let $t\in[d]$ and $x\in \{\pm 1\}^N$. Observe that $M(x)=V_{<t}(x)\cdot V_t(x)\cdot V_{>t}(x)$ where $V_{<t}(x):=\prod_{t'<t}V_t(x)$ and $V_{>t}(x):=\prod_{t'>t}V_t(x)$ have operator norm at most 1, hence $\|M(x)\|_\frob \le \|V_t(x)\|_\frob =\|U_t\|_\frob.$ This completes the proof.
\end{proof}

We will need a slightly more complicated version of this lemma that incorporates inequality constraints over the indices being summed over. 

Define functions $g_1,g_2:[N]\times [N]\to \mathbb{R}$ as follows. $g_1(i,j)=\indi\sbra{i\neq j}$ and $g_2(i,j)=(-1)^{\indi\sbra{i=j}}$.

\begin{lemma} Let $U_1,\ldots,U_{d}$ be $M\times M$ matrices with $\|U_t\|_\op \le 1$ for $t\in[d]$ and let $t_1\in [d]$. For all $\gamma\in[2]$, define a matrix $U_\gamma$ such that for all $\iwk{1},\iwk{d+1}\in[M]$
\begin{align*} U_{\gamma}[\iwk{1}\gap \iwk{d+1}] 
&=\sum_{\iwk{2},\ldots,\iwk{d}\in[M]} \pbra{\prod_{t\in[1,d]} U_{t}[\iwk{t}\gap\iwk{t+1}]} \cdot \prod_{t\in (t_1,d]}g_\gamma(i_t,i_{t_1}).\end{align*}
Then, for all $\gamma\in[2]$, $\|U_\gamma\|_\frob \le   O(d)\cdot \min_{t\in[d]}  \|U_t\|_\frob$ and $\|U_\gamma\|_\op \le O(d)$. Furthermore, $\|U_\gamma\|_\frob \le \min_{t\in [t_1,d]}(O(t)\cdot \|U_{t}\|_\frob,\sqrt{M})$.
\label{lem:main_lemma_2}
\end{lemma} 

\begin{corollary} Under the same hypothesis as~\Cref{lem:main_lemma_2}, let $T\subseteq[d]$ and $L\subseteq [N]$. Define a matrix $\tilU$ such that for all $\iwk{1},\iwk{d+1}\in[M],S_{d+1}\subseteq[N],$
\begin{align*} \tilU_\gamma[\iwk{1}\gap \iwks{d+1}] 
&=\sum_{\iwk{2},\ldots,\iwk{d}\in[M]} \pbra{\prod_{t\in[1,d]} U_{t}[\iwk{t}\gap\iwk{t+1}]}
\cdot \indi\sbra{S_{d+1}=\bigoplus_{\substack{t\in [1,d]\setminus T\\i_t\in L}}\{i_{t}\} }\cdot \prod_{t\in (t_1,d]}g_\gamma(i_t,i_{t_1}).\end{align*}
Then, $\|\tilU_\gamma \|_\frob \le  O(d)\cdot \min_t  \|U_t\|_\frob$ and $\|\tilU_\gamma \|_\frob\le O(\sqrt{M})$.
\label{corollary:main_lemma_2}
\end{corollary} 
\begin{figure}[t!]
\centering

\tikzset{every picture/.style={line width=0.75pt}} %set default line width to 0.75pt        

\begin{tikzpicture}[x=0.75pt,y=0.75pt,yscale=-1,xscale=1]
%uncomment if require: \path (0,134); %set diagram left start at 0, and has height of 134

%Shape: Grid [id:dp6911456838322174] 
\draw  [draw opacity=0] (80,61) -- (279.14,61) -- (279.14,94.57) -- (80,94.57) -- cycle ; \draw   (80,61) -- (80,94.57)(113,61) -- (113,94.57)(146,61) -- (146,94.57)(179,61) -- (179,94.57)(212,61) -- (212,94.57)(245,61) -- (245,94.57)(278,61) -- (278,94.57) ; \draw   (80,61) -- (279.14,61)(80,94) -- (279.14,94) ; \draw    ;
%Straight Lines [id:da9567092149281238] 
\draw    (146.14,34.57) -- (146.01,59) ;
\draw [shift={(146,61)}, rotate = 270.29] [color={rgb, 255:red, 0; green, 0; blue, 0 }  ][line width=0.75]    (10.93,-3.29) .. controls (6.95,-1.4) and (3.31,-0.3) .. (0,0) .. controls (3.31,0.3) and (6.95,1.4) .. (10.93,3.29)   ;
%Straight Lines [id:da9415459163557188] 
\draw    (179.14,34.57) -- (179.01,59) ;
\draw [shift={(179,61)}, rotate = 270.29] [color={rgb, 255:red, 0; green, 0; blue, 0 }  ][line width=0.75]    (10.93,-3.29) .. controls (6.95,-1.4) and (3.31,-0.3) .. (0,0) .. controls (3.31,0.3) and (6.95,1.4) .. (10.93,3.29)   ;
%Straight Lines [id:da9503977281945598] 
\draw    (212.14,34.57) -- (212.01,59) ;
\draw [shift={(212,61)}, rotate = 270.29] [color={rgb, 255:red, 0; green, 0; blue, 0 }  ][line width=0.75]    (10.93,-3.29) .. controls (6.95,-1.4) and (3.31,-0.3) .. (0,0) .. controls (3.31,0.3) and (6.95,1.4) .. (10.93,3.29)   ;
%Straight Lines [id:da49593077500415084] 
\draw    (245.14,34.57) -- (245.01,59) ;
\draw [shift={(245,61)}, rotate = 270.29] [color={rgb, 255:red, 0; green, 0; blue, 0 }  ][line width=0.75]    (10.93,-3.29) .. controls (6.95,-1.4) and (3.31,-0.3) .. (0,0) .. controls (3.31,0.3) and (6.95,1.4) .. (10.93,3.29)   ;
%Straight Lines [id:da5603240543223443] 
\draw    (113.14,34.57) -- (113.01,59) ;
\draw [shift={(113,61)}, rotate = 270.29] [color={rgb, 255:red, 0; green, 0; blue, 0 }  ][line width=0.75]    (10.93,-3.29) .. controls (6.95,-1.4) and (3.31,-0.3) .. (0,0) .. controls (3.31,0.3) and (6.95,1.4) .. (10.93,3.29)   ;
%Shape: Grid [id:dp5671259860462929] 
\draw  [draw opacity=0] (321,62) -- (520.14,62) -- (520.14,95.57) -- (321,95.57) -- cycle ; \draw   (321,62) -- (321,95.57)(354,62) -- (354,95.57)(387,62) -- (387,95.57)(420,62) -- (420,95.57)(453,62) -- (453,95.57)(486,62) -- (486,95.57)(519,62) -- (519,95.57) ; \draw   (321,62) -- (520.14,62)(321,95) -- (520.14,95) ; \draw    ;
%Straight Lines [id:da5135223254896308] 
\draw [color={rgb, 255:red, 255; green, 0; blue, 0 }  ,draw opacity=1 ][line width=3.75]    (354,62) -- (354,95) ;
%Straight Lines [id:da2030685873201291] 
\draw [color={rgb, 255:red, 0; green, 93; blue, 235 }  ,draw opacity=1 ][line width=3.75]    (387,62) -- (387,95) ;
%Straight Lines [id:da2041572535115017] 
\draw [color={rgb, 255:red, 0; green, 93; blue, 235 }  ,draw opacity=1 ][line width=3.75]    (420,62) -- (420,95) ;
%Straight Lines [id:da39343556810188685] 
\draw [color={rgb, 255:red, 0; green, 93; blue, 235 }  ,draw opacity=1 ][line width=3.75]    (453,62) -- (453,95) ;
%Straight Lines [id:da30425429780905555] 
\draw [color={rgb, 255:red, 0; green, 93; blue, 235 }  ,draw opacity=1 ][line width=3.75]    (486,62) -- (486,95) ;

% Text Node
\draw (89,69.4) node [anchor=north west][inner sep=0.75pt]  [font=\small]  {$U_{1}$};
% Text Node
\draw (120,69.4) node [anchor=north west][inner sep=0.75pt]  [font=\small]  {$U_{2}$};
% Text Node
\draw (252,70.4) node [anchor=north west][inner sep=0.75pt]  [font=\small]  {$U_{t}$};
% Text Node
\draw (213,69.4) node [anchor=north west][inner sep=0.75pt]  [font=\small]  {$U_{t-1}$};
% Text Node
\draw (104,9.4) node [anchor=north west][inner sep=0.75pt]    {$ \begin{array}{l}
\mathnormal{Q^{( i_{t_{1}})}}\\
\end{array}$};
% Text Node
\draw (108,96.4) node [anchor=north west][inner sep=0.75pt]    {$t_{1}$};
% Text Node
\draw (142,11.4) node [anchor=north west][inner sep=0.75pt]    {$ \begin{array}{l}
P^{( i_{t_{1}})}\\
\end{array}$};
% Text Node
\draw (176,11.4) node [anchor=north west][inner sep=0.75pt]    {$ \begin{array}{l}
P^{( i_{t_{1}})}\\
\end{array}$};
% Text Node
\draw (211,11.4) node [anchor=north west][inner sep=0.75pt]    {$ \begin{array}{l}
P^{( i_{t_{1}})}\\
\end{array}$};
% Text Node
\draw (245,12.4) node [anchor=north west][inner sep=0.75pt]    {$ \begin{array}{l}
P^{( i_{t_{1}})}\\
\end{array}$};
% Text Node
\draw (330,70.4) node [anchor=north west][inner sep=0.75pt]  [font=\small]  {$U_{1}$};
% Text Node
\draw (361,70.4) node [anchor=north west][inner sep=0.75pt]  [font=\small]  {$U_{2}$};
% Text Node
\draw (493,71.4) node [anchor=north west][inner sep=0.75pt]  [font=\small]  {$U_{t}$};
% Text Node
\draw (454,70.4) node [anchor=north west][inner sep=0.75pt]  [font=\small]  {$U_{t-1}$};
% Text Node
\draw (349,99.4) node [anchor=north west][inner sep=0.75pt]    {$t_{1}$};
% Text Node
\draw (289,69.4) node [anchor=north west][inner sep=0.75pt]    {$=$};
% Text Node
\draw (347,39.4) node [anchor=north west][inner sep=0.75pt]    {$ \begin{array}{l}
i_{t_{1}}\\
\end{array}$};
% Text Node
\draw (383,39.4) node [anchor=north west][inner sep=0.75pt]    {$ \begin{array}{l}
i_{t_{1}}\\
\end{array}$};
% Text Node
\draw (416,39.4) node [anchor=north west][inner sep=0.75pt]    {$ \begin{array}{l}
i_{t_{1}}\\
\end{array}$};
% Text Node
\draw (449,39.4) node [anchor=north west][inner sep=0.75pt]    {$ \begin{array}{l}
i_{t_{1}}\\
\end{array}$};
% Text Node
\draw (482,39.4) node [anchor=north west][inner sep=0.75pt]    {$ \begin{array}{l}
i_{t_{1}}\\
\end{array}$};

\end{tikzpicture}

\caption{We use red and blue lines indexed by $i_{t_1}$ to depict the insertion of $Q^{(i_{t_1})}$ and $P^{(i_{t_1})}$ respectively, where we omit the superscript $\gamma$.}
\label{fig_1:definition}

\end{figure}

\begin{proof}[Proof of~\Cref{corollary:main_lemma_2} from~\Cref{lem:main_lemma_2}]
Define a diagonal matrix $O_x$ whose $i$-th entry is a variable $x_i$ if $i\in L$ and is 1 otherwise, and define matrices $V_t=V_t(x)=U_t$ if $t\in T$ and $V_t=V_t(x)=(O_x\otimes \id)\cdot U_t$ otherwise. Apply~\Cref{lem:main_lemma_2} to $V_1(x),\ldots,V_d(x)$ to obtain a matrix $U_\gamma(x)$ with \[\|U_\gamma(x)\|_\frob \le O(d)\cdot \min_t\pbra{\|V_t(x)\|_\frob}=O(d)\cdot \min_t\pbra{  \|U_t\|_\frob}\]
such that for all $\iwk{1},\iwk{d+1}\in [M]$,
\begin{align*}
U_\gamma(x)[\iwk{1}\gap\iwk{d+1}]
&=\sum_{\iwk{2},\ldots,\iwk{d}\in[M]} \pbra{\prod_{t\in[1,d]} V_{t}(x)[\iwk{t}\gap\iwk{t+1}]}\cdot \prod_{t\in (t_1,d]}g_\gamma(i_t, i_1)\\
&=\sum_{\iwk{2},\ldots,\iwk{d}\in[M]} \pbra{\prod_{t\in[1,d]} U_{t}[\iwk{t}\gap\iwk{t+1}]}
\cdot \prod_{\substack{t\in[1,d]\setminus T\\ i_t\in L}}x_{i_t}\cdot \prod_{t\in (t_1,d]}g_\gamma(i_t, i_1).
\end{align*}
Observe that for all $S\subseteq[N]$, we have $\widehat{U}(S)[\iw{1}\gap\iw{d+1}]= \tilU[\iw{1}\gap\iw{d+1}S]$. By Parseval's, \begin{align*}\|\tilU\|_\frob^2= \sum_{S\subseteq[N]} \sum_{\iw{1},\iw{d+1}\in[M]}\abs{\widehat{U}(S)[\iw{1}\gap\iw{d+1}]}^2 \triangleq
\sum_{S\subseteq[N]} \|\widehat{U}(S)\|_\frob^2 =\E_{x}\sbra{\|U(x)\|_\frob^2}.
\end{align*}
This completes the proof.
\end{proof}
\begin{proof}[Proof of~\Cref{lem:main_lemma_2}] 

\input{fig_2_telescoping}

For $i\in[N]$, and $\gamma\in [2]$, define matrices $P^{(i,\gamma)}$ as follows. Define $M\times M$ matrices $P^{(i,1)},P^{(i,2)},Q^{(i)}$ as follows.\begin{align*} P^{(i,1)}&:=(\id-\ket{i}\bra{i})\otimes \id\\ P^{(i,2)}&:=(\id - 2\ket{i}\bra{i})\otimes \id\\Q^{(i)}&:=\ket{i}\bra{i}\otimes \id.\end{align*}
Observe that $\|Q^{(i)}\|_\op,\|P^{(i,1)}\|,\|P^{(i,2)}\|_\op\le 1$, furthermore, 
\begin{equation}\label{eq:Q_i_difference} Q^{(i)}=\id-P^{(i,1)}=\frac{1}{2}\pbra{\id -P^{(i,2)}}. \end{equation}
Secondly, the $Q^{(i)}$ are orthogonal, i.e., for all $i,i'\in[N]$ we have
\begin{equation} \label{eq:Q_i_ortho} Q^{(i')}\cdot Q^{(i)\dagger}=Q^{(i)\dagger}\cdot Q^{(i')}=0\quad\text{ if }i\neq i'.\end{equation}
For $\gamma\in[2]$, define $\tilP_t^{i,\gamma)}= U_t\cdot P^{(i,\gamma)}$ for all $i\in[N],t\in[t_1,d-1]$ and let $\tilP_d^{(i,\gamma)}=U_d$. Observe that for all $\gamma\in[2]$,
\[U_\gamma=U_{[1,t_1)}\cdot \sum_{i_{t_1}\in [N]} Q^{(i_{t_1})}\cdot \tilP_{t_1}^{(i_{t_1},\gamma)}\cdots \tilP_d^{(i_{t_1},\gamma)}\triangleq U_{[1,t_1)}\cdot\sum_{i_{t_1}\in [N]} Q^{(i_{t_1})}\cdot \tilP_{[t_1,d]}^{(i_{t_1},\gamma)}.\]
An example of this is depicted in~\Cref{fig_1:definition}. Fix any $t\in[t_1,d]$ and $\gamma\in[2]$. We have
\begin{align}
\|U_\gamma\|_\frob^2&\triangleq \vabs{U_{[1,t_1)}\cdot\sum_{i_{t_1}\in[N]} Q^{(i_{t_1})} \cdot \tilP_{[t_1,d]}^{(i_{t_1},\gamma)}}_\frob^2\le \vabs{\sum_{i_{t_1}\in[N]} Q^{(i_{t_1})} \cdot \tilP_{[t_1,d]}^{(i_{t_1},\gamma)}}_\frob^2 \tag{since $\|U_{[1,t_1)}\|_\op\le 1$}\\
&=\sum_{i_{t_1}\in[N]}  \vabs{Q^{(i_{t_1})} \cdot \tilP_{[t_1,d]}^{(i_{t_1},\gamma)}}_\frob^2 \tag{by~\Cref{eq:Q_i_ortho}}\\
&\le \sum_{i_{t_1}\in [N]} \vabs{ Q^{(i_{t_1})} \cdot \tilP_{[t_1,t)}^{(i_{t_1},\gamma)}\cdot  U_t}_\frob^2\tag{by~\Cref{fact:frob_op} and since other terms have $\|\cdot\|_\op\le 1$}\\
&= \vabs{\sum_{i_{t_1}\in [N]} Q^{(i_{t_1})} \cdot \tilP_{[t_1,t)}^{(i_{t_1})}\cdot  U_t}_\frob^2 \tag{by~\Cref{eq:Q_i_ortho}}\\
&\le \|U_t\|_\frob^2\cdot \vabs{\sum_{i_{t_1}\in [N]} Q^{(i_{t_1})} \cdot \tilP_{[t_1,t)}^{(i_{t_1})}}_\op^2\tag{by~\Cref{fact:frob_op}}
\end{align}

Setting $t=t_1$ and recalling that $\sum_{i_{t_1}\in[N]}Q^{(i_{t_1})}=\id$, we get $\|U_\gamma\|_\frob\le \|U_{t_1}\|_\frob \le \sqrt{M}$. We now show that $\vabs{\sum_{i_{t_1}\in[N]}Q^{(i_{t_1})}\cdot \tilP^{i_{t_1}}_{[t_1,t)}}_\op\le O(t)$. By the telescoping sum, we have
\begin{align*} Q^{(i_1)}\cdot U_{[t_1,t)}-Q^{(i_1)}\cdot \tilP^{(i_1,\gamma)}_{[t_1,t)}&=\sum_{t'\in [t_1,t-1]}Q^{(i_1)}\cdot \tilP^{(i_1,\gamma)}_{[t_1,t')}\cdot \pbra{U_{t'}-\tilP^{(i_1,\gamma)}_{t'}}\cdot U_{(t',t)}\\
&= \sum_{t'\in [t_1,t-1]}Q^{(i_1)}\cdot \tilP^{(i_1,\gamma)}_{[t_1,t')}\cdot U_{t'}\cdot \pbra{\id-\tilP^{(i_1,\gamma)}}\cdot U_{(t',t)}\tag{by definition}\\
&=\sum_{t'\in [t_1,t-1]}Q^{(i_1)}\cdot \tilP^{(i_1,\gamma)}_{[t_1,t')}\cdot U_{t'}\cdot \gamma\cdot Q^{(i_1)}\cdot U_{(t',t)}.\tag{by~\Cref{eq:Q_i_difference}}
\end{align*}
An example of this telescoping sum is depicted in~\Cref{fig_2:telescoping}. Summing over $i_1\in [N]$ and using the fact that  $\sum_{i_1\in[N]}Q^{(i_1)}=\id$, we see that 
\[ \sum_{i_1\in [N]}Q^{(i_1)}\cdot \tilP^{(i_1)}_{[t_1,t)}=U_{[t_1,t)} -\gamma\sum_{t'\in [t_1,t-1]}\sum_{i_1\in [N]} Q^{(i_1)}\cdot \tilP^{(i_1,\gamma)}_{[1,t')}\cdot U_{t'}\cdot Q^{(i_1)}\cdot U_{(t',t)}.\]
Hence, using the fact that $\|U_{t'}\|_\op \le 1$ for all $t'\in[1,d]$, we get
\[ \vabs{\sum_{i_1\in [N]}Q^{(i_1)}\cdot \tilP^{(i_1,\gamma)}_{[t_1,t)}}_\op \le 1 + (t-1)\cdot \gamma\cdot \max_{t'\in [t_1,t-1]}\vabs{\sum_{i_1\in [N]}Q^{(i_1)}\cdot \tilP^{(i_1,\gamma)}_{[t_1,t')}\cdot U_{t'}\cdot Q^{(i_1)}}_\op.\]
We observe that $\gamma\le 2$ for all $\gamma\in[2]$. Finally, we observe that each of the terms inside the maximum is at most 1, due to the following: 
\begin{align*}
&\pbra{\sum_{i_1\in [N]}Q^{(i_1)}\cdot \tilP^{(i_1,\gamma)}_{[t_1,t')}\cdot U_{t'}\cdot Q^{(i_1)}}\cdot \pbra{\sum_{i_1\in [N]}Q^{(i_1)}\cdot \tilP^{(i_1,\gamma)}_{[t_1,t')}\cdot U_{t'}\cdot Q^{(i_1)}}^\dagger \\ 
&=\sum_{i_1\in [N]}Q^{(i_1)}\cdot \tilP^{(i_1,\gamma)}_{[t_1,t')}\cdot U_{t'}\cdot Q^{(i_1)}\cdot Q^{(i_1)\dagger}\cdot U_{t'}^\dagger\cdot  \tilP^{(i_1,\gamma)\dagger}_{[t_1,t')}\cdot  Q^{(i_1)\dagger} \tag{by~\Cref{eq:Q_i_ortho}} \\
&\preceq \sum_{i_1\in [N]} Q^{(i_1)}\cdot Q^{(i_1)\dagger}=\id\tag{since $\vabs{\tilP^{(i_1,\gamma)}_{[1,t')}\cdot U_{t'}\cdot Q^{(i_1)}}_\op \le 1$}
\end{align*}
%This sequence of operations is depicted in~\Cref{fig_3:cancellation}. 
The desired operator norm is thus at most $O(t)$, and we have argued that $\|U_\gamma\|_\frob \le O(d)\cdot \min_{t\in[t_1,d]}\|U_t\|_\frob$. The same argument as above implies that
\begin{align*}\|U_\gamma\|_\op&\le \vabs{\sum_{i\in[N]}Q^{(i_1)}\cdot \tilP^{(i_1,\gamma)}_{[t_1,d]}}_\op \le 1+2d\cdot \max_{t'\in[t_1,d]}\vabs{\sum_{i_1\in[N]} Q^{(i_1)}\cdot \tilP^{(i_1,\gamma)}_{[t_1,t')}\cdot U_{t'}\cdot Q^{(i_1)} }_\op\le O(d).
\end{align*}
Finally, we see that 
\begin{align*}
\|U_\gamma\|_\frob^2&\triangleq \vabs{U_{[1,t_1)}\cdot\sum_{i_{t_1}\in[N]} Q^{(i_{t_1})} \cdot \tilP_{[t_1,d]}^{(i_{t_1},\gamma)}}_\frob^2\\
&\le \vabs{U_{[1,t_1)}}_\frob^2\cdot  \vabs{\sum_{i_{t_1}\in[N]} Q^{(i_{t_1})} \cdot \tilP_{[t_1,d]}^{(i_{t_1},\gamma)}}_\op^2 \tag{by \Cref{fact:frob_op}}\\
&\le \min_{t\in[1,t_1)} \|U_t\|_\frob^2 \cdot O(d^2).
\end{align*}
Altogether, we have $\|U_\gamma\|_\frob\le O(d)\cdot \min_{t\in[d]}\|U_t\|_\frob.$ This completes the proof. 
\end{proof}

\subsection{Matrix Decomposition Lemma for $\DQC{1}$}

\begin{lemma} Let $U_1,\ldots,U_{d}$ be $M\times M$ matrices with $\|U_t\|_\op \le 1$ for $t\in[d]$. Let $t^*\in (1,d]$. Define a matrix $U$ such that for all $\iwk{1},\iwk{d+1}\in[M],$
\begin{align*} \tilU[\iwk{1}\gap \iwk{d+1}] 
&=\sum_{\iwk{2},\ldots,\iwk{d}\in[M] } \pbra{\prod_{t\in[1,d]} U_{t}[\iwk{t}\gap\iwk{t+1}]}
\cdot \prod_{t\in [t^*,d]}\indi\sbra{i_{t}\neq i_1,i_{d+1}} .\end{align*}
Then, $\|\tilU\|_\frob \le O(d)\cdot  \|U_{t^*}\|_\frob$.
\label{lem:main_lemma_3}
\end{lemma} 
\begin{corollary} Under the same hypothesis as~\Cref{lem:main_lemma_3}, let $T\subseteq[d]$ and $L\subseteq [N]$. Define a matrix $\tilU$ such that for all $\iwk{1},\iwk{d+1}\in[M],S_{d+1}\subseteq[N],$
\begin{align*} \tilU[\iwk{1}\gap \iwks{d+1}] 
&=\sum_{\iwk{2},\ldots,\iwk{d}\in[M] } \pbra{\prod_{t\in[1,d]} U_{t}[\iwk{t}\gap\iwk{t+1}]}
\cdot \indi\sbra{S_{d+1}=\bigoplus_{\substack{t\in [1,d]\setminus T\\i_t\in L}}\{i_{t}\} }\cdot \prod_{t\in [t^*,d]}\indi\sbra{i_{t}\neq i_1,i_{d+1}}.\end{align*}
Then, $\|\tilU\|_\frob \le O(d)\cdot  \|U_{t^*}\|_\frob$.
\label{corollary:main_lemma_3}
\end{corollary} 
The proof of~\Cref{corollary:main_lemma_3} from~\Cref{lem:main_lemma_3} is identical to that of~\Cref{corollary:main_lemma_2} from~\Cref{lem:main_lemma_2} and is omitted. The proof of~\Cref{lem:main_lemma_3} is deferred to~\Cref{sec:app_main_lemma_3}.

\subsection{Matrix Decomposition Lemma for $\hBQP$}

Let $g:[N]\times [N]\to \mathbb{R}$ where $g(i,j)=(-1)^{\indi\sbra{i=j}}$ for $i,j\in[N]$. 

\begin{lemma} Let $U_1,\ldots,U_{d}$ be $M\times M$ matrices with $\|U_t\|_\op \le 1$ for $t\in[d]$. Let $t_1<t_2\in[2,d]$. Define a matrix $U$ such that for all $\iwk{1},\iwk{d+1}\in[M],$
\begin{align*} U[\iwk{1}\gap \iwk{d+1}] 
&=\sum_{\iwk{2},\ldots,\iwk{d}\in[M]} \pbra{\prod_{t\in[1,d]} U_{t}[\iwk{t}\gap\iwk{t+1}]}\cdot \prod_{t\in (t_1,d]}g(i_t,i_{t_1})\cdot \prod_{t\in (t_2,d]}g(i_t,i_{t_2}).\end{align*}
Then, $\|U\|_\frob \le O(d^4)\cdot \min_{t\in[d]}\pbra{\|U_t\|_\frob} $.
\label{lem:main_lemma_4}
\end{lemma}

\begin{corollary} Under the same hypothesis as~\Cref{lem:main_lemma_4}, let $T\subseteq[d]$ and $L\subseteq [N]$. Define a matrix $\tilU$ such that for all $\iwk{1},\iwk{d+1}\in[M],S_{d+1}\subseteq[N],$
\begin{align*} \tilU[\iwk{1}\gap \iwks{d+1}] 
&=\sum_{\iwk{2},\ldots,\iwk{d}\in[M]} \pbra{\prod_{t\in[1,d]} U_{t}[\iwk{t}\gap\iwk{t+1}]}\cdot \indi\sbra{S_{d+1}=\bigoplus_{\substack{t\in [1,d]\setminus T\\i_t\in L}}\{i_{t}\} }\\
&\cdot \prod_{t\in (t_1,d]}g(i_t, i_{t_1})\cdot \prod_{t\in (t_2,d]}g(i_t, i_{t_2}).\end{align*}
Then, $\|\tilU\|_\frob \le O(d^4)\cdot \min_{t\in[d]}\pbra{\|U_t\|_\frob } $.
\label{corollary:main_lemma_4}
\end{corollary}

The proof of~\Cref{corollary:main_lemma_4} from~\Cref{lem:main_lemma_4} is identical to that of~\Cref{corollary:main_lemma_2} from~\Cref{lem:main_lemma_2} and is omitted. The proof of~\Cref{lem:main_lemma_4} is deferred to~\Cref{sec:proof_main_lemma_4}.

\section{Fourier Growth of $\DQC{k}$: Level-Two}
\label{sec:proof_dqck}
%In~\Cref{sec:acceptance_dqck}, we will derive an expression for the acceptance probability of a $\DQC{k}$ algorithm. In~\Cref{sec:proof_dqck}, we will use this to prove~\Cref{thm:main_theorem_dqck}.
%\subsection{An Expression for the Fourier coefficients of a $\DQC{k}$ Algorithm}
In this section, we will show the $\ell=2$ case of~\Cref{thm:main_theorem_dqck}, i.e., the level-two Fourier growth of a $d$-query $\DQC{k}$ algorithm is at most $O(d^3\sqrt{K})$. Since $\DQC{k}$ algorithms are a sub-class of $\BQP$ algorithms, the bounds from~\Cref{thm:main_theorem_bqp} immediately apply to $\DQC{k}$ algorithms and complete the proof when $\min\pbra{2^{k/2},\sqrt{N}}=\sqrt{N}$. It suffices to handle the other case, i.e., $\min\pbra{2^{k/2},\sqrt{N}}=2^{k/2}$ which will be the focus of this section.

\label{sec:acceptance_dqck}
Throughout this section, to simplify notation, we use the shorthand $\iwk{t}$ to denote $(i_t,w_t,k_t)$ where $i_t\in[N],w_t\in[W],k_t\in [K]$ for $N=2^n,W=2^w,K=2^k$. We use $\bI$ to denote the identity matrix, where the dimension is implicit.

Let $f(x)$ be the acceptance probability of a $d$-query $\DQC{k}$ algorithm and $\rho$ be any restriction of the input variables. 
We will now derive an expression for the Fourier coefficients of $f|_\rho(x)$. Let $\rho$ keep the variables in $L\subseteq[N]$ alive and fix the rest. Thus, only Fourier coefficients corresponding to $S\subseteq[L]$ are non-zero and are described by the following claim.

\begin{claim}\label{claim:dqck_fourier_coefficients}
Let $f(x)$ be the acceptance probability of a $d$-query $\DQC{k}$ algorithm and let $\rho\in\{-1,1,\star\}^N$ be any restriction that leaves $L\subseteq[N]$ unfixed and fixes the rest. Then, there exist matrices $V^\rho_1,\ldots,V^\rho_{2d}$ such that for all $S\subseteq[L]$,
\begin{align*}  \widehat{f|_\rho}(S)&=(NW)^{-1}\sum_{\iwk{1},\ldots,\iwk{2d}\in[M]} \pbra{\prod_{t\in[2d]} V^\rho_t[\iwk{t}\gap\iwk{t+1}]} \cdot \indi\sbra{\bigoplus_{\substack{t\in[2d]\\\text{with }i_t\in L}} \{i_t\}=S}.\end{align*} where $V^\rho_1,\ldots,V^\rho_{2d}\in\C^{M\times M}$ satisfy $\|V_t^\rho\|_\op\le 1$ for $t\in[2d]$ and $\|V_1^\rho\|_\frob \le \sqrt{M/K}$ 
\end{claim}

The proof of this is fairly simple and is deferred to~\Cref{sec:appendix_fourier_coefficients}. We will now establish $L_{1,2}$ bounds for $\DQC{k}$ algorithms. The goal of this section is to upper bound 
\begin{align}
    L_{1,2}(f|_\rho) \triangleq \max_{\alpha\in [-1,1]^{\binom{N}{2}}} L_{1,2}^\alpha(f|_\rho)=\max_{\alpha\in [-1,1]^{\binom{N}{2}}}\sum_{S\in{\binom{[N]}{2}}} \alpha_S\cdot  \widehat{f|_\rho}(S).\label{eq:dqck_proof_1}
\end{align} 
Fix any $\alpha_S\in [-1,1]$ for each $S\in{\binom{L}{2}}$. From \Cref{eq:dqck_proof_1} and~\Cref{claim:dqck_fourier_coefficients}, we see that our goal is to upper bound
\begin{align} \label{eq:dqc1_6}L_{1,2}^\alpha(f|_\rho)= \sum_{\substack{S\subseteq L\\|S|=2}} (NW)^{-1}\sum_{\iwk{1},\ldots,\iwk{2d}\in[M]} \pbra{\prod_{t\in[2d]} V^\rho_t[\iwk{t}\gap\iwk{t+1}]} \cdot \indi\sbra{\bigoplus_{\substack{t\in[2d]\\\text{with }i_t\in L}} \{i_t\}=S}\cdot \alpha_S
\end{align}

Observe that if $\bigoplus_{ {t\in[2d],i_t\in L}} \{i_t\}$ has size $2$, then there must exist a subset $T=\{t_1,t_2\}\subseteq [2d]$ for $t_1\neq t_2$ such that $\{i_{t_1},i_{t_2}\}$ has two distinct elements in $L$ and $\bigoplus_{\substack{t\in[2d]\setminus T \\\text{with }i_t\in L}}\{i_t\}=\emptyset.$ For any fixed $S$, there may be many such $T$, but to uniquely identify one, we define $T$ to correspond to the first times that variables in $S$ appear. More precisely, let $t_1$ be the first time $t$ for which $i_t\in S$, and let $t_2$ be the first time $t>t_1$ for which $i_t\in S\setminus \{i_{t_1}\}$. This says that $i_{t_1}$ is the first element of $S$ to appear, $i_{t_2}$ is the next unseen element of $S$ to appear. Conversely, for any $T$ and $\{i_t\}_{t\in T}$ satisfying the above conditions, it defines a unique $S=\{i_t:t\in T\}$. 

Fix $T\subseteq[2d]$ of size $2$ (this can be done in $\binom{2d}{2}$ ways) and let $T=\{t_1,t_2\}$ for $t_1<t_2$. Define
\begin{align}\label{eq:dqc1_delta_higher}\begin{split}
\Delta_{T}&:=  \sum_{\iwk{1},\ldots,\iwk{2d}\in[M]}   \pbra{\prod_{t\in[2d]} V^\rho_t[\iwk{t}\gap\iwk{t+1}]} \cdot  \indi\sbra{\bigoplus_{\substack{t\in[2d]\setminus T\\\text{with }i_t\in L}} \{i_t\}=\emptyset}\\ 
& \cdot \prod_{t\in [1,t_1)}\indi\sbra{i_{t}\neq i_{t_1}}\cdot \prod_{t\in [1,t_2)}\indi\sbra{i_{t}\neq i_{t_2}} \cdot \indi\sbra{i_{t_1}\neq i_{t_2}\in L}\cdot \alpha_{\{i_{t_1},i_{t_2}\}}.\end{split}\end{align} 
From the above paragraph, it follows that
\[L_{1,2}^\alpha(f|_\rho)= (NW)^{-1} \sum_{T\in \binom{[2d]}{2}}\Delta_T \le \binom{2d}{2}\cdot (NW)^{-1}\cdot \max_{T\in\binom{[2d]}{2}} \abs{\Delta_T}.\]
We will now show that for all $T\in \binom{[2d]}{2}$, we have $\abs{\Delta_T}\le O(d)\cdot M\cdot K^{-1/2}$. This, along with the above equation (and the fact that $M=KNW$) would imply that $L_{1,2}^\alpha(f|_\rho)\le \sqrt{K} \cdot O(d^3)$ as desired. We now show the desired bound of $\Delta_T\le O(d)\cdot M\cdot K^{-1/2}$.
 
We will group the terms $t\in[2d]$ into circular intervals $[t_1,t_2),[t_2,t_1)$. Define matrices $\tilV_{[t_1,t_2)}$ and $\tilV_{[t_2,t_1)}$ such that for all $\iwk{t_1},\iwk{t_2}\in[M],S_{t_1}\subseteq [N]$, we have 
\begin{align}\begin{split}\label{eq:dqc1_7_higher} \tilV_{[t_1,t_2)}[\iwks{t_1}\gap\iwk{t_2}]&=\sum_{\iwk{t_1+1},\ldots,\iwk{t_2-1}\in[M]}  \pbra{ \prod_{ t\in [t_1,t_2)}V^\rho_t[\iwk{t}\gap\iwk{t+1}]}\cdot \indi\sbra{S_{t_1}=\bigoplus_{\substack{t\in (t_1,t_2)\\ i_t\in L }} \{i_t\}}\\
&\cdot \prod_{t\in (t_1,t_2)}\indi \sbra{i_{t}\neq i_{t_2}},\text{ and }\end{split}\\ \label{eq:dqc1_8}
\begin{split}\tilV_{[t_2,t_1)}[\iwk{t_2}\gap\iwks{t_1}]&=\sum_{\iwk{t_2+1},\ldots,\iwk{t_1-1}\in[M]}  \pbra{ \prod_{ t\in [t_2,t_1)}V^\rho_t[\iwk{t}\gap\iwk{t+1}]}\cdot \indi\sbra{S_{t_1}=\bigoplus_{\substack{t\in (t_2,t_1)\\ i_t\in L }} \{i_t\}}\\
&\cdot \prod_{t\in [1,t_1)}\indi \sbra{i_{t}\neq i_{t_2},i_{t_1}},\end{split}\end{align}
%Since $1\in(t_2,t_1]$ and since $\|V_1^\rho\|_\frob \le \sqrt{M/K}$,~\Cref{lem:main_lemma_3} implies that \begin{equation}\label{eq:higher_1}\|\tilV_{(t_2,t_1]}\|_\frob \le O(d)\cdot \sqrt{M/K}.\end{equation}
Define $\alpha'_{\{i_{t_1},i_{t_2}\}}=\alpha_{\{i_{t_1},i_{t_2}\}}\cdot \indi\sbra{i_{t_1},i_{t_2}\in L}$ for all $i_{t_1},i_{t_2}\in[N]$. Combining~\Cref{eq:dqc1_delta_higher} with~\Cref{eq:dqc1_7_higher,eq:dqc1_8}, and using the fact that
\[ \prod_{t\in [1,t_1)}\indi\sbra{i_{t}\neq i_{t_1}}\cdot \prod_{t\in [1,t_2)}\indi\sbra{i_{t}\neq i_{t_2}} = \indi\sbra{i_{t_1}\neq i_{t_2}}\cdot  \prod_{t\in [1,t_1)}\indi \sbra{i_{t}\neq i_{t_2},i_{t_1}}\cdot \prod_{t\in (t_1,t_2)}\indi \sbra{i_{t}\neq i_{t_2}}  \]
we obtain that
\begin{align*}
        \Delta_T& = \sum_{\substack{\iwk{t_1},\iwk{t_2}\in[M]\\ S_{t_2}\subseteq[N]}} \tilV_{[t_1,t_2)}[  \iwks{t_1}\gap \iwk{t_2}]\cdot \tilV_{[t_2,t_1)}[\iwk{t_2}\gap\iwks{t_1}]\cdot \alpha'_{\{i_{t_1},i_{t_2}\}}\cdot\indi\sbra{i_{t_1}\neq i_{t_2}}  \\
        |\Delta_T|&\le \| \tilV_{[t_1,t_2)}\|_\frob\cdot \| \tilV_{[t_2,t_1)}\|_\frob  \tag{by \Cref{fact:frob_product}}
\end{align*}
It suffices to upper bound $\|\tilV_{[t_1,t_2)}\|_\frob\cdot \| \tilV_{[t_2,t_1)}\|_\frob$. We consider two cases:

\paragraph*{Case 1:} $t_1=1$. In this case, observe that $1\in [t_1,t_2)$. We observe that the matrix $\tilV_{[t_1,t_2)}$ in~\Cref{eq:dqc1_7_higher} is obtained by applying~\Cref{corollary:main_lemma_2} on the matrices $V_{t_1}^\rho,\ldots,V_{t_2-1}^\rho$ (with omitted set $T$) in reverse order with transposes and hence, $\|\tilV_{[t_1,t_2)}\|_\frob \le O(d)\cdot \|V_1^\rho \|_\frob\le O(d)\cdot  \sqrt{M/K} $. Since $[1,t_1)=\emptyset$, the matrix $\tilV_{[t_2,t_1)}$ is obtained by applying~\Cref{lem:main_lemma} on matrices $V_{t_2}^\rho,\ldots,V_{t_1-1}^\rho$ (with omitted set $T$) in this order and hence, $\|\tilV_{[t_2,t_1)}\|_\frob \le \sqrt{M}$. In this case, \[\|\tilV_{[t_1,t_2)}\|_\frob\cdot \| \tilV_{[t_2,t_1)}\|_\frob \le O(d)\cdot \sqrt{M/K}\cdot \sqrt{M}\le O(d)\cdot M/\sqrt{K}.\]
\paragraph*{Case 2:} $t_1>1$. In this case, observe that $1\in (t_2,t_1)$. Observe that the matrix $\tilV_{[t_2,t_1)}$in~\Cref{eq:dqc1_8} is obtained by applying~\Cref{corollary:main_lemma_3} on matrices $V_{t_2}^\rho,\ldots,V_{t_1-1}^\rho$ (with omitted set $T$ and $t^*=1$, which is well-defined) in forward order, and hence, $\|\tilV_{[t_2,t_1)}\|_\frob \le O(d)\cdot \|V^\rho_1\|_\frob \le O(d)\cdot\sqrt{M/K}$. As before, the matrix $\tilV_{[t_1,t_2)}$ in~\Cref{eq:dqc1_7_higher} is obtained by applying~\Cref{corollary:main_lemma_2} on the matrices $V_{t_1}^\rho,\ldots,V_{t_2-1}^\rho$ (with omitted set $T$) in reverse order with transposes and hence, $\|\tilV_{[t_1,t_2)}\|_\frob \le  \sqrt{M}$. In this case, \[\|\tilV_{[t_1,t_2)}\|_\frob\cdot \| \tilV_{[t_2,t_1)}\|_\frob \le O(d)\cdot \sqrt{M}\cdot \sqrt{M/K} \le O(d)\cdot M/\sqrt{K}.\]

\subsection{Tightness of our Bounds for $\DQC{1}$}
\label{sec:tightness_DQC}

\begin{figure}
\centering
\mbox{ 
\Qcircuit @C=1em @R=.7em {
\lstick{} & &  &  \qw & \multigate{1}{H_N}  & \multigate{1}{O_{x^{(d)}}} &     \qw &  \cds{1}{\cdots\cdots}  & \multigate{1}{H_N} & \multigate{1}{O_{x^{(1)}}} & \qw \\
\lstick{}  & &  &  \qw  & \ghost{H_N}  & \ghost{O_{x^{(d)}}} & \qw & \cds{1}{\cdots\cdots} & \ghost{H_N} & \ghost{O_{x^{(1)}}} &  \qw  \inputgroupv{1}{2}{.8em}{.8em}{n} \\ 
\lstick{}  & \ket{0}   &  &  \gate{H} & \ctrl{-1}  & \ctrl{-1}\qw    & \qw  &\qw & \ctrl{-1} \qw   &\ctrl{-1} \qw & \gate{H}  &\meter & & \text{output}\\ \\
}
}
\caption{A $d$-query $\DQC{1}$ algorithm with $n$ maximally mixed qubits.}
\label{fig:tight_dqc1}
\end{figure}
In this section, we will show that the dependence on $k$ and $N$ is tight in~\Cref{thm:main_theorem_dqck}.

\paragraph*{Dependence on $N$.} First, we consider the case $k=1$ and show that $\DQC{1}$ algorithms can indeed achieve level-$\ell$ Fourier growth of roughly $N^{(\ell-2)/2}$, i.e., the dependence on $N$ is tight in~\Cref{thm:main_theorem_dqck}. We will do so by producing an algorithm on inputs of length $dN$ which makes $d$ oracle queries and whose level-$\ell$ Fourier growth for $\ell=d$ is $\Omega\pbra{N^{(\ell-2)/2}}$.

Let $H_N$ be the Hadamard matrix as in~\Cref{def:hadamard} and view this matrix as an $n$-qubit unitary operator. For $t\in[d]$, let $N_t$ denote the interval $((t-1) N,t N]$ so that $N_1\sqcup\ldots\sqcup N_d=[dN]$. We view the input $x\in\{-1,1\}^{Nd}$ as comprising of $d$ input strings $x^{(1)},\ldots,x^{(d)}$ of length $N$ each such that $x^{(t)}$ is supported on $N_t$. Instead of the oracle $O_x$, we will consider $d$ oracles $O_{x^{(1)}},\ldots,O_{x^{(d)}}$. Consider the $d$-query $\DQC{1}$ algorithm as in~\Cref{fig:tight_dqc1}.\footnote{ Typically, we express $\DQC{1}$ in terms of a single oracle $O_x$, as opposed to $d$ smaller oracles $O_{x^{(1)}},\ldots,O_{x^{(d)}}$, nevertheless, it is easy to embed the circuit in~\Cref{fig:tight_dqc1} into a larger one consisting only of $O_x$ oracle calls for $x=(x^{(1)},\ldots,x^{(d)})$ by applying the following sequence of operators $d$ times: $H_N\otimes \bI$, followed $O_x$, followed by the permutation matrix $\Pi$ that maps $\ket{i}\to \ket{i-N\text{ (mod }Nd)}$ for all computational basis states $i\in[Nd]$.} As we saw in~\Cref{eq:proof_overview_1}, it is not too difficult to show the bias of this algorithm is precisely
\[ f(x)=\tfrac{1}{2N}\Tr\pbra{O_{x^{(1)}}\cdot H_N \cdots O_{x^{(d)}}\cdot H_N}.\]
We observe the Fourier coefficients of $f$ correspond to subsets $S\subseteq[Nd]$ that pick exactly one element from each $N_t$. There are $N^d$ such non-zero Fourier coefficients and they are given by
\begin{align*} \widehat{f}(S)&=\tfrac{1}{2N}\sum_{\substack{i_t\in N_t\\\text{ for }t\in[d]}}(-1)^{\langle i_1,i_2\rangle+ \ldots + \abra{i_d,i_1}}\cdot \frac{1}{N^{d/2}}\cdot  \indi\sbra{S=\{i_1, \ldots, i_d\}}. 
\end{align*}
Each such $S$ uniquely identifies $i_1\in N_1,\ldots,i_d\in N_d$ and we set $\alpha_S:=(-1)^{\langle i_1,i_2\rangle+\ldots+\langle i_d,i_1\rangle}$. Thus, we obtain that the level-$d$ Fourier growth is at least
\[N^d\cdot \frac{1}{2N} \cdot \frac{1}{N^{d/2}}\ge \Omega\pbra{N^{(d-2)/2}}. \]
This completes the proof.

\paragraph*{Dependence on $k$.} It is clear to see that a $\DQC{k}$ algorithm can solve the Forrelation problem on inputs of length $2^k$, since we can run the $k$-qubit Forrelation circuit on the clean qubits. As the Forrelation function on $2^k$-bit inputs has level-two Fourier growth of $2^{k/2}$, this saturates the bound from~\Cref{thm:main_theorem_dqck} for level two.

\section{Fourier Growth of $\hBQP$: Level-Three}
\label{sec:hbqp}

In this section, we will show the $\ell=3$ case of~\Cref{thm:main_theorem_hbqp}, i.e., the level-three Fourier growth of a $d$-query $\hBQP$ algorithm is at most $O(d^7\sqrt{N})$. Throughout this section, to simplify notation, we use the shorthand $\iwk{t}$ to denote $(i_t,w_t)$ where $i\in[N],w\in[W]$ for $N=2^n,W=2^w$. We use $\bI$ to denote the identity matrix, where the dimension is implicit.

Given the expression for the acceptance probability of a $d$-query $\hBQP$ algorithm (\Cref{claim:acceptance_probability_hbqp}), it is not too difficult to derive an expression for the Fourier coefficients under any restriction -- this part is similar to the proof of~\Cref{claim:dqck_fourier_coefficients} from~\Cref{claim:dqck_acceptance_probability}. We obtain the following claim, whose proof is deferred to~\Cref{sec:appendix_fourier_coefficients}. 

\begin{claim}\label{claim:fourier_coefficients_hbqp}
Let $f(x)$ be the acceptance probability of a $d$-query $\hBQP$ algorithm and $\rho\in\{-1,1,\star\}^N$ be any restriction that leaves the coordinates in $L\subseteq[N]$ free and fixes the rest. Then, there exist matrices $V_1^\rho,\ldots,V_{2d+2}^\rho \in \mathbb{C}^{M\times M}$ such that for all $S\subseteq L$, 
\begin{align}\label{eq:hbqp_3} \widehat{f|_\rho}(S)&=M^{-1} \sum_{\iw{1},\ldots,\iw{2d+2}\in[M]}  F_{\iw{1},\iw{d+2}}\cdot \prod_{t\in[2d+2]} V^\rho_t[\iw{t}\gap\iw{t+1}] \cdot   \indi\sbra{ \bigoplus_{\substack{t\in[2d+2]\setminus\{1,d+2\}\\i_t\in L}} \{i_t\} =S}
\end{align}
where $\|V_t^\rho\|_\op \le 1$ for all $t\in[2d+2]$.
\end{claim}

Now that we have an expression for the Fourier coefficients, we turn our attention to proving Fourier growth bounds. Let $\tilA=A\cap L,\tilB=B\cap L,\tilC=C\cap L$, where $A,B,C$ are as in~\Cref{def:alpha}. Let $D=[2d+2]\setminus \{1,d+2\}$. We use $I$ to denote $\iw{1},\ldots,\iw{2d+2}$. As mentioned before, we will only be able to bound $L_{1,3}^{\alpha(\gamma)}(f|_\rho)$ where $\gamma\in[-1,1]^{3N}$ and $\alpha(\gamma)$ is as in~\Cref{def:alpha}. Fix any such $\alpha(\gamma)$. Recalling~\Cref{def:alpha}, for $\alpha(\gamma)_S$ to be non-zero, we must have $S=\{a,b,c\}$ where $a\in \tilA,b\in \tilB, c\in \tilC$. Fix any such $S$. For $\bigoplus_{\substack{t\in D\\i_t\in L}} \{i_t\} =S$, there exist distinct $t_1,t_2,t_3\in D$ such that 
\[i_{t_1}=a,i_{t_2}=b,i_{t_3}=c\text{ and } \bigoplus_{\substack{t\in D\setminus\{t_1,t_2,t_3\}\\i_t\in L}} \{i_t\} =\emptyset,\]
however, these $t_1,t_2,t_3$ may not be unique. Instead of defining them uniquely, we add each of them but with a multiplier function that cancels out redundancies. 

Let $g(i_1,i_2)=(-1)^{\indi\sbra{i_1=i_2}}$ for $i_1,i_2\in[N]$ be the function as before. For $q\in[3]$, let 
\begin{equation} r_{t_q}[I]:=(-1)^{\# \{ t<t_q, t\in D\gap i_t=i_{t_q}\}}=\prod_{t<t_q,t\in D} g(i_t,i_{t_q}).\label{eq:r_t_q}\end{equation}
\begin{equation} s_{t_q}[I]:=(-1)^{\# \{ t>t_q, t\in D\gap i_t=i_{t_q}\}}=\prod_{t>t_q,t\in D} g(i_t,i_{t_q}) \label{eq:s_t_q}.\end{equation}
 
\begin{claim}
    Consider $\Phi(I)$ defined as follows.
    \[\Phi(I):=\sum_{\substack{ t_1,t_2,t_3\in D\\\text{distinct}}} {r_{t_1}[I]}\cdot \indi\sbra{i_{t_1}=a}\cdot  {r_{t_2}[I]}\cdot \indi\sbra{i_{t_2}=b}\cdot  {r_{t_3}[I]}\cdot \indi\sbra{i_{t_3}=c} \]
    Whenever  $\bigoplus_{\substack{t\in D\\i_t\in L}} \{i_t\} =\{a,b,c\}$, we have $\Phi(I)=1$.\label{claim:phi}
\end{claim}
\begin{proof}[Proof of~\Cref{claim:phi}] Consider\:
\begin{align*} \Phi(I):=&\sum_{\substack{ t_1,t_2,t_3\in D\\\text{distinct}}} {r_{t_1}[I]}\cdot \indi\sbra{i_{t_1}=a}\cdot  {r_{t_2}[I]}\cdot \indi\sbra{i_{t_2}=b}\cdot  {r_{t_3}[I]}\cdot \indi\sbra{i_{t_3}=c} \\
&=\sum_{ \substack{t_2, t_3\in D\\\text{ distinct }}} {r_{t_2}[I]}\cdot \indi\sbra{i_{t_2}=b}\cdot  {r_{t_3}[I]}\cdot \indi\sbra{i_{t_3}=c}\cdot \sum_{\substack{t_1\in D\\t_1\neq t_2,t_3}} {r_{t_1}[I]}\cdot \indi\sbra{i_{t_1}=a}\\
&= \sum_{ \substack{t_2, t_3\in D\\\text{ distinct }}} {r_{t_2}[I]}\cdot \indi\sbra{i_{t_2}=b}\cdot  {r_{t_3}[I]}\cdot \indi\sbra{i_{t_3}=c}\cdot 1
\end{align*}
where the last equality follows because for any fixed $t_2,t_3$, we have $i_{t_1}=a$ for an odd number of $t_1$ (because of the condition $\oplus_{\substack{t\in D\\i_t\in L}}\{i_t\}=\{a,b,c\}$) and any such $t_1$ is distinct from $t_2,t_3$ (since $b,c\neq a$) and as $t_1$ varies over these possible times, the quantity ${r_{t_1}[I]}$ alternates between $1$ and $-1$ (at the first time it is $+1$, the second time it is $-1$, and so on). We can now repeat the same argument to obtain
\begin{align*}
\Phi(I)&:=\sum_{t_2\in D}  {r_{t_2}[I]}\cdot \indi\sbra{i_{t_2}=b}\cdot \pbra{\sum_{\substack{t_3\in D\\t_3\neq t_2}} {r_{t_3}[I]}\cdot \indi\sbra{i_{t_3}=c}}= \sum_{t_2\in D} {r_{t_2}[I]}\cdot \indi\sbra{i_{t_2}=b}=1.
\end{align*}
\end{proof}
Therefore, for any $a\in \tilA,b\in \tilB,c\in\tilC$ we have
\begin{align*}
\indi\sbra{\bigoplus_{\substack{t\in D \\i_t\in L}} \{i_t\} =\{a,b,c\}}&= \indi\sbra{\bigoplus_{\substack{t\in D\\i_t\in L}} \{i_t\} =\{a,b,c\}} \cdot \Phi(I)\\
&=\sum_{\substack{t_1,t_2,t_3\in D\\\text{ distinct }}} \indi\sbra{\bigoplus_{\substack{t\in D\setminus\{t_1,t_2,t_3\}\\i_t\in L}} \{i_t\} =\emptyset }\\
&\cdot \indi\sbra{i_{t_1}=a}\cdot \indi\sbra{i_{t_2}=b}\cdot \indi\sbra{i_{t_3}=c}\cdot  {r_{t_1}(I)}\cdot {r_{t_2}(I)}\cdot  {r_{t_3}(I)}
\end{align*}

Now, we multiply both sides by $\alpha(\gamma)_{a,b,c}$ and sum over all possible $a \in \tilA,b\in \tilB,c\in \tilC$ to obtain 
\begin{align*}
&L_{1,3}^{\alpha(\gamma)}[f|_\rho]\triangleq \sum_{a\in \tilA,b\in \tilB,c\in \tilC} \alpha(\gamma)_{a,b,c}\cdot \widehat{f|_\rho}(\{a,b,c\}) \\
&=M^{-1}\sum_{\substack{a\in\tilA,b\in \tilB,c\in \tilC\\ \iw{1},\ldots\iw{2d+2}\in[M]}}F_{\iw{1},\iw{d+2}}\cdot \prod_{t\in[2d+2]}V_t^\rho[\iw{t}\gap\iw{t+1}]\cdot \indi\sbra{\bigoplus_{\substack{t\in D\\i_t\in L}} \{i_t\} =\{a,b,c\}}\cdot \alpha(\gamma)_{a,b,c}\tag{by \Cref{claim:fourier_coefficients_hbqp}}\\
&=M^{-1}\sum_{\substack{a\in\tilA,b\in \tilB,c\in \tilC\\ \iw{1},\ldots\iw{2d+2}\in[M]}} F_{\iw{1},\iw{d+2}}\cdot \prod_{t\in[2d+2]}V_t^\rho[\iw{t}\gap\iw{t+1}]\cdot \indi\sbra{\bigoplus_{\substack{t\in D\setminus\{t_1,t_2,t_3\}\\i_t\in L}} \{i_t\} =\emptyset}\\
&\cdot \sum_{\substack{t_1,t_2,t_3\in D\\\text{ distinct }}}  \indi\sbra{i_{t_1}=a}\cdot \indi\sbra{i_{t_2}=b}\cdot \indi\sbra{i_{t_3}=c}\cdot  {r_{t_1}(I)}\cdot  {r_{t_2}(I)}\cdot  {r_{t_3}(I)}\cdot \alpha(\gamma)_{i_{t_1},i_{t_2},i_{t_3}}\\
&=M^{-1}\sum_{\substack{ \iw{1},\ldots\iw{2d+2}\in[M] \\ t_1,t_2,t_3\in D\text{ distinct }}} F_{\iw{1},\iw{d+2}}\cdot \prod_{t\in[2d+2]}V_t^\rho[\iw{t}\gap\iw{t+1}]\cdot \indi\sbra{\bigoplus_{\substack{t\in D\setminus\{t_1,t_2,t_3\}\\i_t\in L}} \{i_t\} =\emptyset}\\
&\cdot \indi\sbra{i_{t_1}\in \tilA}\cdot \indi\sbra{i_{t_2}\in \tilB}\cdot \indi\sbra{i_{t_3}\in \tilC}\cdot {r_{t_1}(I)}\cdot  {r_{t_2}(I)}\cdot  {r_{t_3}(I)}\cdot \alpha(\gamma)_{i_{t_1},i_{t_2},i_{t_3}}
\end{align*}
where the last line follows from the fact that $\sum_{a\in \tilA}\indi\sbra{i_{t_1}=a}=\indi\sbra{i_{t_1}\in \tilA}$. Fix any $i_{t_2}^*\in \tilB$ and distinct $t_1,t_2,t_3\in D$. (There are at most $N$ possibilities for such $i_{t_2}^*$ and $O(d^3)$ possibilities for $(t_1,t_2,t_3)$.) Define
\begin{align} \label{def:hbqp_delta} \begin{split}\Delta_{t_1,t_2,t_3,i_{t_2}^*}^{\gamma}&:=
\sum_{\iw{1},\ldots,\iw{2d+2}\in[M]} F_{\iw{1},\iw{d+2}}\cdot\prod_{t\in[2d+2]}V_t^\rho[\iw{t}\gap\iw{t+1}]  \cdot \indi\sbra{\bigoplus_{\substack{t\in D\setminus\{t_1,t_2,t_3\}\\\text{with }i_t\in L}} \{i_t\}=\emptyset}\\
&\cdot \indi\sbra{i_{t_1}\in \tilA}\cdot \indi\sbra{i_{t_2}=i_{t_2}^*}\cdot \indi\sbra{i_{t_3}\in \tilC} \cdot  {r_{t_1}(I)}\cdot  {r_{t_2}(I)}\cdot {r_{t_3}(I)}\cdot \alpha(\gamma)_{i_{t_1},i_{t_2}^*,i_{t_3}} \end{split}\end{align}
Substituting this in the expression for the Fourier growth, we have
\begin{align*}  L_{1,3}^{\alpha(\gamma)}
&= M^{-1} \sum_{\substack{\text{distinct } t_1,t_2,t_3\\\text{in }[2d+2]\setminus\{1,d+2\},\\ i_{t_2}^*\in \tilB}} \Delta_{t_1,t_2,t_3,i_{t_2}^*}^{\gamma} \tag{from ~\Cref{def:hbqp_delta,claim:fourier_coefficients_hbqp}} \\
&\le M^{-1} \cdot O(d^3 N)\cdot \max_{\substack{\text{distinct } t_1,t_2,t_3\\\text{in }[2d+2]\setminus\{1,d+2\},\\ i_{t_2}^*\in \tilB}}\abs{\Delta_{t_1,t_2,t_3,i_{t_2}^*}^{\gamma}} . \end{align*}
We will show that for each distinct $t_1,t_2,t_3\in [2d+2]\setminus\{1,d+2\}$ and $i_{t_2}^*\in \tilB$, we have $\abs{\Delta^\gamma_{t_1,t_2,t_3,i_{t_2}^*}}\le O(d^4)\cdot \sqrt{MW}$. Substituting this above, we would get 
\[ L_{1,3}^{\alpha(\gamma)}(f|_\rho) \le M^{-1}\cdot O(d^7 N)\cdot \sqrt{MW}=O(d^7)\cdot  \sqrt{N}\cdot \sqrt{NMW}\cdot M^{-1}\le O(d^7)\cdot \sqrt{N}, \]
where we used the fact that $M=NW$. It now suffices prove the bound  $\Delta^\gamma_{t_1,t_2,t_3,i_{t_2}^*}\le O(d^4)\cdot \sqrt{MW}$. 

Fix any distinct $t_1,t_2,t_3\in D$ and $i_{t_2}^*\in \tilB$. We will now use \Cref{def:alpha} to get
\[\alpha(\gamma)_{i_{t_1},i_{t_2},i_{t_3}}= \tilH(i_{t_1},i_{t_2}^*)\cdot  \tilH(i_{t_3},i_{t_2}^*)\cdot \gamma_{i_{t_1}}\cdot \gamma_{i_{t_2}^*}\cdot \gamma_{i_{t_3}}\]
We will use this to encode the action of multiplication by $\alpha(\gamma)_{i_{t_1},i_{t_2}^*,i_{t_3}}$ using a matrix product with diagonal matrices. Define matrices $P_t$ as follows. Let $\Gamma_{t_1},\Gamma_{t_2},\Gamma_{t_3}$ be $M\times M$ diagonal matrices with $[-1,1]$-valued entries defined as follows. For $\iw{t_1},\iw{t_2},\iw{t_3}\in [M]$, let
\begin{equation}\label{eq:diag_1} \Gamma_{t_1}[\iw{t_1}\gap\iw{t_1}]= \begin{cases}\gamma_{i_{t_1}}\cdot \tilH(i_{t_1},i_{t_2}^*) &\text{if }i_{t_1}\in \tilA\\ 0 &\text{otherwise.}\end{cases} \end{equation}
\begin{equation} \label{eq:diag_2} \Gamma_{t_2}[\iw{t_2}\gap \iw{t_2}]=\begin{cases} \gamma_{i_{t_2}^*} &\text{if } i_{t_2}=i_{t_2}^* \\ 0&\text{otherwise.}\end{cases} \end{equation}
\begin{equation}\label{eq:diag_3} \Gamma_{t_3}[\iw{t_3}\gap\iw{t_3}]= \begin{cases}\gamma_{i_{t_3}}\cdot \tilH(i_{t_3},i_{t_2}^*) &\text{if }i_{t_3}\in  \tilC\\ 0&\text{otherwise.}\end{cases}\end{equation}
Let $P'_{t_1}:= \Gamma_{t_1}\cdot V_{t_1}^\rho$,  $P'_{t_2}:= \Gamma_{t_2}\cdot V_{t_2}^\rho$, and $P'_{t_3}:=\Gamma_{t_3}\cdot V_{t_3}^\rho$ and for all $t\neq t_1,t_2,t_3$, set $P'_t=V_t^\rho$. Finally, define a diagonal matrix $D^{(i_{t_2}^*)}=(\id-2\ket{i_{t_2}^*}\bra{i_{t_2}^*})\otimes \id$, or equivalently, for all $I\in[M]$, let
\[
D^{(i_{t_2}^*)}[I,I]:=\begin{cases} -1 & \text{if }i=i_{t_2^*} \\ 1 &\text{otherwise} \end{cases}.\] 
For $t\in D,t<t_2$, set $P_t=D^{(i_{t_2}^*)}\cdot P'_t$ and for all other $t$, set $P_t= P'_t$.
Firstly, observe that 
\begin{align} \|P_t\|_\op\le 1\text{  for all }t\in[d], \label{eq:hbqp_5}
\end{align}
since  $\gamma\in[-1,1]^{3N}$ and $\|D^{(i_{t_2}^*)}\|_\op,\|\Gamma_t\|_\op, \|V_t^\rho\|_\op\le 1$. Secondly, observe that \begin{align} \label{eq:hbqp_6}\|P_{t_2}\|_\frob\le \sqrt{W},
\end{align}
since $P_{t_2}=\Gamma_{t_2}\cdot V_{t_2}^\rho$ and multiplying by the matrix $\Gamma_{t_2}$ has the effect of zeroing out all but $W$ rows (only the rows indexed by $i_{t_2}^*$ survive), and each row of $V_{t_2}$ has norm at most one. Furthermore, recalling the definition of $r_{t_2}(I)$ in~\Cref{eq:r_t_q}, we see that multiplication by $D^{(i_{t_2}^*)}$ in the appropriate locations has the effect of multiplying by ${r_{t_2}(I)}$. This construction allows us to simplify~\Cref{def:hbqp_delta} as
\begin{align}  \label{eq:hbqp_7}\begin{split}
\Delta_{t_1,t_2,t_3,i_{t_2}^*}^{\gamma}&= 
\sum_{\iw{1},\ldots,\iw{2d+2}\in[M]} F_{\iw{1},\iw{d+2}}\cdot \prod_{t\in[2d+2]}P_t[\iw{t}\gap\iw{t+1}] \cdot \indi\sbra{\bigoplus_{\substack{t\in D\setminus\{t_1,t_2,t_3\}\\\text{with }i_t\in L}}\{ i_t\}=\emptyset}\\
&\cdot {r_{t_1}(I)}\cdot {r_{t_3}(I)}.\end{split}\end{align}
Because we are imposing the constraint $\bigoplus_{\substack{t\in D\setminus\{t_1,t_2,t_3\}\\\text{with }i_t\in L}}\{ i_t\}=\emptyset$, recalling~\Cref{eq:r_t_q,eq:s_t_q}, we must have
\[ {r_{t_1}(I)}= {s_{t_1}(I)}\quad \text{ and } {r_{t_3}(I)}={s_{t_3}(I)}. \]
This is because for all $q\in\{1,3\}$, the number of times $i_{t_q}$ occurs prior to $t_q$ plus the number of times $i_{t_q}$ occurs after $t_q$ must be even, otherwise, the set  $\bigoplus_{\substack{t\in D\setminus\{t_1,t_2,t_3\}\\\text{with }i_t\in L}}\{ i_t\}$ would contain $i_{t_q}$. 

Let $D_L=D\cap [1,d+1]$ and $D_R=D\cap [d+2,2d+2]$. For $q\in\{1,3\}$, let $\eta(q)=D_L$ if $t_q\in D_L$ and $\eta(q)=D_R$ if $t_q\in D_R$. 
Define matrices $\tilP_{[1,d+1]}$ and $\tilP_{[d+2,2d+2]}$ such that for all $\iw{1},\iw{d+2}\in[M],S_{d+2}\subseteq[N],$ we have  
\begin{align} \label{eq:hbqp_8}\begin{split} \tilP_{[1,d+1]}[\iws{1}\gap\iw{d+2}]&= \sum_{\iw{2},\ldots,\iw{d+1}\in[M]}\pbra{\prod_{t\in[1,d+1]}P_t[\iw{t}\gap \iw{t+1}]}\cdot \indi\sbra{S_{1}=\bigoplus_{\substack{t\in D_L\setminus\{t_1,t_2,t_3\} \\\text{with }i_t\in L}}\{i_t\}} \\
& \cdot \prod_{q\in\{1,3\}:\eta(q)=D_L}  {r_{t_q}(I)}
\end{split}\end{align}
\begin{align} \label{eq:hbqp_9} \begin{split}\tilP_{[d+2,2d+2]}[ \iw{d+2}\gap \iws{1}]&= \sum_{\iw{d+3},\ldots,\iw{2d+2}\in[M]}\pbra{\prod_{t\in[d+2,2d+2]}P_t[\iw{t}\gap \iw{t+1}]}\cdot \indi\sbra{S_{1}=\bigoplus_{\substack{t\in D_R\setminus\{t_1, t_2,t_3\}\\\text{ with }i_t\in L}}\{i_t\}} \\
&\cdot \prod_{q\in\{1,3\}:\eta(q)=D_R}  {s_{t_q}(I)}\end{split} \end{align}
%\begin{equation} \|\tilP_{[1,d+1]}\|_\frob \le \min_{t\in[1,d+1]}\|P_t\|_\frob \quad\text{ and }\quad  \|\tilP_{[d+2,2d+2]}\|_\frob \le \min_{t\in[d+2,2d+2]}\|P_t\|_\frob.   \label{eq:hbqp_10} \end{equation}
Recall that $D_L\cup D_R=D$ and $[1,d+1]\cup[d+2,2d+2]=[1,2d+2]$. Plugging in~\Cref{eq:hbqp_8,eq:hbqp_9} into~\Cref{eq:hbqp_7}, we have
\begin{align*}
\Delta^\gamma_{t_1,t_2,t_3,i_{t_2}^*} &=\sum_{\substack{\iw{1},\iw{d+2}\in [M]\\S_{d+2}\subseteq[N]}} F_{\iw{1},\iw{d+2}}\cdot \tilP_{[1,d+1]}[\iws{1}\gap\iw{d+2}]\cdot \tilP_{[d+2,2d+2]}[\iw{d+2}\gap \iws{1}] \\
&\le \sum_{\substack{\iw{1},\iw{d+2}\in [M]\\S_{d+2}\subseteq[N]}} \abs{\tilP_{[1,d+1]}[\iws{1}\gap\iw{d+2}]}\cdot \abs{\tilP_{[d+2,2d+2]}[\iw{d+2}\gap \iws{1}]} \tag{since $F_{\iwk{1},\iwk{d+2}}\in\{0,1\}$}\\
&\le \|\tilP_{[1,d+1]}\|_\frob \cdot  \|\tilP_{[d+2,2d+2]}\|_\frob \tag{by~\Cref{fact:frob_product}}. 
\end{align*} 
It suffices to show upper bounds on $\|\tilP_{[1,d+1]}\|_\frob \cdot \|\tilP_{[d+2,2d+2]}\|_\frob$. We consider two cases separately.

\paragraph{Case 1: $\eta(1)\neq \eta(3)$.} In this case, either $t_1\in D_L$ and $t_3\in D_R$ or $t_1\in D_R$ and $t_3\in D_L$. For simplicity, consider the case that $t_1\in D_L,t_3\in D_R$ and the analysis for the other case is identical. In this case, we see that $\{t\in D,t>t_3\}=(t_3,2d+2]$. Thus, $\tilP_{[d+2,2d+2]}$ is the matrix one would obtain by applying~\Cref{corollary:main_lemma_2} to the matrices $P_{d+2},\ldots, P_{2d+2}$ (with parameters $\gamma=2$ and $t_1\leftarrow t_3$ and omitted set $\{t_1,t_2,t_3,1\}$). This implies that 
\[ \|\tilP_{[d+2,2d+2]}\|_\frob \le O(d)\cdot \min_{t\in[d+2,2d+2]}\|P_t\|_\frob. \]
Similarly, we see that $\{t\in D,t<t_1\}=[2,t_1)$ and hence, the matrix $\tilP_{[1,d+1]}$ is the matrix one would obtain by applying~\Cref{corollary:main_lemma_2} in reverse order to the transpose of the matrices $P_{1},\ldots, P_{d+1}$ (with parameters $\gamma=2$ and $t_1\leftarrow t_1$ and omitted set $\{t_1,t_2,t_3,d+2\}$). This implies that 
\[ \|\tilP_{[1,d+1]}\|_\frob \le O(d)\cdot \min_{t\in[1,d+1]}\|P_t\|_\frob. \]
Altogether, since $\|P_{t_2^*}\|_\frob =\sqrt{W}$ and $\|P_t\|_\frob \le\sqrt{M}$ for all $t\in[2d+2]$, we have
\[ \|\tilP_{[d+2,2d+2]}\|_\frob\cdot \|\tilP_{[1,d+1]}\|_\frob \le O(d^2) \cdot \sqrt{MW}. \]
\paragraph{Case 2: $\eta(1)= \eta(3)$.}
 In this case, either $t_1,t_3\in D_L$ or $t_1,t_3\in D_R$. For simplicity, consider the case that $t_1,t_3\in D_R$ and the analysis for the other case is identical. In this case, observe that $\{t>t_1,t\in D\}=(t_1,2d+2]$ and $\{t>t_3,t\in D\}=(t_3,2d+2]$. Hence, $\tilP_{[d+2,2d+2]}$ is precisely the matrix obtained by applying~\Cref{corollary:main_lemma_4} to the matrices $P_{d+2},\ldots,P_{2d+2}$ (with parameters $t_1\leftarrow t_1,t_2\leftarrow t_3$ and omitted set $\{t_1,t_2,t_3,d+2\}$). This implies that 
 \[ \|\tilP_{[d+2,2d+2]}\|_\frob \le O(d^4)\cdot \min_{t\in[d+2,2d+2]}\|P_t\|_\frob \]
Similarly, observe that $\tilP_{[1,d+1]}$ is precisely the matrix obtained by applying~\Cref{lem:main_lemma} to the matrices $P_{1},\ldots,P_{d+1}$ in reverse order (with omitted set $\{t_1,t_2,t_3,d+2\}$). This implies that 
 \[ \|\tilP_{[1,d+1]}\|_\frob \le \min_{t\in[1,d+1]}\|P_t\|_\frob .\]
Altogether, since $\|P_{t_2^*}\|_\frob =\sqrt{W}$ and $\|P_t\|_\frob \le\sqrt{M}$ for all $t\in[2d+2]$, we have
\[ \|\tilP_{[d+2,2d+2]}\|_\frob\cdot \|\tilP_{[1,d+1]}\|_\frob \le O(d^4) \cdot \sqrt{MW}. \]
This completes the proof.

\section{Higher Level Fourier Growth from Lower Levels}
\label{sec:bootstrapping}
In this section, we will present the bootstrapping argument that allows us to prove higher level Fourier growth bounds from lower levels. A version of this argument originally appeared in~\cite{CHLT19}.

\begin{theorem}
    \label{thm:bootstrapping}
Let $\cF$ be any restriction-closed family of boolean functions on $\{-1,1\}^N$. Then, for all $\ell\in[N],$ we have $L_{1,\ell+1}(\cF)\le  L_{1,\ell}(\cF)\cdot  O\pbra{\sqrt{N\log{\binom{N}{\ell}} \log N}}$.
\end{theorem}
Using the fact that $L_{1,1}$ of $d$-query bounded low-degree polynomials is at most $O(d)$~\cite{IRR+21}, we obtain the following corollary, which as a special case, includes $d/2$-query $\BQP$ algorithms. 
\begin{corollary}[\Cref{thm:main_theorem_bqp}] Let $f$ be a degree-$d$ polynomial bounded on $\{\pm 1\}^N$ by $[-1,1]$ and $\rho$ be any restriction, then $L_{1,\ell}(f|_\rho)\le c^\ell\cdot  d\cdot N^{(\ell-1)/2}\cdot \log^{\ell-1}(N)\cdot \sqrt{\ell!}$ for some constant $c>0$.
\end{corollary}
Using the fact that $L_{1,2}$ of $d$-query $\DQC{k}$ algorithms is at most $O(d^3)$ as proved in~\Cref{sec:proof_dqck}, we obtain the following corollary. 
\begin{corollary}[\Cref{thm:main_theorem_dqck}] Let $f$ be a $d$-query $\DQC{k}$ algorithm on oracle of length $N$ and $\rho$ be any restriction, then $L_{1,\ell}(f|_\rho)\le \sqrt{K}\cdot c^\ell\cdot d^3\cdot N^{(\ell-2)/2}\cdot \log^{\ell-2}(N)\cdot \sqrt{\ell!}$ for some constant $c>0$.
\end{corollary}
For $\hBQP$ algorithms, combining~\Cref{thm:bootstrapping} and the level-3 bound will only give a $\poly(d)\cdot N^2$ bound on the level-6 bound, which does not suffice to prove lower bounds for \forrthree. We will require a more refined analysis described in~\Cref{sec:improved_bootstrapping}.

\begin{proof}[Proof of~\Cref{thm:bootstrapping}]
Fix any $f\in \cF$ and any signs $\alpha_S\in[-1,1]$ for $|S|=\ell+1$. We wish to bound $L_{1,\ell+1}^\alpha(f)=\sum_{|S|=\ell+1}\alpha_S\hat{f}(S)$. We randomly partition $[N]$ into $X\sqcup Y$ as follows. For each $i\in[N]$, independently include it in $X$ with probability $1/(\ell+1)$. Define 
\begin{equation} L(X):= \sum_{|S|=\ell+1} \alpha_S\cdot\hat{f}(S)\cdot \indi\sbra{|S\cap X|=1}=\sum_{\substack{i\in X\\ T\subseteq [Y],|T|=\ell}} \alpha_{T\cup\{i\}}\hat{f}(T\cup\{i\}) .\label{eq:bootstrap_0}\end{equation}
Observe that for any set $S\subseteq[N]$ of size $\ell+1$, the probability that $|S\cap X|=1$ is precisely $c(\ell)=(\ell+1)\cdot \frac{1}{\ell+1}\cdot \pbra{\frac{\ell}{\ell+1}}^\ell \ge \Omega(1).$ This implies that
\begin{equation}L_{1,\ell+1}^\alpha(f)=\mathbb{E}_X[L(X)]\cdot \tfrac{1}{c(\ell)}. \label{eq:bootstrap_1} \end{equation}
Since $c(\ell)\ge \Omega(1)$, it suffices to upper bound $\E_X[L(X)]$ and we will in fact bound $L(X)$ for any partition $X$. To do this, we express $\hat{f}(T\cup \{i\})$ as $\E_{\substack{x\sim \{\pm 1\}^X\\y\sim \{\pm 1\}^{Y}}}\sbra{f(x,y)\cdot x_i\cdot \chi_T(y)}$ where the underlying distributions are uniform -- this follows from the definition of Fourier coefficients. Substituting this in~\Cref{eq:bootstrap_0}, we obtain 
\begin{align}\begin{split} L(X)&=\sum_{\substack{i\in X\\ T\subseteq [Y],|T|=\ell}} \alpha_{T\cup\{i\}}\cdot \E_{\substack{x\sim \{\pm 1\}^X\\y\sim \{\pm 1\}^{Y}}}\sbra{f(x,y)\cdot x_i\cdot \chi_T(y)} \\
&= \E_{\substack{x\sim \{\pm 1\}^X\\y\sim \{\pm 1\}^{Y}}}\sbra{ f(x,y)\cdot \sum_{T\subseteq [Y],|T|=\ell}\chi_T(y)\cdot \sum_{i\in X} x_i\cdot \alpha_{T\cup\{i\}}}. 
\end{split}\end{align}
Consider the function $f_x$ where $f_x(y):=f(x,y)$ for all $y\in \{\pm 1\}^Y$. By definition, we have $\hat{f_x}(T)=\E_{Y\sim \{\pm 1\}^Y}[f(x,y)\cdot \chi_T(y)]$. Substituting this above, we see that
\begin{align}\begin{split} L(X)&= \E_{x\sim \{\pm 1\}^X}\sbra{  \sum_{T\subseteq [Y],|T|=\ell} \hat{f_x}(T)\cdot \sum_{i\in X} x_i\cdot \alpha_{T\cup\{i\}}}\\
&\le  \E_{x\sim \{\pm 1\}^X}\sbra{  \sum_{T\subseteq [Y],|T|=\ell} \abs{\hat{f_x}(T)}\cdot \max_{T\subseteq [Y],|T|=\ell}\abs{\sum_{i\in X} x_i\cdot \alpha_{T\cup\{i\}}}}.
\end{split}\end{align}
We now use the fact that $f_x\in \cF$ due the restriction-closed property and that $\sum_{T\subseteq [Y],|T|=\ell} \abs{\hat{f_x}(T)}=L_{1,\ell}(f_x)\le L_{1,\ell}(\cF)$ to conclude that
\[ L(X)\le L_{1,\ell}(\cF) \cdot   \E_{x\sim \{\pm 1\}^X}\max_{T\subseteq [Y],|T|=\ell}\sbra{\abs{\sum_{i\in X} x_i\cdot \alpha_{T\cup\{i\}}}}.\]
Fix any $T\subseteq [Y]$ with $|T|=\ell$. We now observe that when $x\sim \{\pm 1\}^X$, the random variable $\sum_{i\in X} x_i\cdot \alpha_{T\cup\{i\}}$ is a sum of $|X|$ many independent $[-1,1]$-random variables where $|X|\le N$, hence,
\[\Pr\sbra{\abs{\sum_{i\in X} x_i\cdot \alpha_{T\cup\{i\}}}\ge t\sqrt{N}}\le \exp(-\Omega(t^2)).\]
We now set $t=\Theta(\sqrt{\log N}\cdot \sqrt{\log\binom{N}{\ell}})$ so that we can union bound over all $T\subseteq[Y]$ with $|T|=\ell$, to conclude that with probability at least $1-1/N$ over $x\sim \{\pm 1\}^X$, for all $T\subseteq [Y]$ with $|T|=\ell$,
\[\abs{\sum_{i\in X} x_i\cdot \alpha_{T\cup\{i\}}}\le O\pbra{\sqrt{N}\cdot \sqrt{\log\binom{N}{\ell}}\cdot \sqrt{\log N}}. \] 
For all $x\in \{\pm 1\}^X$, we have that the L.H.S. above is at most $N$. This, along with the above implies that 
\begin{align*} \E_{x\sim \{\pm 1\}^X}\sbra{\max_{T\subseteq [Y],|T|=\ell}\abs{\sum_{i\in X} x_i\cdot \alpha_{T\cup\{i\}}}} &\le \frac{1}{N}\cdot N + \sqrt{N}\cdot \sqrt{\log\binom{N}{\ell}}\cdot \sqrt{\log N} \\
&\le O\pbra{\sqrt{N \log\binom{N}{\ell}\log N}}+1.
\end{align*}
Finally, we observe that if $\ell=N$, then $L_{1,\ell+1}(\cF)=0$ and if $\ell<N$, then $\log\binom{N}{\ell}\ge 1$, and hence we can ignore the $+1$ factor in the R.H.S. by absorbing into the $O()$. This completes the proof.
\end{proof}

\subsection{Improved Bootstrapping for $\hBQP$ algorithms}
\label{sec:improved_bootstrapping}
\begin{theorem} Let $\cF$ be a restriction-closed family of boolean functions of degree at most $d$ on $\{\pm 1\}^{3N}$. Let $\gamma,\gamma'\in[-1,1]^{3N}$ and $\alpha(\gamma),\beta(\gamma,\gamma')$ be signs as in~\Cref{def:alpha}. Then, for any restriction $\rho$ on $[3N]$,
$L_{1,6}^{\beta(\gamma,\gamma')}(f|_\rho)\le \max_{\gamma''} (L_{1,3}^{\alpha(\gamma'')}(\cF))\cdot O(d^3 N)$, where the maximum is over $\gamma''\in [-1,1]^{3N}$.\label{thm:bootstrapping_hBQP}
\end{theorem}
Using the bound on $L_{1,3}^{\alpha(\gamma)}$ for $\hBQP$ algorithms, we obtain the following corollary. 
\begin{corollary}[\Cref{thm:main_theorem_hbqp}]
Let $\gamma,\gamma'\in[-1,1]^{3N}$ and $\beta(\gamma,\gamma')$ be signs as in~\Cref{def:alpha}. Let $f$ be a $d$-query $\hBQP$ algorithm on oracle of length $N$. Then, $L_{1,6}^{\beta(\gamma,\gamma')}(f)\le d^{10}\cdot N^{3/2}$.
\end{corollary}
\begin{proof}[Proof of~\Cref{thm:bootstrapping_hBQP}]
Our goal is to upper bound 
\[ L_{1,6}^{\beta(\gamma,\gamma')}(f)=\sum_{\substack{i_1< i_4\in A\\ i_2< i_5\in B\\ i_3< i_6\in C}} \beta(\gamma,\gamma')_{i_1,\ldots,i_6}\cdot \widehat{f}(\{i_1,\ldots,i_6\})\]
Let $A=[N],B=(N,2N],C=(2N,3N]$ as in~\Cref{def:alpha}. Partition these sets as $A=X_A\sqcup Y_A,B=X_B\sqcup Y_B,C=X_C\sqcup Y_C$, where for each element, we independently include it in the corresponding $X$ with probability $1/2$. Let $X=X_A\cup X_B\cup X_C$. Define 
\begin{equation} \label{eq:imp_bootstrap_0} L(X):=\sum_{\substack{i_1\in X_A,i_2\in X_B, i_3\in X_C \\ i_4\in Y_A,i_5\in Y_B,i_6\in Y_C }}\beta(\gamma,\gamma')_{i_1,\ldots,i_6} \widehat{f}(\{i_1,\ldots,i_6\}).\end{equation}
Observe that for any $i_1< i_4\in A, i_2< i_5\in B,i_3 < i_6\in C$, the probability over $X$ that this term contributes to $L(X)$ is precisely the probability that $\{i_1,i_2,i_3\}\subseteq X$ and $\{i_4,i_5,i_6\}\subseteq Y$ which is exactly $2^{-6}$. Therefore,
 \[ L_{1,6}^{\beta(\gamma,\gamma')}(f|_\rho)= \E_{X}\sbra{ L(X)}\cdot \frac{1}{2^6}.\]
It thus suffices to upper bound $L(X)$ for any partition $X$. As before, we will express $\widehat{f}(\{i_1,\ldots,i_6\})$ as $\E_{\substack{x\sim \{\pm 1\}^X\\ y\sim \{\pm1\}^Y }}\sbra{f(x,y)\prod_{t\in[6]}x_t}$, where the underlying distributions are uniform -- this follows from the definition of Fourier coefficients. Substituting this in~\Cref{eq:imp_bootstrap_0}, we see that
\begin{align}\begin{split}
L(X)&=\sum_{\substack{i_1\in X_A,i_2\in X_B, i_3\in X_C \\ i_4\in Y_A,i_5\in Y_B,i_6\in Y_C}}\beta(\gamma,\gamma')_{i_1,\ldots,i_6}  \E_{\substack{x\sim \{\pm 1\}^X\\ y\sim \{\pm1\}^Y }}\sbra{f(x,y)\prod_{t\in[6]}x_t} \\
&=\E_{x\sim \{\pm 1\}^X}\sbra{ \sum_{\substack{i_1\in X_A\\i_2\in X_B\\ i_3\in X_C} } \alpha(\gamma)_{i_1,i_2,i_3} \cdot x_{i_1}x_{i_2}x_{i_3}\cdot  \E_{y\sim \{\pm 1\}^Y}\sbra{\sum_{\substack{i_4\in Y_A\\i_5\in Y_B\\i_6\in Y_C} }\sbra{f(x,y)\cdot \alpha(\gamma')_{i_4,i_5,i_6}  \cdot y_{i_4}y_{i_5}y_{i_6}}}}
\end{split}\label{eq:imp_bootstrap_1}\end{align}
where in the last line, we used the fact that $\beta(\gamma,\gamma')_{i_1,\ldots,i_6}=\alpha(\gamma)_{i_1,i_2,i_3}\cdot \alpha(\gamma')_{i_4,i_5,i_6}$.  As before, we define a function $f_x$ such that $f_x(y):=f(x,y)$ for all $y\in \{\pm 1\}^Y$. Define a function $g(x)$ at $x\in \{\pm1 \}^X$ by
\begin{align}\label{eq:function_g}\begin{split}
g(x)&:=\E_{y\sim \{\pm 1\}^Y}\sbra{\sum_{\substack{i_4\in Y_A\\i_5\in Y_B\\i_6\in Y_C} }\sbra{f(x,y)\cdot \alpha(\gamma')_{i_4,i_5,i_6}  \cdot y_{i_4}y_{i_5}y_{i_6}}} \triangleq \sum_{\substack{i_4\in Y_A\\i_5\in Y_B\\i_6\in Y_C} }\alpha(\gamma')_{i_4,i_5,i_6}\cdot \hat{f_x}(\{i_4,i_5,i_6\})\end{split}
\end{align}
where the last equality follows by the definition of Fourier coefficients. Substituting~\Cref{eq:function_g} in~\Cref{eq:imp_bootstrap_1}, we obtain
\begin{equation} L(X)=\E_{x\sim \{\pm 1\}^X}\sbra{ \sum_{\substack{i_1\in X_A\\i_2\in X_B\\ i_3\in X_C} } \alpha(\gamma)_{i_1,i_2,i_3} \cdot x_{i_1}x_{i_2}x_{i_3}\cdot g(x)} = \sum_{\substack{i_1\in X_A\\i_2\in X_B\\ i_3\in X_C} } \alpha(\gamma)_{i_1,i_2,i_3} \cdot \hat{g}(\{i_1,i_2,i_3\}) \label{eq:imp_bootstrap_3}\end{equation}
where again, the last equality follows by the definition of Fourier coefficients. Furthermore, by defining $\gamma''$ to be equal to $\gamma$ in $X$ and 0 otherwise, we see that 
the R.H.S. of~\Cref{eq:imp_bootstrap_3} equals $L_{1,3}^{\alpha(\gamma'')}(g)$. Observe from~\Cref{eq:function_g} that $g(x)$ is a polynomial in $x$ of degree at most $d$, since $f(x,y)$ is of degree at most $d$.  We now use the level-3 bounds for bounded low-degree polynomials from~\cite{IRR+21} to conclude that 
\[ L_{1,3}^{\alpha(\gamma'')}(g) \le O(d^3 N)\cdot \max_{x\in\{\pm 1\}^X}\abs{g(x)}.\]
We will now show that $\abs{g(x)}$ is bounded for all $x\in \{\pm1\}^N$. Since $\cF$ is restriction closed, we have $f_x\in \cF$. Furthermore, setting $\gamma''$ to be equal to $\gamma'$ on $Y$ and 0 otherwise, and recalling~\Cref{eq:function_g}, we obtain that \[\abs{g(x)}=\abs{\sum_{\substack{i_4\in Y_A\\i_5\in Y_B\\i_6\in Y_C} }  \alpha(\gamma')_{i_4,i_5,i_6}  \cdot \widehat{f_x}(\{i_4,i_5,i_6\})}=\max(L_{1,3}^{\alpha(\gamma'')}(f_x),L_{1,3}^{-\alpha(\gamma'')}(f_x))\le \max_{\gamma''}(L_{1,3}^{\alpha(\gamma'')}(\cF)).\]
This completes the proof. 
\end{proof}

\section{Acknowledgements}
The author would like to thank Vishnu Iyer, Dale Jacobs, Natalie Parham, Makrand Sinha, Kewen Wu and Henry Yuen for very helpful comments and discussions. The author is also grateful to Francisco Escudero Gutierrez
and Miquel Saucedo Cuesta for
pointing out a critical bug in an earlier version of this paper that has since been corrected.

\bibliographystyle{alpha}
\bibliography{ref}

\appendix
\section{Appendix}

\subsection{Proof of~\Cref{lem:main_lemma_3}}
\label{sec:app_main_lemma_3}

\begin{proof}[Proof of~\Cref{lem:main_lemma_3}]
For $i_1\in[N]$, let $P^{(i_1)}$ be the $M\times M$ projection matrix $(\id-\ket{i_1}\bra{i_1})\otimes \id$ for $i_1\in[N]$ and let $Q^{(i_1)}=\id-P^{(i_1)}=\ket{i_1}\bra{i_1}\otimes \id$. Define $\tilP_t^{(i_1)}= U_t\cdot P^{(i_1)}$ for all $i_1\in[N],t\in[t^*-1,d-1]$ and $\tilP_d=U_d$. For $i_1,i_{d+1}\in[N]$, let $P^{(i_1,i_{d+1})}=P^{(i_1)}\cdot P^{(i_{d+1})}$, and let $\tilP_t^{(i_1,i_{d+1})}=U_t\cdot P^{(i_1,i_{d+1})}$ for $t\in [t^*-1,d-1]$. As before, we have
\begin{equation}\label{eq:Q_i_ortho_3}
Q^{(i')\dagger}\cdot Q^{(i)}=Q^{(i)\dagger}\cdot Q^{(i')}=0\text{ if }i\neq i',
\end{equation}
\begin{equation}\label{eq:Q_i_sum_to_one_3}
\sum_i Q^{(i)}=\sum_i Q^{(i)\dagger}\cdot Q^{(i)}=\sum_i Q^{(i)}\cdot Q^{(i)\dagger}=\id ,
\end{equation}
\begin{equation}\label{eq:Q_i_difference_3}
Q^{(i)}=\id -P^{(i)}.
\end{equation}
Observe that 
\[ U=\sum_{i_1,i_{d+1}\in[N]}Q^{(i_1)}\cdot  U_{[1,t^*-1)}\cdot \tilP_{[t^*-1,d-1]}^{(i_1,i_{d+1})}\cdot \tilP_d\cdot  Q^{(i_{d+1})}. \]
Define 
\begin{align*} U'&=\sum_{i_1,i_{d+1}\in[N]} Q^{(i_1)}\cdot  U_{[1,t^*-1)}\cdot \tilP_{[t^*-1,d-1]}^{(i_1)}\cdot \tilP_d\cdot  Q^{(i_{d+1})}\\
&=\sum_{i_1\in[N]} Q^{(i_1)}\cdot  U_{[1,t^*-1)}\cdot \tilP_{[t^*-1,d-1]}^{(i_1)}\cdot \tilP_d \tag{since $\sum_{i_{d+1}\in[N]}Q^{(i_{d+1})}=\id.$}\end{align*}
Observe that $U'$ is precisely the matrix we would get by an application of~\Cref{lem:main_lemma_2} to the matrices $U_{[1,t^*-1)},U_{t^*-1},U_{t^*},U_{t^*+1},\ldots,U_d$ with $\gamma=1$. By~\Cref{lem:main_lemma_2}, we have 
\[\|U'\|_\frob \le O(d)\cdot \|U_{t^*}\|_\frob.\] 
We now compare $U$ and $U'$. We observe that the only difference between $U$ and $U'$ comes $\tilP^{(i_1,i_{d+1})}_{[t^*-1,d-1]}$ versus $\tilP^{(i_1)}_{[t^*-1,d-1]}$. As before, we use the telescoping expansion to obtain
%\[\prod_{t\in[t^*,d]}\indi\sbra{i_t\neq i_{d+1}}= 1- \sum_{t'\in [t^*,d]}\indi\sbra{i_{d+1}=i_{t'}}\cdot \prod_{t''\in[t'+1,d]}\indi\sbra{i_{t''}\neq i_{d+1}}\]
%Recall that $\tilP^{(i_1,i_{d+1})}_{[t^*-1,d-1]}$, in addition to imposing $\tilP^{(i_1)}_{[t^*-1,d-1]}$, also imposes $i_t\neq i_{d+1}$ for all $t\in[t^*,d]$. This implies that 
\begin{align*}
&U'-U\\
&=\sum_{t'\in[t^*,d]}\sum_{i_1,i_{d+1}\in[N]}  Q^{(i_1)} \cdot U_{[1,t^*-1)}\cdot \tilP^{(i_1)}_{[t^*-1,t'-1)}\cdot\pbra{\tilP^{(i_1)}_{t'-1}-\tilP^{(i_1,i_{d+1})}_{t'-1}}\cdot \tilP_{[t',d-1]}^{(i_1,i_{d+1})}\cdot \tilP^{(i_1)}_d\cdot Q^{(i_{d+1})}\\
&=\sum_{t'\in[t^*,d]}\sum_{i_1,i_{d+1}\in[N]}  Q^{(i_1)} \cdot U_{[1,t^*-1)}\cdot \tilP^{(i_1)}_{[t^*-1,t'-1)} \cdot \pbra{ \tilP^{(i_1)}_{t'-1}\cdot Q^{(i_{d+1})}}\cdot \tilP_{[t',d-1]}^{(i_1,i_{d+1})}\cdot \tilP_d^{(i_1)}\cdot Q^{(i_{d+1})}\\
\tag{by definition and~\Cref{eq:Q_i_difference_3}}\\
&=\sum_{t'\in[t^*,d]}W_{t'_1}, \text{ where}
\end{align*}
\[W_{t_1'}=\sum_{i_1\in[N]}  Q^{(i_1)} \cdot U_{[1,t^*-1)}\cdot  \tilP^{(i_1)}_{[t^*-1,t'-1]} \cdot \underbrace{\sum_{i_{d+1}\in[N]}Q^{(i_{d+1})}\cdot \tilP_{[t',d-1]}^{(i_1,i_{d+1})}\cdot \tilP^{(i_1)}_d \cdot Q^{(i_{d+1})}}_{Z_{t'_1}^{(i_1)}}.\]
We will show that for all $t'\in[t^*,d]$, we have $\|W_{t'}\|_\frob \le O(1)\cdot \|U_{t^*}\|_\frob$ and this would complete the proof. Suppose $t'=t^*$, then the corresponding term is
\[W_{t^*}:= \sum_{i_1\in[N]} Q^{(i_1)}\cdot U_{[1,t^*-1)} \cdot \tilP_{t^*-1}^{(i_1)} \cdot \sum_{i_{d+1}\in[N]} Q^{(i_{d+1})} \cdot \tilP^{(i_1,i_{d+1})}_{[t^*,d-1]}\cdot \tilP_d^{(i_1)} \cdot Q^{(i_{d+1})}\]
Observe that 
\begin{align*}\|W_t\|^2_\frob &= \sum_{i_1,i_{d+1}\in[N]} \vabs{Q^{(i_1)}\cdot U_{[1,t^*-1)} \cdot \tilP_{t^*-1}^{(i_1)} \cdot Q^{(i_{d+1})} \cdot \tilP^{(i_1,i_{d+1})}_{[t^*,d-1]}\cdot \tilP_d^{(i_1)}\cdot Q^{(i_{d+1})}}_\frob^2 \tag{\Cref{eq:Q_i_ortho_3} applied to $Q^{(i_1)},Q^{(i_{d+1})}$} \\
&\le \sum_{\substack{i_1\in[N]\\i_{d+1}\in[N]}} \vabs{Q^{(i_1)}\cdot U_{[1,t^*-1)} \cdot \tilP_{t^*-1}^{(i_1)} \cdot  Q^{(i_{d+1})}\cdot   U_{t^*}}_\frob^2 \tag{by~\Cref{fact:frob_op} and since the rest of the terms have $\|\cdot\|_\op \le 1$}\\
&= \sum_{i_{d+1}\in[N]} \vabs{\sum_{i_1\in [N]}Q^{(i_1)}\cdot U_{[1,t^*-1)} \cdot \tilP_{t^*-1}^{(i_1)} \cdot  Q^{(i_{d+1})}\cdot   U_{t^*}}_\frob^2\tag{\Cref{eq:Q_i_ortho_3} applied to $Q^{(i_1)}$} \\
&\le \sum_{i_{d+1}\in[N]} \vabs{\sum_{i_1\in[N]}Q^{(i_1)}\cdot U_{[1,t^*-1)} \cdot \tilP_{t^*-1}^{(i_1)}}_\op^2 \cdot  \vabs{Q^{(i_{d+1})}\cdot  U_{t^*}}_\frob^2 \tag{by~\Cref{fact:frob_op}}\\
&\le 4\sum_{i_{d+1}\in[N]} \vabs{Q^{(i_{d+1})}\cdot  U_{t^*}}_\frob^2\tag{by~\Cref{lem:main_lemma_2} applied to $U_{[1,t^*-1]}$}\\
&= 4\sum_{i_{d+1}\in[N]} \Tr\pbra{Q^{(i_{d+1})}\cdot  U_{t^*}  \cdot U_{t^*}^\dagger \cdot Q^{(i_{d+1})\dagger}  }\\
&=4\Tr(U_{t^*}\cdot U_{t^*}^\dagger)=4\|U_{t^*}\|_\frob^2\tag{by~\Cref{eq:Q_i_sum_to_one_3} applied to $Q^{(i_{d+1})}$}.
\end{align*}
We will now show that for $t'>t^*,$ $\|W_{t'}\|_\frob \le O(d)\cdot \|U_{t^*}\|_\frob$. To do this, we first claim that $\|Z_{t_1'}^{(i_1)}\|_\op \le 1$. To see this, observe that
\begin{align*}
   Z_{t_1'}^{(i_1)}\cdot Z_{t_1'}^{(i_1)\dagger} & \triangleq \pbra{\sum_{i_{d+1}} Q^{(i_{d+1})}\cdot \tilP_{[t',d-1]}^{(i_1,i_{d+1})}\cdot \tilP^{(i_1)}_d \cdot Q^{(i_{d+1})} }\cdot \pbra{\sum_{i_{d+1}} Q^{(i_{d+1})}\cdot \tilP_{[t',d-1]}^{(i_1,i_{d+1})}\cdot \tilP^{(i_1)}_d \cdot Q^{(i_{d+1})}}^\dagger\\
   &=\sum_{i_{d+1}} Q^{(i_{d+1})}\cdot \tilP_{[t',d-1]}^{(i_1,i_{d+1})}\cdot \tilP^{(i_1)}_d \cdot Q^{(i_{d+1})}\cdot Q^{(i_{d+1})\dagger}\cdot \tilP^{(i_1)\dagger}_d \cdot \tilP_{[t',d-1]}^{(i_1,i_{d+1})\dagger}\cdot Q^{(i_{d+1})\dagger}\tag{\Cref{eq:Q_i_ortho_3}}\\
   &\preceq \sum_{i_{d+1}} Q^{(i_{d+1})}\cdot Q^{(i_{d+1})\dagger}=\id. \tag{since $\vabs{  \tilP^{(i_1,i_{d+1})}_{[t',d-1]}\cdot \tilP_d^{(i_1)}\cdot Q^{(i_{d+1})}}_\op\le 1$.}
\end{align*}
Therefore, we obtain 
\begin{align*}
\|W_{t'}\|_\frob^2 &\triangleq  \vabs{\sum_{i_1\in [N]}  Q^{(i_1)}\cdot U_{[1,t^*-1)}\cdot \tilP^{(i_1)}_{[t^*-1,t'-1]} \cdot Z_{t_1'}^{(i_1)}}^2_\frob\\
&=  \sum_{i_1\in [N]} \vabs{ Q^{(i_1)}\cdot  U_{[1,t^*-1)}\cdot \tilP^{(i_1)}_{[t^*-1,t'_1-1]}\cdot Z_{t_1'}^{(i_1)}}_\frob^2 \tag{\Cref{eq:Q_i_ortho_3}}\\
&\le \sum_{i_1\in [N]} \vabs{ Q^{(i_1)} \cdot U_{[1,t^*-1)}\cdot \tilP^{(i_1)}_{[t^*-1,t'_1-1]}}_\frob^2  \tag{since $\|Z_{t_1'}^{(i_1)}\|_\op \le 1$}\\
&= \vabs{ \sum_{i_1\in [N]} Q^{(i_1)}\cdot  U_{[1,t^*-1)}\cdot \tilP^{(i_1)}_{[t^*-1,t'_1-1]}}_\frob^2 \tag{\Cref{eq:Q_i_ortho_3}} \\
&\le O(1)\cdot \|U_{t^*}\|_\frob^2 \tag{by an application of~\Cref{lem:main_lemma_2} to matrices $U_{[1,t^*-1]},U_{t^*},\ldots, U_{t'_1-1}$ with $t_1=1$.}
\end{align*}
Thus, we get $\|U\|_\frob \le O(d)\cdot \|U_{t^*}\|_\frob$. This completes the proof. \end{proof}

\subsection{Proof of~\Cref{lem:main_lemma_4}}

\begin{lemma}\label{lemma:sum_frob_norm}
For $k\in [K]$ let $X_k\in \C^{M\times N},Y_k\in C^{N\times M}$ and $C\in \C^{N\times N}$. Suppose for a constant $a>0$, we have $\sum_{k\in[K]} X_k^\dagger X_k \preceq a^2 \cdot \id$, $\sum_{k\in [K]} Y_k^\dagger Y_k \preceq  \id$. Then,
\[ \vabs{\sum_{k\in [K]} X_k C Y_k}_\frob \le a\cdot \|C\|_\frob.\]
\end{lemma}
\begin{proof}[Proof of~\Cref{lemma:sum_frob_norm}]
Consider
\begin{align*}
    \vabs{ \sum_k X_k C Y_k}_\frob^2 &=\vabs{\begin{bmatrix} X_1 C &\ldots  &X_k C\end{bmatrix}\cdot \begin{bmatrix} Y_1 \\ \vdots \\ Y_k\end{bmatrix}}_\frob^2 \le \vabs{\begin{bmatrix} X_1 C &\ldots  &X_k C\end{bmatrix}}_\frob^2\cdot\vabs{ \begin{bmatrix} Y_1 \\ \vdots \\ Y_k\end{bmatrix}}_\op^2
\end{align*}
where the last inequality follows by~\Cref{fact:frob_op}. Firstly, we observe that 
\[\vabs{\begin{bmatrix} Y_1 \\ \vdots \\ Y_k\end{bmatrix}}_\op^2 =\vabs{\begin{bmatrix} Y_1^\dagger & \ldots & Y_k^\dagger\end{bmatrix}\cdot \begin{bmatrix} Y_1 \\ \vdots \\ Y_k\end{bmatrix}}_\op =\vabs{\sum_{k\in [K]}Y_k^\dagger Y_k }_\op\le 1.\]
Next, we observe that
\begin{align*}
\vabs{\begin{bmatrix} X_1 C &\ldots  &X_k C\end{bmatrix}}^2_\frob&=\Tr\pbra{\sum_{k\in[K]}X_kCC^\dagger X_k^\dagger}=\Tr\pbra{\sum_{k\in[K]}X_k^\dagger X_k\cdot CC^\dagger }\le a^2 \cdot \|C\|^2_\frob.
\end{align*}
\end{proof}

\begin{fact}\label{fact:sum_dagger} For matrices $A_1,\ldots,A_K$ we have $\pbra{\sum_{k=1}^KA_k}^\dagger\cdot \pbra{\sum_{k=1}^K A_k}\preceq K\cdot \sum_{k=1}^K A_k^\dagger A_k$.
\end{fact}
\begin{proof}[Proof of~\Cref{fact:sum_dagger}]
Observe that for all pairs $k<k'\in[K]$, we have  
\[ (A_k-A_{k'})^\dagger (A_k-A_{k'})\succeq 0 \implies A_k^\dagger A_{k'} + A_{k'}^\dagger A_k\preceq A_k^\dagger A_k  + A_{k'}^\dagger A_{k'}.\]
Summing this over all $k<k'\in [K]$ and adding $\sum_{k=1}^K A_k^\dagger A_k$ on both sides gives the desired inequality.
\end{proof}
\begin{lemma} For $i,j\in[N]$, let $F^{(i,j)}\in \C^{M\times M}$ be a matrix. Furthermore, assume that for some parameter $K\in \mathbb{N}$, there exist matrices $F^{(i,j)}_{1},\ldots,F^{(i,j)}_{K}$ such that $F^{(i,j)}=\sum_{k=1}^K F^{(i,j)}_k$ and for all $k\in [K]$, we have 
\begin{align*}
\sum_{i,j}F^{(i,j)\dagger}_k\cdot  F^{(i,j)}_{k}\preceq \id ,\quad &\sum_{i}\pbra{\sum_j F^{(i,j)\dagger}_k}\cdot \pbra{ \sum_j F^{(i,j)}_k}\preceq \id\\
\sum_{j}\pbra{\sum_i F^{(i,j)\dagger}_k}\cdot \pbra{ \sum_i F^{(i,j)}_k}\preceq \id,\quad  &\pbra{\sum_{i,j} F^{(i,j)\dagger}_k}\cdot \pbra{ \sum_{i,j} F^{(i,j)}_k}\preceq \id.\\
\end{align*}
Then, the same inequalities hold for $F^{(i,j)}$ with the R.H.S. above replaced by $K^2\cdot \id$. \label{lemma:sum_op_norm}  
\end{lemma}
\begin{proof}[Proof of~\Cref{lemma:sum_op_norm}]
We apply~\Cref{fact:sum_dagger} to the matrices $F^{(i,j)}=\sum_{k=1}^K F^{(i,j)}_k$ and then sum over $i,j$ to obtain
\[\sum_{i,j} F^{(i,j)\dagger}\cdot F^{(i,j)} \preceq K \cdot \sum_{i,j}  \sum_{k=1}^K F^{(i,j)\dagger}_k\cdot F^{(i,j)}_k\preceq K^2\cdot \id\]
We apply~\Cref{fact:sum_dagger} to the matrices $\sum_j F^{(i,j)}=\sum_{k=1}^K (\sum_j F^{(i,j)}_k)$ and then sum over $i$ to obtain
\[\sum_{i} \pbra{\sum_j F^{(i,j)\dagger}}\cdot\pbra{\sum_j F^{(i,j)}} \preceq K\cdot \sum_{i}  \sum_{k=1}^K \pbra{\sum_j F^{(i,j)\dagger}_k}\cdot\pbra{\sum_j F^{(i,j)}_k}\preceq K^2\cdot \id\]
The same calculation works when the roles of $i$ and $j$ are swapped. We apply~\Cref{fact:sum_dagger} to the matrices 
$\sum_{i,j} F^{(i,j)}=\sum_{k=1}^K (\sum_{i,j}F^{(i,j)}_k)$ to obtain
\[ \pbra{\sum_{i,j} F^{(i,j)\dagger}}\cdot\pbra{\sum_{i,j} F^{(i,j)}} \preceq K\cdot  \sum_{k=1}^K \pbra{\sum_{i,j} F^{(i,j)\dagger}_k}\cdot\pbra{\sum_{i,j} F^{(i,j)}_k}\preceq K^2\cdot \id.\]
\end{proof}
\begin{figure}
    \centering

\tikzset{every picture/.style={line width=0.75pt}} %set default line width to 0.75pt        

\begin{tikzpicture}[x=0.75pt,y=0.75pt,yscale=-1,xscale=1]
%uncomment if require: \path (0,300); %set diagram left start at 0, and has height of 300

%Straight Lines [id:da6624828832624435] 
\draw [color={rgb, 255:red, 255; green, 0; blue, 0 }  ,draw opacity=1 ][line width=3.75]    (373,64) -- (373,97) ;
%Straight Lines [id:da8585028969417052] 
\draw [color={rgb, 255:red, 0; green, 93; blue, 235 }  ,draw opacity=1 ][line width=3.75]    (405,64) -- (405,97) ;
%Straight Lines [id:da9724084283114508] 
\draw [color={rgb, 255:red, 0; green, 93; blue, 235 }  ,draw opacity=1 ][line width=3.75]    (437,65) -- (437,98) ;
%Straight Lines [id:da8711643614681429] 
\draw [color={rgb, 255:red, 0; green, 93; blue, 235 }  ,draw opacity=1 ][line width=3.75]    (468,64) -- (468,97) ;
%Straight Lines [id:da5143662271228638] 
\draw [color={rgb, 255:red, 0; green, 93; blue, 235 }  ,draw opacity=1 ][line width=3.75]    (502,64) -- (502,97) ;
%Shape: Rectangle [id:dp0831698256093164] 
\draw   (322.64,65.2) -- (566.14,65.2) -- (566.14,127.59) -- (322.64,127.59) -- cycle ;
%Straight Lines [id:da2935686940720631] 
%special line
%\draw    (322.14,96.39) -- (567.14,96.39) ;
%Straight Lines [id:da039267016759568896] 
\draw [color={rgb, 255:red, 0; green, 93; blue, 235 }  ,draw opacity=1 ][line width=3.75]    (535,66) -- (535,99) ;
%Straight Lines [id:da9252635521734356] 
\draw [color={rgb, 255:red, 245; green, 166; blue, 35 }  ,draw opacity=1 ][line width=3.75]    (437,98) -- (437.14,129.39) ;
%Straight Lines [id:da6703283319982858] 
\draw [color={rgb, 255:red, 126; green, 211; blue, 33 }  ,draw opacity=1 ][line width=3.75]    (468,97) -- (468,127) ;
%Straight Lines [id:da5112124994313495] 
\draw [color={rgb, 255:red, 126; green, 211; blue, 33 }  ,draw opacity=1 ][line width=3.75]    (502,97) -- (502,127) ;
%Straight Lines [id:da16539109292773746] 
\draw [color={rgb, 255:red, 126; green, 211; blue, 33 }  ,draw opacity=1 ][line width=3.75]    (535,99) -- (535,129) ;
%Shape: Grid [id:dp28515430982035816] 
\draw  [draw opacity=0] (62,74) -- (261.14,74) -- (261.14,107.57) -- (62,107.57) -- cycle ; \draw   (62,74) -- (62,107.57)(95,74) -- (95,107.57)(128,74) -- (128,107.57)(161,74) -- (161,107.57)(194,74) -- (194,107.57)(227,74) -- (227,107.57)(260,74) -- (260,107.57) ; \draw   (62,74) -- (261.14,74)(62,107) -- (261.14,107) ; \draw    ;
%Straight Lines [id:da6833486312101934] 
\draw    (94.94,56.44) -- (94.99,72) ;
\draw [shift={(95,74)}, rotate = 269.79] [color={rgb, 255:red, 0; green, 0; blue, 0 }  ][line width=0.75]    (10.93,-3.29) .. controls (6.95,-1.4) and (3.31,-0.3) .. (0,0) .. controls (3.31,0.3) and (6.95,1.4) .. (10.93,3.29)   ;
%Straight Lines [id:da04187982778416477] 
\draw    (127.94,56.44) -- (127.99,72) ;
\draw [shift={(128,74)}, rotate = 269.79] [color={rgb, 255:red, 0; green, 0; blue, 0 }  ][line width=0.75]    (10.93,-3.29) .. controls (6.95,-1.4) and (3.31,-0.3) .. (0,0) .. controls (3.31,0.3) and (6.95,1.4) .. (10.93,3.29)   ;
%Straight Lines [id:da569134045719505] 
\draw    (160.94,56.44) -- (160.99,72) ;
\draw [shift={(161,74)}, rotate = 269.79] [color={rgb, 255:red, 0; green, 0; blue, 0 }  ][line width=0.75]    (10.93,-3.29) .. controls (6.95,-1.4) and (3.31,-0.3) .. (0,0) .. controls (3.31,0.3) and (6.95,1.4) .. (10.93,3.29)   ;
%Straight Lines [id:da07086260285282397] 
\draw    (193.94,56.44) -- (193.99,72) ;
\draw [shift={(194,74)}, rotate = 269.79] [color={rgb, 255:red, 0; green, 0; blue, 0 }  ][line width=0.75]    (10.93,-3.29) .. controls (6.95,-1.4) and (3.31,-0.3) .. (0,0) .. controls (3.31,0.3) and (6.95,1.4) .. (10.93,3.29)   ;
%Straight Lines [id:da6134257776505063] 
\draw    (226.94,56.44) -- (226.99,72) ;
\draw [shift={(227,74)}, rotate = 269.79] [color={rgb, 255:red, 0; green, 0; blue, 0 }  ][line width=0.75]    (10.93,-3.29) .. controls (6.95,-1.4) and (3.31,-0.3) .. (0,0) .. controls (3.31,0.3) and (6.95,1.4) .. (10.93,3.29)   ;
%Straight Lines [id:da8155475792740398] 
\draw    (160.94,133.44) -- (161,109) ;
\draw [shift={(161,107)}, rotate = 90.14] [color={rgb, 255:red, 0; green, 0; blue, 0 }  ][line width=0.75]    (10.93,-3.29) .. controls (6.95,-1.4) and (3.31,-0.3) .. (0,0) .. controls (3.31,0.3) and (6.95,1.4) .. (10.93,3.29)   ;
%Straight Lines [id:da3814991211400799] 
\draw    (193.94,133.44) -- (194,109) ;
\draw [shift={(194,107)}, rotate = 90.14] [color={rgb, 255:red, 0; green, 0; blue, 0 }  ][line width=0.75]    (10.93,-3.29) .. controls (6.95,-1.4) and (3.31,-0.3) .. (0,0) .. controls (3.31,0.3) and (6.95,1.4) .. (10.93,3.29)   ;
%Straight Lines [id:da607636477658741] 
\draw    (226.94,133.44) -- (227,109) ;
\draw [shift={(227,107)}, rotate = 90.14] [color={rgb, 255:red, 0; green, 0; blue, 0 }  ][line width=0.75]    (10.93,-3.29) .. controls (6.95,-1.4) and (3.31,-0.3) .. (0,0) .. controls (3.31,0.3) and (6.95,1.4) .. (10.93,3.29)   ;

% Text Node
\draw (366,41.4) node [anchor=north west][inner sep=0.75pt]    {$ \begin{array}{l}
i_{t_{1}}\\
\end{array}$};
% Text Node
\draw (396,40.4) node [anchor=north west][inner sep=0.75pt]    {$ \begin{array}{l}
i_{t_{1}}\\
\end{array}$};
% Text Node
\draw (429,40.4) node [anchor=north west][inner sep=0.75pt]    {$ \begin{array}{l}
i_{t_{1}}\\
\end{array}$};
% Text Node
\draw (462,40.4) node [anchor=north west][inner sep=0.75pt]    {$ \begin{array}{l}
i_{t_{1}}\\
\end{array}$};
% Text Node
\draw (495,40.4) node [anchor=north west][inner sep=0.75pt]    {$ \begin{array}{l}
i_{t_{1}}\\
\end{array}$};
% Text Node
\draw (436,131.4) node [anchor=north west][inner sep=0.75pt]    {$ \begin{array}{l}
i_{t_{2}}\\
\end{array}$};
% Text Node
\draw (467,131.4) node [anchor=north west][inner sep=0.75pt]    {$ \begin{array}{l}
i_{t_{2}}\\
\end{array}$};
% Text Node
\draw (495,133.4) node [anchor=north west][inner sep=0.75pt]    {$ \begin{array}{l}
i_{t_{2}}\\
\end{array}$};
% Text Node
\draw (528,42.4) node [anchor=north west][inner sep=0.75pt]    {$ \begin{array}{l}
i_{t_{1}}\\
\end{array}$};
% Text Node
\draw (527,132.4) node [anchor=north west][inner sep=0.75pt]    {$ \begin{array}{l}
i_{t_{2}}\\
\end{array}$};
% Text Node
\draw (71,82.4) node [anchor=north west][inner sep=0.75pt]  [font=\small]  {$U_{1}$};
% Text Node
\draw (102,82.4) node [anchor=north west][inner sep=0.75pt]  [font=\small]  {$U_{2}$};
% Text Node
\draw (234,83.4) node [anchor=north west][inner sep=0.75pt]  [font=\small]  {$U_{t}$};
% Text Node
\draw (195,82.4) node [anchor=north west][inner sep=0.75pt]  [font=\small]  {$U_{t-1}$};
% Text Node
\draw (86,31.4) node [anchor=north west][inner sep=0.75pt]    {$ \begin{array}{l}
\mathnormal{Q^{( i_{t_{1}})}}\\
\end{array}$};
% Text Node
\draw (86,109.4) node [anchor=north west][inner sep=0.75pt]    {$t_{1}$};
% Text Node
\draw (124,33.4) node [anchor=north west][inner sep=0.75pt]    {$ \begin{array}{l}
P^{( i_{t_{1}})}\\
\end{array}$};
% Text Node
\draw (158,33.4) node [anchor=north west][inner sep=0.75pt]    {$ \begin{array}{l}
P^{( i_{t_{1}})}\\
\end{array}$};
% Text Node
\draw (193,33.4) node [anchor=north west][inner sep=0.75pt]    {$ \begin{array}{l}
P^{( i_{t_{1}})}\\
\end{array}$};
% Text Node
\draw (227,34.4) node [anchor=north west][inner sep=0.75pt]    {$ \begin{array}{l}
P^{( i_{t_{1}})}\\
\end{array}$};
% Text Node
\draw (271,82.4) node [anchor=north west][inner sep=0.75pt]    {$=$};
% Text Node
\draw (145,109.4) node [anchor=north west][inner sep=0.75pt]    {$t_{2}$};
% Text Node
\draw (150,128.4) node [anchor=north west][inner sep=0.75pt]    {$ \begin{array}{l}
\mathnormal{Q^{( i_{t_{2}})}}\\
\end{array}$};
% Text Node
\draw (187,131.4) node [anchor=north west][inner sep=0.75pt]    {$ \begin{array}{l}
P^{( i_{t_{2}})}\\
\end{array}$};
% Text Node
\draw (220,133.4) node [anchor=north west][inner sep=0.75pt]    {$ \begin{array}{l}
P^{( i_{t_{2}})}\\
\end{array}$};
% Text Node
\draw (375,79.9) node [anchor=north west][inner sep=0.75pt]    {$t_{1}$};
% Text Node
\draw (440,109.4) node [anchor=north west][inner sep=0.75pt]    {$t_{2}$};

\end{tikzpicture}

\caption{We use red, blue, orange and green lines to denote insertions of $Q^{(i_{t_1})}$,$P^{(i_{t_1})}$,$Q^{(i_{t_2})}$ and $P^{(i_{t_2})}$ respectively.}
    \label{fig_2:definition_hBQP}
\end{figure}

\begin{proof}[Proof of~\Cref{lem:main_lemma_4}]
\label{sec:proof_main_lemma_4}
Define $M\times M$ matrices $P^{(i)},Q^{(i)}$ as follows for $i\in[N]$.
\[ 
P^{(i)}=(\id-2\ket{i}\bra{i})\otimes \id \]
\[ Q^{(i)}=\ket{i}\bra{i}\otimes \id\]
Observe that $\|Q^{(i)}\|_\op,\|P^{(i)}\|_\op \le 1$ for all $i\in[N]$, furthermore, 
\begin{equation}
Q^{(i)}=\frac{1}{2}\pbra{\id -P^{(i)}}\label{eq:Q_i_difference_2}.
\end{equation}
Secondly, the $Q^{(i)}$ are orthogonal, i.e., for all $i,i'\in[N]$, we have 
\begin{equation}
Q^{(i)}\cdot Q^{(i')\dagger}=Q^{(i')\dagger}\cdot Q^{(i)}=0\quad \text{if }i\neq i'\label{eq:Q_i_ortho_2}.
\end{equation}
Finally, we have
\begin{equation}\sum_i Q(i)=\sum_i
Q^{(i)\dagger}\cdot Q^{(i)}=\sum_i Q^{(i)}\cdot Q^{(i)\dagger}=\id.\label{eq:Q_i_sum_to_one}
\end{equation}
Let $P^{(i_{t_1},i_{t_2})}=P^{(i_{t_1})}\cdot P^{(i_{t_2})}$. Let $\tilP_t^{(i_{t_1})}:= U_t\cdot P^{(i_{t_1})} $ if $t\in [t_1,d-1]$ and $\tilP_t^{(i_{t_1},i_{t_2})}:= U_t\cdot P^{(i_{t_1},i_{t_2})} $ for $t\in [t_2,d-1]$. Let $\tilP^{(i_{t_1})}_d=\tilP^{(i_{t_1},i_{t_2})}_d=U_d$. Observe that 
\begin{align*}U&=U_{[1,t_1)}\cdot \sum_{i_{t_1},i_{t_2}\in [N]} Q^{(i_{t_1})}\cdot \tilP^{(i_{t_1})}_{[{t_1},t_2)}
\cdot Q^{(i_{t_2})}\cdot  \tilP_{[t_2,d]}^{(i_{t_1},i_{t_2})}
\end{align*}
\begin{figure}
    \centering

\tikzset{every picture/.style={line width=0.75pt}} %set default line width to 0.75pt        

\begin{tikzpicture}[x=0.75pt,y=0.75pt,yscale=-1,xscale=1,scale=0.9]
%uncomment if require: \path (0,300); %set diagram left start at 0, and has height of 300

%Straight Lines [id:da3895954790922338] 
\draw [color={rgb, 255:red, 255; green, 0; blue, 0 }  ,draw opacity=1 ][line width=3.75]    (66,51) -- (66,84) ;
%Straight Lines [id:da891943781154098] 
\draw [color={rgb, 255:red, 0; green, 93; blue, 235 }  ,draw opacity=1 ][line width=3.75]    (98,51) -- (98,84) ;
%Straight Lines [id:da8351114230909333] 
\draw [color={rgb, 255:red, 0; green, 93; blue, 235 }  ,draw opacity=1 ][line width=3.75]    (130,52) -- (130,85) ;
%Straight Lines [id:da1825620824560784] 
\draw [color={rgb, 255:red, 0; green, 93; blue, 235 }  ,draw opacity=1 ][line width=3.75]    (161,51) -- (161,84) ;
%Straight Lines [id:da4825803055940139] 
\draw [color={rgb, 255:red, 0; green, 93; blue, 235 }  ,draw opacity=1 ][line width=3.75]    (195,51) -- (195,84) ;
%Shape: Rectangle [id:dp6264675592590424] 
\draw   (15.64,52.2) -- (287.14,52.2) -- (287.14,114.59) -- (15.64,114.59) -- cycle ;
%Straight Lines [id:da9189624408845729] 
%special line
%\draw    (15.14,83.39) -- (287.14,83.39) ;
%Straight Lines [id:da24561282769017256] 
\draw [color={rgb, 255:red, 0; green, 93; blue, 235 }  ,draw opacity=1 ][line width=3.75]    (228,52) -- (228,85) ;
%Straight Lines [id:da21678508164076593] 
\draw [color={rgb, 255:red, 245; green, 166; blue, 35 }  ,draw opacity=1 ][line width=3.75]    (130,84) -- (130.14,115.39) ;
%Straight Lines [id:da03692538045878013] 
\draw [color={rgb, 255:red, 126; green, 211; blue, 33 }  ,draw opacity=1 ][line width=3.75]    (161,84) -- (161,114) ;
%Straight Lines [id:da4894109546882274] 
\draw [color={rgb, 255:red, 126; green, 211; blue, 33 }  ,draw opacity=1 ][line width=3.75]    (195,84) -- (195,114) ;
%Straight Lines [id:da9927661254862065] 
\draw [color={rgb, 255:red, 126; green, 211; blue, 33 }  ,draw opacity=1 ][line width=3.75]    (228,84) -- (228,114) ;
%Straight Lines [id:da5922738569465366] 
\draw [color={rgb, 255:red, 255; green, 0; blue, 0 }  ,draw opacity=1 ][line width=3.75]    (392,51) -- (392,84) ;
%Straight Lines [id:da7650496153857768] 
\draw [color={rgb, 255:red, 0; green, 93; blue, 235 }  ,draw opacity=1 ][line width=3.75]    (424,51) -- (424,84) ;
%Straight Lines [id:da5379502016104557] 
\draw [color={rgb, 255:red, 0; green, 93; blue, 235 }  ,draw opacity=1 ][line width=3.75]    (456,52) -- (456,85) ;
%Straight Lines [id:da09983804729117574] 
\draw [color={rgb, 255:red, 0; green, 93; blue, 235 }  ,draw opacity=1 ][line width=3.75]    (487,51) -- (487,84) ;
%Straight Lines [id:da1900498018906206] 
\draw [color={rgb, 255:red, 0; green, 93; blue, 235 }  ,draw opacity=1 ][line width=3.75]    (521,51) -- (521,84) ;
%Shape: Rectangle [id:dp4216003590570323] 
\draw   (340.94,52.5) -- (523.19,52.5) -- (523.19,114.89) -- (340.94,114.89) -- cycle ;
%Straight Lines [id:da845132123772993] 
%special line
%\draw    (341.14,83.39) -- (521,84) ;
%Straight Lines [id:da2862106998845254] 
\draw [color={rgb, 255:red, 245; green, 166; blue, 35 }  ,draw opacity=1 ][line width=3.75]    (456,84) -- (456.14,115.39) ;
%Straight Lines [id:da4241819520082942] 
\draw [color={rgb, 255:red, 126; green, 211; blue, 33 }  ,draw opacity=1 ][line width=3.75]    (487,84) -- (487,114) ;
%Straight Lines [id:da21296472389247734] 
\draw [color={rgb, 255:red, 126; green, 211; blue, 33 }  ,draw opacity=1 ][line width=3.75]    (521,84) -- (521,114) ;
%Shape: Rectangle [id:dp8428592981631211] 
\draw   (582.41,52.61) -- (642.14,52.61) -- (642.14,115) -- (582.41,115) -- cycle ;
%Straight Lines [id:da04620898710734134] 
%special line
%\draw    (581.41,83.5) -- (643.14,84.11) ;
%Straight Lines [id:da6507309197662187] 
\draw [color={rgb, 255:red, 126; green, 211; blue, 33 }  ,draw opacity=1 ][line width=3.75]    (583.41,83.5) -- (583.41,113.5) ;
%Straight Lines [id:da7348018028277038] 
\draw [color={rgb, 255:red, 126; green, 211; blue, 33 }  ,draw opacity=1 ][line width=3.75]    (612.27,83.8) -- (612.27,113.8) ;
%Straight Lines [id:da10498826030035657] 
\draw [color={rgb, 255:red, 0; green, 93; blue, 235 }  ,draw opacity=1 ][line width=3.75]    (583.41,52.61) -- (583.41,85.61) ;
%Straight Lines [id:da9714307763129059] 
\draw [color={rgb, 255:red, 0; green, 93; blue, 235 }  ,draw opacity=1 ][line width=3.75]    (612.27,50.8) -- (612.27,83.8) ;
%Shape: Rectangle [id:dp5018051940637726] 
\draw   (536.41,51.61) -- (570.14,51.61) -- (570.14,113.39) -- (536.41,113.39) -- cycle ;
%Straight Lines [id:da6805655700801244] 
\draw [color={rgb, 255:red, 126; green, 211; blue, 33 }  ,draw opacity=1 ][line width=3.75]    (260,84) -- (260,114) ;
%Straight Lines [id:da4310701520538843] 
\draw [color={rgb, 255:red, 0; green, 93; blue, 235 }  ,draw opacity=1 ][line width=3.75]    (260,51) -- (260,84) ;
%Shape: Brace [id:dp46823409907886804] 
\draw   (342,135) .. controls (342.01,139.67) and (344.35,141.99) .. (349.02,141.98) -- (424.09,141.73) .. controls (430.76,141.7) and (434.1,144.02) .. (434.11,148.69) .. controls (434.1,144.02) and (437.42,141.68) .. (444.09,141.66)(441.09,141.67) -- (519.16,141.41) .. controls (523.83,141.4) and (526.15,139.06) .. (526.14,134.39) ;
%Shape: Brace [id:dp5333777178813743] 
\draw   (580.14,136.39) .. controls (580.21,141.06) and (582.58,143.35) .. (587.25,143.27) -- (599.76,143.06) .. controls (606.42,142.95) and (609.79,145.22) .. (609.87,149.89) .. controls (609.79,145.22) and (613.08,142.83) .. (619.75,142.72)(616.75,142.77) -- (632.26,142.51) .. controls (636.92,142.43) and (639.21,140.06) .. (639.14,135.39) ;

% Text Node
\draw (64,34.4) node [anchor=north west][inner sep=0.75pt]  [font=\small]  {$i$};
% Text Node
\draw (202,59.4) node [anchor=north west][inner sep=0.75pt]    {$U_{t}$};
% Text Node
\draw (543.41,57.01) node [anchor=north west][inner sep=0.75pt]    {$U_{t}$};
% Text Node
\draw (124,115.4) node [anchor=north west][inner sep=0.75pt]  [font=\small]  {$j$};
% Text Node
\draw (304,76.4) node [anchor=north west][inner sep=0.75pt]    {$=$};
% Text Node
\draw (95,35.4) node [anchor=north west][inner sep=0.75pt]  [font=\small]  {$i$};
% Text Node
\draw (128,34.4) node [anchor=north west][inner sep=0.75pt]  [font=\small]  {$i$};
% Text Node
\draw (159,35.4) node [anchor=north west][inner sep=0.75pt]  [font=\small]  {$i$};
% Text Node
\draw (192,34.4) node [anchor=north west][inner sep=0.75pt]  [font=\small]  {$i$};
% Text Node
\draw (223,35.4) node [anchor=north west][inner sep=0.75pt]  [font=\small]  {$i$};
% Text Node
\draw (256,34.4) node [anchor=north west][inner sep=0.75pt]  [font=\small]  {$i$};
% Text Node
\draw (155,116.4) node [anchor=north west][inner sep=0.75pt]  [font=\small]  {$j$};
% Text Node
\draw (190,117.4) node [anchor=north west][inner sep=0.75pt]  [font=\small]  {$j$};
% Text Node
\draw (224,118.4) node [anchor=north west][inner sep=0.75pt]  [font=\small]  {$j$};
% Text Node
\draw (253,121.4) node [anchor=north west][inner sep=0.75pt]  [font=\small]  {$j$};
% Text Node
\draw (452,117.4) node [anchor=north west][inner sep=0.75pt]  [font=\small]  {$j$};
% Text Node
\draw (483,118.4) node [anchor=north west][inner sep=0.75pt]  [font=\small]  {$j$};
% Text Node
\draw (518,118.4) node [anchor=north west][inner sep=0.75pt]  [font=\small]  {$j$};
% Text Node
\draw (578,116.4) node [anchor=north west][inner sep=0.75pt]  [font=\small]  {$j$};
% Text Node
\draw (609,116.4) node [anchor=north west][inner sep=0.75pt]  [font=\small]  {$j$};
% Text Node
\draw (388,37.4) node [anchor=north west][inner sep=0.75pt]  [font=\small]  {$i$};
% Text Node
\draw (419,38.4) node [anchor=north west][inner sep=0.75pt]  [font=\small]  {$i$};
% Text Node
\draw (452,37.4) node [anchor=north west][inner sep=0.75pt]  [font=\small]  {$i$};
% Text Node
\draw (483,38.4) node [anchor=north west][inner sep=0.75pt]  [font=\small]  {$i$};
% Text Node
\draw (516,37.4) node [anchor=north west][inner sep=0.75pt]  [font=\small]  {$i$};
% Text Node
\draw (581,38.4) node [anchor=north west][inner sep=0.75pt]  [font=\small]  {$i$};
% Text Node
\draw (610,37.4) node [anchor=north west][inner sep=0.75pt]  [font=\small]  {$i$};
% Text Node
\draw (421.41,150.01) node [anchor=north west][inner sep=0.75pt]    {$X_{t}^{( i,j)}$};
% Text Node
\draw (596.41,151.01) node [anchor=north west][inner sep=0.75pt]    {$Y_{t}^{( i,j)}$};

\end{tikzpicture}
    \caption{Our goal is to show that the Frobenius norm of the above matrix (summed over $i,j$) is at most the Frobenius norm of any $U_t$ appearing in this sequence (up to $\poly(d)$ factors). To do so, we split the decomposition into the part before $t$ (call it $X^{(i,j)}$) and the part after $t$ (call it $Y^{(i,j)})$ and our goal then becomes to show that $\vabs{\sum_{i,j}X^{(i,j)}U_t Y^{(i,j)}}_\frob\le \poly(d)\cdot \|U_t\|_\frob$. }
    \label{fig_2:definition_split_hBQP}
\end{figure}
This is depicted in~\Cref{fig_2:definition_hBQP}. Let us fix any $t\in[t_2,d)$. We will now show that $\|U\|_\frob \le O(d^4)\cdot \|U_t\|_\frob$. To do so, we will split this quantity around $U_t$, as depicted in~\Cref{fig_2:definition_split_hBQP}. To do so, we use the fact that $\tilP_t^{(i_{t_1},i_{t_2})}=U_t\cdot P^{(i_{t_1},i_{t_2})}$ to conclude that 
\[ U=U_{[1,t_1)}\cdot \sum_{i_{t_1},i_{t_2}\in[N]} \underbrace{Q^{(i_{t_1})}\cdot \tilP^{(i_{t_1})}_{[t_1,t_2)}\cdot Q^{(i_{t_2})}\cdot \tilP^{(i_{t_1},i_{t_2})}_{[t_2,t)}}_{X^{(i_{t_1},i_{t_2})}_{t}}\cdot U_t\cdot \underbrace{P^{(i_{t_1},i_{t_2})}\cdot \tilP^{(i_{t_1},i_{t_2})}_{(t,d]}}_{Y^{(i_{t_1},i_{t_2})}_{t}} \]
Since $\|U_{[1,t_1)}\|_\op \le 1$, we have 
\begin{equation}\|U\|_\frob \le \vabs{\sum_{i_{t_1},i_{t_2}\in[N]}X^{(i_{t_1},i_{t_2})}_{t}\cdot U_t \cdot Y^{(i_{t_1},i_{t_2})}_{t}}_\frob.\label{eq:X_U_Y}\end{equation}
This brings us to the form where we wish to apply~\Cref{lemma:sum_frob_norm}, but the various $Y_t^{(i_{t_1},i_{t_2})}$ unfortunately don't satisfy hypothesis of~\Cref{lemma:sum_frob_norm}. To get around this, we are going to use a telescoping sum to expand $Y^{(i_{t_1},i_{t_2})}_{t}$ into terms that do satisfy the hypothesis, furthermore, the number of terms is small, $\poly(d)$. 

For ease of notation, let $\tilR_{t'}^{(i,j)}:=P^{(i,j)}\cdot U_{t'}$ for $t'\in (t,d]$ and let $\tilR_{t'}^{(i)}:=P^{(i)}\cdot U_{t'}$ for $t'\in (t,d]$ and observe that $Y^{(i_{t_1},i_{t_2})}_{t}\triangleq \tilR^{(i_{t_1},i_{t_2})}_{(t,d]}$. For $r,r'\in(t,d]$, define
\begin{align} \label{eq:W_r} \begin{split}U'&:= U_{(t,d]} \\
W_r^{(i)}&:=  \tilR^{(i)}_{(t,r)} \cdot Q^{(i)}\cdot U_{[r,d]}  \\
W_{r,r'}^{(i,j)}&:= \begin{cases}\tilR^{(i,j)}_{(t,r')}\cdot Q^{(j)}\cdot \tilR^{(i)}_{[r',r)} \cdot Q^{(i)}\cdot U_{[r,d]}&\text{ if }r>r'\\ \tilR^{(i,j)}_{(t,r)}\cdot Q^{(i)}\cdot \tilR^{(j)}_{[r,r')} \cdot Q^{(j)}\cdot U_{[r',d]} &\text{ if } r\le r' \end{cases} \end{split}\end{align}
\input{fig_2_telescoping_Y}
By~\Cref{eq:Q_i_difference_2} and the telescoping argument depicted in~\Cref{fig_2:telescoping_Y}, for all $i_{t_1},i_{t_2}\in[N]$, 
\[Y_t^{(i_{t_1},i_{t_2})}=U'-2\sum_{r\in(t,d]}W_r^{(i_{t_1})} -2\sum_{r'\in(t,d]}W_{r'}^{(i_{t_2})}+4\sum_{r, r'\in(t,d]}W_{r,r'}^{(i_{t_1},i_{t_2})}.\]
Substituting the above in~\Cref{eq:X_U_Y}, we have
\begin{align}\label{eq:bound_split}\begin{split}
  \|U\|_\frob&\le  \vabs{\pbra{\sum_{i_{t_1},i_{t_2}}X^{(i_{t_1},i_{t_2})}_{t}}\cdot U_t \cdot U_{(t,d]}}_\frob   +O(d)\cdot \max_{r\in(t,d]}\vabs{\sum_{i_{t_1}}\pbra{\sum_{i_{t_2}}X^{(i_{t_1},i_{t_2})}_{t}}\cdot U_t \cdot W_r^{(i_{t_1})}}_\frob \\ 
& + O(d)\cdot \max_{r'\in(t,d]}\vabs{\sum_{i_{t_2}}\pbra{ \sum_{i_{t_1}}X^{(i_{t_1},i_{t_2})}_{t}}\cdot U_t \cdot W_{r'}^{(i_{t_2})}}_\frob \\
&+O(d^2)\cdot \max_{r,r'\in(t,d]}\vabs{\sum_{i_{t_1},i_{t_2}}X^{(i_{t_1},i_{t_2})}_{t}\cdot U_t \cdot W_{r,r'}^{(i_{t_1},i_{t_2})}}_\frob.\end{split}
\end{align} 

We will now use~\Cref{lemma:sum_frob_norm} to control each term in the R.H.S. above. 
\begin{claim}\label{claim:W_r} For all $r,r'\in(t,d]$, the matrices $W_r^{(i)},W_{r,r'}^{(i,j)}$ as in~\Cref{eq:W_r} satisfy the hypothesis of~\Cref{lemma:sum_frob_norm}, playing the role of $Y'$s, i.e.,
\[ \sum_i W_r^{(i)\dagger}\cdot W_r^{(i)}\preceq \id \]
\[ \sum_{i,j} W_{r,r'}^{(i,j)\dagger}\cdot W_{r,r'}^{(i,j)}\preceq \id \]
\end{claim}

\begin{proof}[Proof of~\Cref{claim:W_r}]
Recalling the definition of $W_r^{(i)},W_{r,r'}^{(j)},$ we see that 
\begin{align*}
    \sum_i W_r^{(i)\dagger}\cdot W_r^{(i)}&=\sum_i U_{[r,d]}^\dagger\cdot Q^{(i)\dagger}\cdot \tilR^{(i)\dagger}_{(t,r)}\cdot \tilR^{(i)}_{(t,r)}\cdot Q^{(i)}\cdot U_{[r,d]}\\
    &\preceq \sum_i U_{[r,d]}^\dagger\cdot Q^{(i)\dagger}\cdot Q^{(i)}\cdot U_{[r,d]}\tag{since $\|\tilR^{(i)}_{(t,r)}\|_\op\le 1$}\\
    &=U_{[r,d]}^\dagger\cdot U_{[r,d]}\preceq \id.\tag{by~\Cref{eq:Q_i_sum_to_one} applied to $Q^{(i)}$}
\end{align*}
Similarly we show the desired inequality for $r\le r'$ and the proof for $r>r'$ is identical. Consider 
\begin{align*}
    \sum_{i,j} W_{r,r'}^{(i,j)\dagger}\cdot W_{r,r'}^{(i,j)}&=\sum_{i,j}  U_{[r',d]}^\dagger \cdot Q^{(j)\dagger }\cdot \tilR^{(j)\dagger }_{[r,r')}\cdot  Q^{(i)\dagger}\cdot \tilR^{(i,j)\dagger}_{(t,r)}\cdot \tilR^{(i,j)}_{(t,r)}\cdot Q^{(i)}\cdot \tilR^{(j)}_{[r,r')} \cdot Q^{(j)}\cdot U_{[r',d]} \\
    &\preceq \sum_{i,j}  U_{[r',d]}^\dagger \cdot Q^{(j)\dagger }\cdot \tilR^{(j)\dagger }_{[r,r')}\cdot  Q^{(i)\dagger}\cdot Q^{(i)}\cdot \tilR^{(j)}_{[r,r')} \cdot Q^{(j)}\cdot U_{[r',d]} \tag{since $\|\tilR^{(i,j)}_{(t,r)}\|_\op \le1$}\\
    &\preceq \sum_{j}  U_{[r',d]}^\dagger \cdot Q^{(j)\dagger }\cdot \tilR^{(j)\dagger }_{[r,r')}\cdot \tilR^{(j)}_{[r,r')} \cdot Q^{(j)}\cdot U_{[r',d]} \tag{by~\Cref{eq:Q_i_sum_to_one} applied to $Q^{(i)}$}\\
     &\preceq \sum_{j}  U_{[r',d]}^\dagger \cdot Q^{(j)\dagger }  \cdot Q^{(j)}\cdot U_{[r',d]} \tag{since $\|\tilR^{(j)}_{[r,r')}\|_\op\le 1$}\\
     &\preceq U_{[r',d]}^\dagger \cdot U_{[r',d]}\preceq \id.\tag{by~\Cref{eq:Q_i_sum_to_one} applied to $Q^{(j)}$}
\end{align*}
\end{proof}
\begin{claim}\label{claim:new_claim} Let \[ X_t^{(i_{t_1},i_{t_2})}\triangleq  Q^{(i_{t_1})}\cdot \tilP^{(i_{t_1})}_{[t_1,t_2)}\cdot Q^{(i_{t_2})}\cdot \tilP^{(i_{t_1},i_{t_2})}_{[t_2,t)}.\] 
Then, $\cbra{X_t^{(i_{t_1},i_{t_2})}}_{i_{t_1},i_{t_2}}$, $\cbra{\sum_{i_{t_2}}X_t^{(i_{t_1},i_{t_2})} }_{i_{t_1}}$,  $\cbra{\sum_{i_{t_1}}X_t^{(i_{t_1},i_{t_2})} }_{i_{t_2}}$, and  $\cbra{\sum_{i_{t_1},i_{t_2}}X_t^{(i_{t_1},i_{t_2})}}$ all satisfy the hypothesis of~\Cref{lemma:sum_frob_norm}, playing the role of $X'$s with $a=O(d^2)$, i.e.,
\begin{align}\label{eq:bounds_X}\begin{split}
\sum_{i_{t_1},i_{t_2}} X_t^{(i_{t_1},i_{t_2})\dagger}\cdot X_t^{(i_{t_1},i_{t_2})} &\preceq O(d^4)\cdot \id \\
\sum_{i_{t_1}}\pbra{\sum_{i_{t_2}}X_t^{(i_{t_1},i_{t_2})\dagger}}\cdot \pbra{\sum_{i_{t_2}} X_t^{(i_{t_1},i_{t_2})}}&\preceq O(d^4)\cdot \id \\
\sum_{i_{t_2}}\pbra{\sum_{i_{t_1}}X_t^{(i_{t_1},i_{t_2})\dagger}}\cdot \pbra{\sum_{i_{t_1}} X_t^{(i_{t_1},i_{t_2})}}&\preceq O(d^4)\cdot \id \\
\pbra{\sum_{i_{t_1},i_{t_2}}X_t^{(i_{t_1},i_{t_2})\dagger}}\cdot \pbra{\sum_{i_{t_1},i_{t_2}} X_t^{(i_{t_1},i_{t_2})}}&\preceq O(d^4)\cdot \id
\end{split}
\end{align}
\end{claim}

Assuming the above claim, the proof is immediate. Applying~\Cref{lemma:sum_frob_norm} on each term in~\Cref{eq:bound_split} and using the bounds from~\Cref{claim:new_claim} and~\Cref{claim:W_r} would show
that
\begin{equation}\label{eq_bound_1}\|U\|_\frob \le O(d^2)\cdot O(d^2)\cdot \min_{t\in[t_2,d)}\|U_t\|_\frob.\end{equation}
Before proving~\Cref{claim:new_claim}, we proceed to analyze the other cases. Now suppose $t=d$ or $t\in[1,t_1)$. Then using the fact that $\tilP_d^{(i_{t_1},i_{t_2})}=U_d$, we get
\begin{align*}\|U\|_\frob&=\vabs{U_{[1,t_1)}\cdot \sum_{i_{t_1},i_{t_2}\in[N]} Q^{(i_{t_1})}\cdot \tilP^{(i_{t_1})}_{[t_1,t_2)}\cdot Q^{(i_{t_2})}\cdot \tilP^{(i_{t_1},i_{t_2})}_{[t_2,d-1]} \cdot U_d}_\frob\\
&\le \vabs{\sum_{i_{t_1},i_{t_2}\in[N]} Q^{(i_{t_1})}\cdot \tilP^{(i_{t_1})}_{[t_1,t_2)}\cdot Q^{(i_{t_2})}\cdot \tilP^{(i_{t_1},i_{t_2})}_{[t_2,d-1]}}_\op\cdot \min\pbra{\vabs{U_{[1,t_1)}}_\frob,\vabs{U_d}_\frob} \tag{since $\|U_d\|_\op,\|U_{[1,t_1)}\|_\op \le 1$ and by~\Cref{fact:frob_op}}\\
&\le \min\pbra{\vabs{U_{[1,t_1)}}_\frob,\vabs{U_d}_\frob} \cdot O(d^2)
\end{align*}
where the last bound on the operator norm follows from a calculation similar to~\Cref{claim:new_claim} where we obtained bounds of $O(d^2)$ on the operator norms of $\sum_{i_{t_1},i_{t_2}}X_t^{(i_{t_1},i_{t_2})}$.
This proves that
\begin{equation}
    \|U\|_\frob \le O(d^2)\cdot \min_{t\in[1,t_1)\cup\{d\}}\|U_t\|_\frob \label{eq_bound_4}
\end{equation}
Finally, we consider any $t\in[t_1,t_2)$. Observe that
\begin{align*}\|U\|_\frob^2 
 &\le  \vabs{\sum_{i_{t_1}\in [N]} Q^{(i_{t_1})}\cdot \tilP^{(i_{t_1})}_{[t_1,t_2)}\cdot \pbra{ \sum_{i_{t_2}\in[N]}
Q^{(i_{t_2})}\cdot  \tilP_{[t_2,d]}^{(i_{t_1},i_{t_2})}}}_\frob^2 \tag{since $\|U_{[1,t_1)}\|_\op \le 1$}\\
&=  \sum_{i_{t_1}\in [N]} \vabs{Q^{(i_{t_1})}\cdot \tilP^{(i_{t_1})}_{[t_1,t_2)}\cdot \pbra{ \sum_{i_{t_2}\in[N]}
Q^{(i_{t_2})}\cdot  \tilP_{[t_2,d]}^{(i_{t_1},i_{t_2})}}}_\frob^2 \tag{by~\Cref{eq:Q_i_ortho_2}}\\
&\le \sum_{i_{t_1}\in [N]} \vabs{Q^{(i_{t_1})}\cdot \tilP^{(i_{t_1})}_{[t_1,t_2)}}_\frob^2\cdot \vabs{ \sum_{i_{t_2}\in[N]}
Q^{(i_{t_2})}\cdot  \tilP_{[t_2,d]}^{(i_{t_1},i_{t_2})}}^2_\op \tag{by~\Cref{fact:frob_op}}.
\end{align*}
For any $i_{t_1}\in[N],$ we see that the matrix $\sum_{i_{t_2}\in[N]}
Q^{(i_{t_2})}\cdot  \tilP_{[t_2,d]}^{(i_{t_1},i_{t_2})}$ is precisely obtained by applying~\Cref{lem:main_lemma_2} with $\gamma=2$ to the matrices $U_{t_2}\cdot P^{(i_{t_1})},U_{t_2+1}\cdot P^{(i_{t_1})},\ldots, U_{d-1}\cdot P^{(i_{t_1})},U_d$ and therefore, its operator norm is at most $O(d)$. Therefore, we get
\begin{align*} \| U\|_\frob^2 &\le O(d^2)\sum_{i_{t_1}\in[N]} \vabs{Q^{(i_{t_1})}\cdot \tilP^{(i_{t_1})}_{[t_1,t_2)}}_\frob^2\\
&=O(d^2)\vabs{\sum_{i_{t_1}\in[N]} Q^{(i_{t_1})}\cdot \tilP^{(i_{t_1})}_{[t_1,t_2)}}_\frob^2 \tag{by~\Cref{eq:Q_i_ortho_2}}\\
&\le  O(d^4)\min_{t\in[t_1,t_2)} ( \|U_t\|_\frob^2) \tag{by an application of~\Cref{lem:main_lemma_2} to the matrices $U_{t_1},\ldots,U_{t_2-1}$ with $\gamma=2$.}
\end{align*} 
This shows that 
\begin{equation}\label{eq_bound_3}\|U\|_\frob \le O(d^2)\min_{t\in[t_1,t_2)} \|U_t\|_\frob \end{equation}
Combining~\Cref{eq_bound_1,eq_bound_3,eq_bound_4} gives us the desired bound of 
\[\|U\|_\frob \le O(d^4)\cdot \min_t \|U_t\|_\frob. \]

All that remains is to prove~\Cref{claim:new_claim}. 
\begin{proof}[Proof of~\Cref{claim:new_claim}]

\input{fig_2_telescoping_X}

While it can be shown that $\sum_{i_{t_1},i_{t_2}}X_t^{(i_{t_1},i_{t_2})}\cdot X_t^{(i_{t_1},i_{t_2})\dagger}\preceq \id$ using an argument similar to what we did for $W_{r,r'}^{(i_{t_1},i_{t_2})}$, what we wish for is $\sum_{i_{t_1},i_{t_2}}X_t^{(i_{t_1},i_{t_2})\dagger}\cdot X_t^{(i_{t_1},i_{t_2})}\preceq \id$. To show this, we are again going to use a telescoping sum to express $X_t^{(i_{t_1},i_{t_2})}$  as a sum of $O(d^2)$ matrices which satisfy the hypothesis of~\Cref{lemma:sum_op_norm}. We will then use~\Cref{lemma:sum_op_norm} to conclude that $X_t^{(i_{t_1},i_{t_2})}$ satisfies itself the hypothesis of~\Cref{lemma:sum_frob_norm} with $a=O(d^2)$. To express $X^{(i_{t_1},i_{t_2})}$ as desired, we will apply the telescoping identity on the terms $\tilP^{(i_{t_1})}_{[t_1,t_2)}$ and $\tilP^{(i_{t_1},i_{t_2})}_{[t_2,t)}$. Define matrices as follows. For $i,j\in[N]$, and $r\in (t_1,t),r'\in (t_2,t),$ let
\begin{align}\label{eq:new_w_r}\begin{split}
A^{(i,j)}&:=Q^{(i)}\cdot U_{[t_1,t_2)}\cdot Q^{(j)}\cdot U_{[t_2,t)}\\
B_{r'}^{(i,j)}&:=  Q^{(i)}\cdot U_{[t_1,t_2)}\cdot Q^{(j)}\cdot \tilP^{(j)}_{[t_2,r')}\cdot Q^{(j)}\cdot U_{[r',t)}\\
C_{r}^{(i,j)}&:=\begin{cases}Q^{(i)}\cdot \tilP^{(i)}_{[t_1,r)}\cdot Q^{(i)}\cdot   U_{[r,t_2)}\cdot Q^{(j)}\cdot U_{[t_2,t)}&\text{if }r\le  t_2 \\
Q^{(i)}\cdot \tilP^{(i)}_{[t_1,t_2)}\cdot Q^{(j)}\cdot \tilP^{(i)}_{[t_2,r)}\cdot   Q^{(i)}\cdot U_{[r,t)}&\text{if } r>t_2 \end{cases}\\
D_{r,r'}^{(i,j)}&:=\begin{cases}Q^{(i)}\cdot \tilP^{(i)}_{[t_1,r)}\cdot Q^{(i)}\cdot   U_{[r,t_2)}\cdot Q^{(j)}\cdot \tilP^{(j)}_{[t_2,r')}\cdot Q^{(j)}\cdot  U_{[r',t)}&\text{if }r\le t_2 \\
 Q^{(i)}\cdot \tilP^{(i)}_{[t_1,t_2)}\cdot Q^{(j)}\cdot \tilP^{(i,j)}_{[t_2,r)}\cdot   Q^{(i)}\cdot \tilP^{(j)}_{[r,r')}\cdot   Q^{(j)}\cdot U_{[r',t)}&\text{if } t_2<r\le  r' \\
  Q^{(i)}\cdot \tilP^{(i)}_{[t_1,t_2)}\cdot Q^{(j)}\cdot \tilP^{(i,j)}_{[t_2,r')}\cdot   Q^{(j)}\cdot \tilP^{(i)}_{[r',r)}\cdot Q^{(i)}\cdot U_{[r,t)}&\text{if } r>r' \\ 
\end{cases}\end{split}
\end{align}
Again, by the telescoping expansion described in~\Cref{fig_2:telescoping_X}, we have
\[ X_t^{(i_{t_1},i_{t_2})}= A^{(i_{t_1},i_{t_2})} - 2\sum_{r'\in(t_2,t)} B_{r'}^{(i_{t_1},i_{t_2})} - 2\sum_{r\in(t_1,t)} C_{r}^{(i_{t_1},i_{t_2})}+4\sum_{\substack{r'\in(t_2,t)\\r\in(t_1,t)}} D_{r,r'}^{(i_{t_1},i_{t_2})}. \]
We will now show that the matrices in~\Cref{eq:new_w_r} satisfy the hypothesis of~\Cref{lemma:sum_op_norm}. This is done by a case by case analysis. %The idea behind this is to consider terms in $W^\dagger W$  and delete coupled terms in the middle. The important point is that in each of these expressions, the $\dagger$ part contributes two occurrences each of $Q^{(i_{t_1})\dagger}$ and $Q^{(i_{t_2})\dagger}$ to the left that are paired with two corresponding occurrences each of $Q^{(i_{t_1})}$ and $Q^{(i_{t_2})}$ on the right. We keep the outer occurrences $Q^{(i_{t_1})\dagger}$, $Q^{(i_{t_2})\dagger}$, as well as $Q^{(i_{t_1})}, Q^{(i_{t_2})}$ fixed, and we systematically delete symmetric terms in the middle, as their operator norm is at most 1. We keep doing this until we end up matching the outermost pair $Q^{(i)\dagger}\cdot Q^{(i)}$ and for the first such $i$, there will be no other occurrences of $i$ in the remaining expression, so we may use~\Cref{eq:Q_i_sum_to_one} to remove this term and proceed. We describe this in more detail. 
For ease of notation let $i_{t_1}=i$ and $i_{t_2}=j$. We now show the desired bound for $A^{(i,j)}, B_{r'}^{(i,j)},$ $C_r^{(i,j)}$ and $D_{r,r'}^{(i,j)}$.

\paragraph*{Analysis for $A^{(i,j)}$.} Consider $A^{(i,j)}=Q^{(i)}\cdot U_{[t_1,t_2)}\cdot Q^{(j)}\cdot U_{[t_2,t)}$. This is depicted in the top right of~\Cref{fig_3:cancellation_A}. We wish to take $\sum \pbra{\sum A^{(i,j)\dagger}}\cdot \pbra{\sum A^{(i,j)}}$, where we intentionally omit the identity of the indices being summed over in the inner summation -- it could be the empty summation, the summation over $i$, the summation over $j$, or the summation over both $i$ and $j$, and the outer summation is always understood to be over the remaining variables, if any. For all the resulting expressions, we wish to show that the result is $\preceq\id$. In all these cases, the expression consists of a sum over possibly $i',j'$ of the $\dagger$ terms and a sum over possibly $i,j$ of the regular terms. We first present the high level intuition and then the calculation. As before, we observe that the inner terms consist of $Q^{(i')\dagger}\cdot Q^{(i)}$ and by the orthogonality of the $Q^{(i)}$, the cross terms corresponding to $i\neq i'$ to vanish and we can pull the $\sum_i$ outside. Furthermore, we can sum over $i$ to obtain the identity matrix, as the terms before and after do not depend on $i$. Then, we can strip away the matrices between $Q^{(j')\dagger}$ and $Q^{(j)}$, as these only have operator norm at most 1 and do not depend on $j,j'$. We then use the orthogonality of the $Q^{(j)}$ to argue that the cross terms corresponding $j\neq j'$ vanish and we can pull the $\sum_j$ outside. Furthermore, we can sum over $j$ to obtain the identity matrix, as the terms before and after do not depend on $j$. This process is depicted in~\Cref{fig_3:cancellation_A} and is mathematically described below.

\begin{align*}
&\sum \pbra{\sum A^{(i,j)}}\cdot\pbra{\sum A^{(i,j)}}\\
&=\sum_i\underbrace{\pbra{\sum U_{[t_2,t)}^\dagger \cdot Q^{(j)\dagger}\cdot U_{[t_1,t_2)}^\dagger}}_{\text{independent of }i} \cdot \underbrace{ Q^{(i)\dagger }\cdot Q^{(i)}}_{\text{ sum to }\id}\cdot \underbrace{\pbra{\sum U_{[t_1,t_2)}\cdot Q^{(j)}\cdot U_{[t_2,t)}}}_{\text{independent of }i}\tag{\Cref{eq:Q_i_ortho_2} applied to $Q^{(i)}$}\\
&=\sum \pbra{\sum U_{[t_2,t)}^\dagger \cdot Q^{(j)\dagger}} \cdot U_{[t_1,t_2)}^\dagger \cdot U_{[t_1,t_2)}\cdot \pbra{\sum Q^{(j)}\cdot U_{[t_2,t)}}\\
&\preceq \sum \pbra{\sum U_{[t_2,t)}^\dagger \cdot Q^{(j)\dagger}} \cdot \pbra{\sum Q^{(j)}\cdot U_{[t_2,t)}}\\
&= \sum  U_{[t_2,t)}^\dagger \cdot \underbrace{Q^{(j)\dagger}\cdot  Q^{(j)}}_{\text{sum to }\id }\cdot U_{[t_2,t)} =U_{[t_2,t)}^\dagger\cdot  U_{[t_2,t)}\preceq \id.\tag{\Cref{eq:Q_i_ortho_2} applied to $Q^{(j)}$}
\end{align*}

\begin{figure}[t]
    \centering

% Pattern Info
 
\tikzset{
pattern size/.store in=\mcSize, 
pattern size = 5pt,
pattern thickness/.store in=\mcThickness, 
pattern thickness = 0.3pt,
pattern radius/.store in=\mcRadius, 
pattern radius = 1pt}
\makeatletter
\pgfutil@ifundefined{pgf@pattern@name@_khtb7etay}{
\pgfdeclarepatternformonly[\mcThickness,\mcSize]{_khtb7etay}
{\pgfqpoint{0pt}{0pt}}
{\pgfpoint{\mcSize+\mcThickness}{\mcSize+\mcThickness}}
{\pgfpoint{\mcSize}{\mcSize}}
{
\pgfsetcolor{\tikz@pattern@color}
\pgfsetlinewidth{\mcThickness}
\pgfpathmoveto{\pgfqpoint{0pt}{0pt}}
\pgfpathlineto{\pgfpoint{\mcSize+\mcThickness}{\mcSize+\mcThickness}}
\pgfusepath{stroke}
}}
\makeatother

% Pattern Info
 
\tikzset{
pattern size/.store in=\mcSize, 
pattern size = 5pt,
pattern thickness/.store in=\mcThickness, 
pattern thickness = 0.3pt,
pattern radius/.store in=\mcRadius, 
pattern radius = 1pt}
\makeatletter
\pgfutil@ifundefined{pgf@pattern@name@_ip71cc7sz}{
\pgfdeclarepatternformonly[\mcThickness,\mcSize]{_ip71cc7sz}
{\pgfqpoint{0pt}{0pt}}
{\pgfpoint{\mcSize+\mcThickness}{\mcSize+\mcThickness}}
{\pgfpoint{\mcSize}{\mcSize}}
{
\pgfsetcolor{\tikz@pattern@color}
\pgfsetlinewidth{\mcThickness}
\pgfpathmoveto{\pgfqpoint{0pt}{0pt}}
\pgfpathlineto{\pgfpoint{\mcSize+\mcThickness}{\mcSize+\mcThickness}}
\pgfusepath{stroke}
}}
\makeatother
\tikzset{every picture/.style={line width=0.75pt}} %set default line width to 0.75pt        

\begin{tikzpicture}[x=0.75pt,y=0.75pt,yscale=-1,xscale=1]
%uncomment if require: \path (0,300); %set diagram left start at 0, and has height of 300

%Shape: Rectangle [id:dp5226126007001137] 
\draw  [draw opacity=0][pattern=_khtb7etay,pattern size=6pt,pattern thickness=0.75pt,pattern radius=0pt, pattern color={rgb, 255:red, 0; green, 0; blue, 0}] (353.25,59.97) -- (389.47,59.97) -- (389.47,113.44) -- (353.25,113.44) -- cycle ;
%Shape: Rectangle [id:dp7565824076586921] 
\draw   (352.81,61.45) -- (442.14,61.45) -- (442.14,115.66) -- (352.81,115.66) -- cycle ;
%Straight Lines [id:da6393863612549795] 
\draw [color={rgb, 255:red, 255; green, 0; blue, 0 }  ,draw opacity=1 ][line width=3.75]    (353.25,59.97) -- (353.25,88.64) ;
%Straight Lines [id:da6607417280267561] 
\draw [color={rgb, 255:red, 245; green, 166; blue, 35 }  ,draw opacity=1 ][line width=3.75]    (389.7,88.76) -- (389.81,116.03) ;
%Shape: Rectangle [id:dp046485965073254976] 
\draw   (174.81,61.45) -- (264.14,61.45) -- (264.14,115.66) -- (174.81,115.66) -- cycle ;
%Straight Lines [id:da4059760970724299] 
\draw [color={rgb, 255:red, 255; green, 0; blue, 0 }  ,draw opacity=1 ][line width=3.75]    (264.14,61.45) -- (264.14,90.13) ;
%Straight Lines [id:da14650917461346857] 
\draw [color={rgb, 255:red, 245; green, 166; blue, 35 }  ,draw opacity=1 ][line width=3.75]    (225.25,87.69) -- (225.37,114.96) ;
%Shape: Rectangle [id:dp9506475540279135] 
\draw  [draw opacity=0][pattern=_ip71cc7sz,pattern size=6pt,pattern thickness=0.75pt,pattern radius=0pt, pattern color={rgb, 255:red, 0; green, 0; blue, 0}] (227.92,62.19) -- (264.14,62.19) -- (264.14,115.66) -- (227.92,115.66) -- cycle ;
%Shape: Rectangle [id:dp013787964389714524] 
\draw   (391.02,129.45) -- (442.14,129.45) -- (442.14,183.66) -- (391.02,183.66) -- cycle ;
%Straight Lines [id:da3334256430538046] 
\draw [color={rgb, 255:red, 245; green, 166; blue, 35 }  ,draw opacity=1 ][line width=3.75]    (391.02,156.39) -- (391.14,183.66) ;
%Shape: Rectangle [id:dp2055414539268402] 
\draw   (176.03,131.65) -- (226.25,131.65) -- (226.25,185.86) -- (176.03,185.86) -- cycle ;
%Straight Lines [id:da23580888614412976] 
\draw [color={rgb, 255:red, 245; green, 166; blue, 35 }  ,draw opacity=1 ][line width=3.75]    (226.25,159.69) -- (226.37,186.96) ;
%Straight Lines [id:da1940695446512326] 
\draw [color={rgb, 255:red, 0; green, 0; blue, 0 }  ,draw opacity=1 ]   (287.04,78.1) -- (332.05,78.41) ;
\draw [shift={(335.05,78.43)}, rotate = 180.39] [fill={rgb, 255:red, 0; green, 0; blue, 0 }  ,fill opacity=1 ][line width=0.08]  [draw opacity=0] (10.72,-5.15) -- (0,0) -- (10.72,5.15) -- (7.12,0) -- cycle    ;
\draw [shift={(284.04,78.08)}, rotate = 0.39] [fill={rgb, 255:red, 0; green, 0; blue, 0 }  ,fill opacity=1 ][line width=0.08]  [draw opacity=0] (10.72,-5.15) -- (0,0) -- (10.72,5.15) -- (7.12,0) -- cycle    ;
%Straight Lines [id:da9445853394277308] 
\draw [color={rgb, 255:red, 0; green, 0; blue, 0 }  ,draw opacity=1 ]   (235.04,165.09) -- (383.14,165.39) ;
\draw [shift={(386.14,165.39)}, rotate = 180.11] [fill={rgb, 255:red, 0; green, 0; blue, 0 }  ,fill opacity=1 ][line width=0.08]  [draw opacity=0] (10.72,-5.15) -- (0,0) -- (10.72,5.15) -- (7.12,0) -- cycle    ;
\draw [shift={(232.04,165.08)}, rotate = 0.11] [fill={rgb, 255:red, 0; green, 0; blue, 0 }  ,fill opacity=1 ][line width=0.08]  [draw opacity=0] (10.72,-5.15) -- (0,0) -- (10.72,5.15) -- (7.12,0) -- cycle    ;

% Text Node
\draw (350.2,46.2) node [anchor=north west][inner sep=0.75pt]  [font=\small]  {$i$};
% Text Node
\draw (101.3,64.4) node [anchor=north west][inner sep=0.75pt]    {$A^{( i',j') \dagger } =$};
% Text Node
\draw (454.3,67.4) node [anchor=north west][inner sep=0.75pt]    {$=A^{( i,j)}$};
% Text Node
\draw (259.2,46.2) node [anchor=north west][inner sep=0.75pt]  [font=\small]  {$i'$};
% Text Node
\draw (393.14,187.06) node [anchor=north west][inner sep=0.75pt]  [font=\small]  {$j$};
% Text Node
\draw (221.92,187.58) node [anchor=north west][inner sep=0.75pt]  [font=\small]  {$j'$};
% Text Node
\draw (295,52.79) node [anchor=north west][inner sep=0.75pt]  [font=\small]  {$i=i'$};
% Text Node
\draw (286.04,81.48) node [anchor=north west][inner sep=0.75pt]  [font=\small]  {$\sum _{i} =\mathbb{I}$};
% Text Node
\draw (291,143.79) node [anchor=north west][inner sep=0.75pt]  [font=\small]  {$j=j'$};
% Text Node
\draw (285.04,173.48) node [anchor=north west][inner sep=0.75pt]  [font=\small]  {$\sum _{j} =\mathbb{I}$};

\end{tikzpicture}

    \caption{How to show that $A^{(i,j)}$ satisfies the hypothesis of~\Cref{lemma:sum_op_norm}.}
    \label{fig_3:cancellation_A}
\end{figure}.

\paragraph*{Analysis for $B_{r'}^{(i,j)}$.} \input{fig_3_cancellation_B}
Consider $B_{r'}^{(i,j)}=Q^{(i)}\cdot U_{[t_1,t_2)}\cdot Q^{(j)}\cdot \tilP^{(j)}_{[t_2,r')}\cdot Q^{(j)}\cdot U_{[r',t)}$. This is depicted in the top right of~\Cref{fig_3:cancellation_B}. We wish to control $\sum \pbra{\sum B_{r'}^{(i,j)\dagger}}\cdot \pbra{\sum B_{r'}^{(i,j)}}$ The strategy is as follows. As before, we observe that the inner terms consist of $Q^{(i')\dagger}\cdot Q^{(i)}$ and by the orthogonality of the $Q^{(i)}$, the cross terms corresponding to $i\neq i'$ to vanish and we can pull the $\sum_i$ outside. Furthermore, we can sum over $i$ to obtain the identity matrix, as the terms before and after do not depend on $i$. Then, we can strip away the matrices between the closest pair of $Q^{(j')\dagger}$ and $Q^{(j)}$, as these only have operator norm at most 1 and are independent of $j,j'$. We then use the orthogonality of the $Q^{(j)}$ to argue that the cross terms corresponding $j\neq j'$ vanish and we can pull the $\sum_j$ outside. Finally, we strip away all the terms except for the outermost $Q^{(j)\dagger}$ and $Q^{(j)}$, which we can do as these matrices have operator norm at most one. Finally, we can sum over $j$ to obtain the identity matrix, as the terms before and after do not depend on $j$. This process is depicted in~\Cref{fig_3:cancellation_B} and is mathematically described below.
\begin{align*}
&\sum \pbra{\sum B_{r'}^{(i,j)}}\cdot\pbra{\sum B_{r'}^{(i,j)}}\\
&=\sum_i \underbrace{\pbra{\sum 
U_{[r',t)}^\dagger \cdot  Q^{(j)\dagger}\cdot \tilP^{(j)\dagger}_{[t_2,r')}\cdot Q^{(j)\dagger }\cdot U_{[t_1,t_2)}^\dagger }}_{\text{independent of }i}\cdot \underbrace{Q^{(i)\dagger }\cdot Q^{(i)}}_{\text{sum to }\id} \\
&\underbrace{\pbra{\sum 
U_{[t_1,t_2)}\cdot Q^{(j)}\cdot \tilP^{(j)}_{[t_2,r')}\cdot Q^{(j)}\cdot U_{[r',t)}}}_{\text{independent of }i}\tag{\Cref{eq:Q_i_ortho_2} applied to $Q^{(i)}$}\\
&=\sum \pbra{\sum 
U_{[r',t)}^\dagger \cdot  Q^{(j)\dagger}\cdot \tilP^{(j)\dagger}_{[t_2,r')}\cdot Q^{(j)\dagger }}\cdot U_{[t_1,t_2)}^\dagger \cdot U_{[t_1,t_2)}\cdot\pbra{\sum 
 Q^{(j)}\cdot \tilP^{(j)}_{[t_2,r')}\cdot Q^{(j)}\cdot U_{[r',t)}}\tag{\Cref{eq:Q_i_sum_to_one} applied to $Q^{(i)}$}\\
&\preceq\sum \pbra{\sum 
U_{[r',t)}^\dagger \cdot  Q^{(j)\dagger}\cdot \tilP^{(j)\dagger}_{[t_2,r')}\cdot Q^{(j)\dagger }}\cdot \pbra{\sum 
Q^{(j)}\cdot \tilP^{(j)}_{[t_2,r')}\cdot Q^{(j)}\cdot U_{[r',t)}}\\
&=\sum_j  
U_{[r',t)}^\dagger \cdot  Q^{(j)\dagger}\cdot \underbrace{ \tilP^{(j)\dagger}_{[t_2,r')}\cdot Q^{(j)\dagger }\cdot 
Q^{(j)}\cdot \tilP^{(j)}_{[t_2,r')}}_{\text{norm }\le 1} \cdot Q^{(j)}\cdot U_{[r',t)}\tag{\Cref{eq:Q_i_ortho_2} applied to $Q^{(j)}$}\\
&\preceq \sum_j  
U_{[r',t)}^\dagger \cdot  \underbrace{Q^{(j)\dagger}\cdot Q^{(j)}}_{\text{sum to }\id} \cdot U_{[r',t)}  \preceq \sum_j  
U_{[r',t)}^\dagger\cdot U_{[r',t)}\preceq \id.  \tag{\Cref{eq:Q_i_sum_to_one} applied to $Q^{(j)}$}
\end{align*}

\paragraph*{Analysis for $C_{r}^{(i,j)}$.} The analysis for this case is very similar to that of $B_{r'}^{(i,j)}$. We omit the details and describe the strategy in~\Cref{fig_3:cancellation_C} for $r>t_2$ and the analysis for the other case $r\le t_2$ is similar, just with operations performed in a different order.
\input{fig_3_cancellation_C}

\paragraph*{Analysis for $D_{r,r'}^{(i,j)}$.} For simplicity, consider the case $r\le t_2$. In this case, we have $D_{r,r'}^{(i,j)}=Q^{(i)}\cdot \tilP^{(i)}_{[t_1,r)}\cdot Q^{(i)}\cdot   U_{[r,t_2)}\cdot Q^{(j)}\cdot \tilP^{(j)}_{[t_2,r')}\cdot Q^{(j)}\cdot  U_{[r',t)}$. The term $D_{r,r'}^{(i,j)}$ is depicted at the top right of~\Cref{fig_3:cancellation_D}. We wish to take $\sum \pbra{\sum D_{r,r'}^{(i,j)}}\cdot \pbra{\sum D_{r,r'}^{(i,j)}}$ and and show that the result is $\preceq\id$.

Firstly, we observe that the inner terms consist of $Q^{(i')\dagger}\cdot Q^{(i)}$ and by the orthogonality of the $Q^{(i)}$, the cross terms corresponding to $i\neq i'$ to vanish and we can pull the $\sum_i$ outside. Now, for each $i$, we can strip away the matrices between the outer $Q^{(i)\dagger}$ and $Q^{(i)}$, as these only depend on $i$ and have operator norm at most 1. We then sum over $i$ to obtain the identity matrix, as the terms before and after this pair do not depend on $i$. Then, we again strip away all the terms until the remaining pair of $Q^{(j)\dagger}$ and $Q^{(j')}$, as these have operator norm at most 1 and are independent of $j,j'$. Then, we can use the orthogonality of the $Q^{(j)}$ to argue that the cross terms corresponding $j\neq j'$ vanish and we can pull the $\sum_j$ outside. Now, we strip away all the terms until the closest pair of $Q^{(j)\dagger}$ and $Q^{(j)}$ as these have operator norm at most 1. We then sum over $j$ to obtain the identity matrix, as the terms before and after don't depend on $j$. This process is depicted in~\Cref{fig_3:cancellation_D} and is mathematically described below.

\begin{align*}
&\sum \pbra{\sum D_{r,r'}^{(i',j')}}\cdot \pbra{\sum D_{r,r'}^{(i,j)}}\\
&=\sum_{i} \pbra{\sum U_{[r',t)}^\dagger \cdot Q^{(j)\dagger}\cdot \tilP^{(j)\dagger}_{[t_2,r')} \cdot Q^{(j)\dagger} \cdot   \pbra{U_{[r,t_2)}^\dagger\cdot Q^{(i)\dagger}\cdot \tilP^{(i)\dagger}_{[t_1,r)}\cdot Q^{(i)\dagger} }} \\
&\cdot \pbra{\sum \pbra{Q^{(i)}\cdot \tilP^{(i)}_{[t_1,r)}\cdot Q^{(i)}\cdot   U_{[r,t_2)}}\cdot Q^{(j)}\cdot \tilP^{(j)}_{[t_2,r')}\cdot Q^{(j)}\cdot  U_{[r',t)}}\tag{\Cref{eq:Q_i_ortho_2} applied to $Q^{(i)}$}\\
&=\sum_{i} \pbra{\sum U_{[r',t)}^\dagger \cdot Q^{(j)\dagger}\cdot \tilP^{(j)\dagger}_{[t_2,r')} \cdot Q^{(j)\dagger} }\cdot   \underbrace{\pbra{U_{[r,t_2)}^\dagger\cdot Q^{(i)\dagger}\cdot \tilP^{(i)\dagger}_{[t_1,r)}\cdot Q^{(i)\dagger} }}_{\text{ independent of }j} \\
&\cdot   \underbrace{\pbra{Q^{(i)}\cdot \tilP^{(i)}_{[t_1,r)}\cdot Q^{(i)}\cdot   U_{[r,t_2)}}}_{\text{ independent of }j}\cdot \pbra{\sum Q^{(j)}\cdot \tilP^{(j)}_{[t_2,r')}\cdot Q^{(j)}\cdot  U_{[r',t)}}\\
&\preceq \sum_{i} \underbrace{\pbra{\sum U_{[r',t)}^\dagger \cdot Q^{(j)\dagger}\cdot \tilP^{(j)\dagger}_{[t_2,r')} \cdot Q^{(j)\dagger} }}_{\text{independent of } i}\cdot U_{[r,t_2)}^\dagger\cdot \underbrace{Q^{(i)\dagger} \cdot  Q^{(i)}}_{\text{sum to }\id}\cdot U_{[r,t_2)}\\
&\cdot \underbrace{\pbra{\sum Q^{(j)}\cdot \tilP^{(j)}_{[t_2,r')}\cdot Q^{(j)}\cdot  U_{[r',t)}}}_{\text{independent of } i}\\
&= \sum \pbra{\sum U_{[r',t)}^\dagger \cdot Q^{(j)\dagger}\cdot \tilP^{(j)\dagger}_{[t_2,r')} \cdot Q^{(j)\dagger} }\cdot U_{[r,t_2)}^\dagger\cdot U_{[r,t_2)}\cdot \pbra{\sum Q^{(j)}\cdot \tilP^{(j)}_{[t_2,r')}\cdot Q^{(j)}\cdot  U_{[r',t)}}\tag{\Cref{eq:Q_i_sum_to_one} applied to $Q^{(i)}$}\\
&\preceq \sum \pbra{\sum U_{[r',t)}^\dagger \cdot Q^{(j)\dagger}\cdot \tilP^{(j)\dagger}_{[t_2,r')} \cdot Q^{(j)\dagger} }\cdot  \pbra{\sum Q^{(j)}\cdot \tilP^{(j)}_{[t_2,r')}\cdot Q^{(j)}\cdot  U_{[r',t)}}\\
&= \sum_{j}U_{[r',t)}^\dagger \cdot Q^{(j)\dagger}\cdot \underbrace{\tilP^{(j)\dagger}_{[t_2,r')} \cdot Q^{(j)\dagger}\cdot Q^{(j)}\cdot \tilP^{(j)}_{[t_2,r')}}_{\text{ norm }\le 1}\cdot Q^{(j)}\cdot  U_{[r',t)}\tag{\Cref{eq:Q_i_ortho_2} applied to $Q^{(j)}$}\\
&\preceq \sum_{j}U_{[r',t)}^\dagger \cdot \underbrace{Q^{(j)\dagger}\cdot Q^{(j)}}_{\text{sum to } \id}\cdot  U_{[r',t)} \preceq U_{[r',t)}^\dagger\cdot U_{[r',t)}\preceq \id \tag{\Cref{eq:Q_i_sum_to_one} applied to $Q^{(j)}$}.
\end{align*}

\input{fig_3_cancellation_D}

The analysis for the other cases is similar, and for completeness, we depict the strategy to handle $t_2<r\le r'$ in~\Cref{fig_3:cancellation_D(2)} and $r>r'$ in~\Cref{fig_3:cancellation_D(3)}. This completes the proof of~\Cref{claim:new_claim}.
\input{fig_3_cancellation_D_2}
\input{fig_3_cancellation_D_3}

\end{proof}
\end{proof}

\subsection{Quantum Algorithms with Classical Pre-Processing.}
\label{sec:hybrid}

In this section, we prove a variant of~\Cref{thm:main_theorem_dqck} and~\Cref{thm:main_theorem_bqp} in a more general setting of algorithms that can perform classical pre-processing. We now describe this model more formally. A $d$-query $\DQC{k}$ (respectively $\BQP$) algorithm with classical pre-processing consists of two phases:
\begin{itemize}
\item \textsc{Classical Phase:} The algorithm performs $d$ classical queries on clean workspace.
\item  \textsc{Quantum Phase:} Based on the results, the algorithm chooses a $d$-query $\DQC{k}$ (respectively $\BQP$) algorithm to run and returns the output. 
\end{itemize}

\begin{theorem} Let $\cF$ denote the family of acceptance probabilities of a class of algorithms without classical pre-processing. Let $f(x)$ be the acceptance probability of an algorithm with $d$ classical pre-processing queries. Let $\rho\in\{-1,1,*\}^N$ be any restriction and $\alpha\in[-1,1]^{\binom{N}{\ell}}$ signs. Then, there exist $f'\in \cF$ such that
\[ L_{1,\ell}^\alpha(f|_\rho)\le \sum_{k=0}^\ell    \binom{d}{\ell-k}\cdot \max_{\alpha'}L_{1,k}^{\alpha'}(f'|_{\rho}), \]
where the maximum is over $\alpha'$, another family of signs. %Furthermore, if $\alpha$ are level-$\ell$ signs as in~\Cref{def:alpha}, then so is $\alpha'$ when $k=\ell$
\label{thm:hybrid}
\end{theorem}

\begin{corollary} \label{cor:hybrid} Analogues of~\Cref{thm:main_theorem_dqck,thm:main_theorem_bqp} hold even for algorithms with classical pre-processing.
\end{corollary}
 
\begin{proof}[Proof of~\Cref{thm:hybrid}]
We view the classical phase as a decision tree of depth $d$ with $2^d$ leaves where each leaf $y$ selects an algorithm $f_y$ to run. Furthermore, we view each leaf $y$ as a partial assignment in $\{-1,1,*\}^N$ where the coordinates that are queried are assigned $\pm 1$ depending on the outcome of the query, and the coordinates not queried are assigned $*$. We use $y^{-1}(*)$ to denote the coordinates of $y$ that are alive. We know that $|y^{-1}(*)| \ge N-d$. This defines a restriction $\rho_y\in \{-1,1,*\}^{N}$ of the variables which restricts the $i$-th coordinate to $y_i$ if $y_i\in\bin$ and leaves it alive otherwise. We can assume that any $y$ that is ever traversed is consistent with $\rho$. For any such $y$, let $f_y(x)$ be the acceptance probability of the algorithm chosen conditioned on receiving $y$ in the first stage. Consider:
\begin{align*}
&L_{1,\ell}^\alpha(f|_\rho)\\
&=\E_{x\sim\{-1,1\}^N}\sbra{\sum_{|S|=\ell} \alpha_S \cdot f|_\rho(x)\cdot \chi_S(x)}\\
&=\E_{y\text{ consistent with }\rho}\sbra{\E_{\substack{x\sim \bin^N \\\text{consistent with }y}}\sbra{\sum_{|S|=\ell} \alpha_S \cdot f_y(\rho(x))\cdot \chi_S(x) }} \\
&= \E_{y\text{ consistent with }\rho}\sbra{\E_{\substack{x\sim \bin^N \\\text{consistent with }y}}\sbra{\sum_{k=0}^\ell \sum_{\substack{S_1\subseteq y^{-1}(*)\\ S_2\subseteq [N]\setminus y^{-1}(*) \\ |S_1|=k,|S_2|=\ell-k}} \alpha_{S_1\cup S_2} \cdot f_y(\rho(x))\cdot \chi_{S_1}(x)\cdot \chi_{S_2}(x)}}.
\end{align*}
Fix a leaf $y$ that maximizes the above quantity. Since  we are only taking expectations over $x$ consistent with $y$, we can replace $\chi_{S_2}(x)$ by $\chi_{S_2}(y)$ in the R.H.S. above and similarly, $\rho(x)$ only depends on the variables in $ S_1$. Once we do this, $x$ is completely free of $y$ and we can replace the expectation of $x\sim\{-1,1\}^N$ consistent with $y$ by simply $x\sim\{-1,1\}^N$. We obtain 

\begin{align}\label{eq:preprocess_1}  
L_{1,\ell}^\alpha(f|_\rho)&\le  \E_{x\sim \bin^N }\sbra{\sum_{k=0}^\ell \sum_{\substack{S_1\subseteq y^{-1}(*)\\ S_2\subseteq [N]\setminus y^{-1}(*) \\ |S_1|=k,|S_2|=\ell-k}} \alpha_{S_1\cup S_2}\cdot \chi_{S_2}(y)\cdot f_y(\rho(x))\cdot \chi_{S_1}(x) } \end{align}
Since $|\chi_{S_2}(y)|\le 1$, applying Triangle Inequality gives
\begin{align}\label{eq:preprocess_2}  \begin{split} 
L_{1,\ell}^\alpha(f|_\rho) &\le \sum_{k=0}^\ell \sum_{\substack{S_2\subseteq [N]\setminus y^{-1}(*) \\ |S_2|=\ell-k}} \abs{\E_{x\sim \bin^N }\sbra{\sum_{\substack{S_1\subseteq y^{-1}(*)\\ |S_1|=k}} \alpha_{S_1\cup S_2}\cdot f_y(\rho(x))\cdot \chi_{S_1}(x)}} 
%\\ &=\sum_{k=0}^\ell \sum_{\substack{S_2\subseteq [N]\setminus y^{-1}(*) \\ |S_2|=\ell-k}} \abs{ \sum_{\substack{S_1\subseteq y^{-1}(*)\\ |S_1|=k}} \alpha_{S_1\cup S_2}\cdot \widehat{f_y|_\rho}(S_1) }.
\end{split}\end{align} 
Define $\gamma\in[-1,1]^N$ by $\gamma_i=1$ if $i\in y^{-1}(*)$ and 0 otherwise. For any fixed $k\in\{0,\ldots,\ell\}$ and $S_2\subseteq [N]\setminus y^{-1}(\star)$ of size $ \ell-k$, define signs ${\alpha}^{S_2}$ that are non-zero only for $S_1\subseteq [N]$ with size $k$ so that \[\alpha^{S_2}_{S_1}:=\alpha_{S_1\cup S_2}\cdot \chi_{S_1}(\gamma).\] Observe that $\chi_{S_1}(\gamma)=1$ if $S_1\subseteq y^{-1}(*)$ and 0 otherwise. Thus, 
\[ \sum_{\substack{S_1\subseteq y^{-1}(*)\\|S_1|=k}} \alpha_{S_1\cup S_2} \cdot \chi_{S_1}(x) =\sum_{\substack{S_1\subseteq [N]\\|S_1|=k}} \alpha_{S_1\cup S_2}\cdot \chi_{S_1}(\gamma) \cdot \chi_{S_1}(x)\triangleq \sum_{\substack{S_1\subseteq [N]\\|S_1|=k}}  \alpha_{S_1}^{S_2}\cdot \chi_{S_1}(x).\]
Finally, we observe that 
\[ \E_{x\sim \bin^N }\sbra{ \sum_{\substack{S_1\subseteq [N]\\|S|=k}} \alpha^{S_2}_{S_1}\cdot f_y(\rho(x)) \cdot \chi_{S_1}(x)}\triangleq \sum_{\substack{S_1\subseteq [N]\\|S|=k}} \alpha^{S_2}_{S_1}\cdot \widehat{f_y|_\rho}(S_1)\triangleq L_{1,k}^{\alpha^{S_2}}(f_y|_{\rho}).\]
Substituting this in~\Cref{eq:preprocess_2}, we get
\[L_{1,\ell}^\alpha(f)\le \sum_{k=0}^\ell \binom{d}{\ell-k} \cdot \max_{\alpha'}L_{1,k}^{\alpha'}(f_y|_{\rho}),\]
where we used the fact that $N-|y^{-1}(*)|\le d$. %Furthermore, when $k=\ell$ (i.e., when $S_2=\emptyset$), it is easy to see that if $\alpha$ is a family of signs as in~\Cref{def:alpha}, so is the result family $\alpha^{\emptyset}$. 
This completes the proof. 
\end{proof}

\begin{proof}[Proof of~\Cref{cor:hybrid} from~\Cref{thm:hybrid}]
Let $\cF$ (respectively $\cF'$) denote the class of $d$-query $\DQC{k}$ algorithms with (respectively without) classical pre-processing. Applying~\Cref{thm:hybrid}, we have 
\[ L_{1,\ell}(\cF)\le \sum_{k=0}^\ell \binom{d}{\ell-k}\cdot L_{1,k}^{\alpha'}(f|_{\rho'}). \]
We now apply~\Cref{thm:main_theorem_dqck} to bound each $ L_{1,k}^{\alpha'}(f|_{\rho'})$ and this gives 
\[ L_{1,\ell}(\cF)\le  \sum_{k=0}^{\ell} \binom{d}{\ell-k}\cdot c^k\cdot d^3 \cdot N^{(k-2)/2}\cdot \log(N)^{k-2}\cdot \sqrt{k!} \le \tilde{O}\pbra{ d^{\ell+3}\cdot N^{(\ell-2)/2}\cdot \sqrt{\ell !}}\]
as desired. The proof for $\BQP$ algorithms is identical and we obtain a bound of $\tilde{O}\pbra{d^{\ell+1}\cdot N^{(\ell-1)/2}\cdot \sqrt{\ell!}}$.

\end{proof}

\subsection{Simulating $\DQC{k}$ algorithms by $\DQC{k-t}$ algorithms.}
\label{sec:simulation}

\begin{figure}
\centering
\mbox{ 
\Qcircuit @C=1em @R=.7em {
\lstick{n+w} & \qw & \qw &\qw &\multigate{4}{U} & \meter \\
\lstick{} & & \gate{X} &\ctrl{1} &\ghost{U} &\meter   \\ 
\lstick{} & & \gate{X} &  \ctrl{3} & \ghost{U} & \meter   \inputgroupv{2}{3}{.8em}{.8em}{t+1\hspace{1em}} \\
\lstick{}  & \ket{0} &  & \qw&\ghost{U} &\meter \\
\lstick{} &  \ket{0} &  &\qw & \ghost{U} &\meter \inputgroupv{4}{5}{.8em}{.8em}{k-t-1\hspace{3em} } \\
\lstick{} &\ket{0}  &  &  \targ \qw &  \qw  & \meter   \\
}
}
\caption{Simulating a $\DQC{k}$ algorithm by a $\DQC{k-t}$ algorithm.}
\label{fig:qcircuit}
\end{figure}

\begin{claim} \label{claim:simulation}Let $g(x)$ be the bias of a $d$-query $\DQC{k}$ algorithm. Then, there is a $d$-query $\DQC{k-t}$ algorithm whose bias is $g(x)\cdot 2^{-t-1}$.
\end{claim}
\begin{proof}[Proof of~\Cref{claim:simulation}]
Given a $d$-query $\DQC{k}$ algorithm with $n+w$ noisy bits, consider a $\DQC{k-t}$ algorithm which uses $n+w+t+1$ noisy bits and $k-t$ clean qubits as follows. Firstly, the algorithm applies the $X$ gate to the last $t+1$ noisy qubits and applies a Toffoli controlled on these qubits with the target as the final clean qubit. Then, apply the $\DQC{k}$ algorithm on the first $n+w$ noisy qubits and the first $k$ clean qubits. Finally, measure the last clean qubit. If it results in an outcome 1, then return the outcome of the $\DQC{k}$ algorithm, otherwise, return a random bit (by taking an additional noisy qubit for instance). 

Observe that this algorithm behaves identically to the original one whenever the $t+1$ noisy qubits are in the all-zeroes state, which happens with probability $2^{-t-1}$. In all other cases, the algorithm returns a uniformly random bit. Thus, the bias of the resulting algorithm is $2^{-t-1}\cdot g(x)$.  
\end{proof}

\subsection{Acceptance Probability of Quantum Algorithms}

\paragraph*{$\DQC{k}$ algorithms.}
\label{sec:appendix_acceptance_probability}
\begin{proof}[Proof of~\Cref{claim:dqck_acceptance_probability}]
    Consider a $d$-query $\DQC{k}$ algorithm and let $U_1,\ldots,U_{d+1}$ be the unitary operators of the algorithm and $\start=[NW]\times\{1\},\final\subseteq[NWK]$ be the set of initial and accepting final states as in~\Cref{def:dqck} and~\Cref{fig:dqc1}. The final state of the algorithm can be expressed as a uniform mixture over $\iwk{1}\in\start$ of the pure state $ U_{d+1}\cdot (O_x\otimes \bI) \cdot U_d\cdots (O_x\otimes \bI)\cdot  U_1 \ket{\iwk{1}}$. Let $\final\subseteq [NWK]$ be the subset of final basis states that is accepted by the algorithm. We can thus express the acceptance probability of the algorithm as an average over $\iwk{1}\in\start$ of
\[\underset{ \substack{\iwk{d+2}\in \final}}{\sum} \abs{\bra{\iwk{d+2}} U_{d+1}\cdot O\cdot U_d\cdots O\cdot  U_1 \ket{\iwk{1}}}^2\]
Since there are $NW$ elements in $\start$, the overall acceptance probability of the algorithm is given by
\begin{align*}
f(x)&:=    \tfrac{1}{NW}\underset{ \substack{\iwk{1}\in\start\\ \iwk{d+2}\in \final}}{\sum} \abs{\bra{\iwk{d+2}} U_{d+1}\cdot O \cdot U_d\cdots O\cdot  U_1 \ket{\iwk{1}}}^2 \\ 
& = \tfrac{1}{NW}\underset{ \substack{\iwk{1}\in \start \\ \iwk{d+2}\in \final}}{\sum}   \bra{\iwk{1}}  U_1^\dagger\cdot O \cdots U_d^\dagger\cdot O\cdot U_{d+1}^\dagger \ket{\iwk {d+2}} \cdot  \bra{\iwk{d+2}} U_{d+1}\cdot O\cdot U_d\cdots O\cdot U_1\ket{\iwk{1}}  
\\ &= \tfrac{1}{NW}\underset{ \substack{\iwk{1}\in \start \\ \iwk{d+2}\in \final}}{\sum}   \Tr\pbra{U_1 \ket{\iwk{1}}\bra{\iwk{1}}  U_1^\dagger\cdot O \cdots U_d^\dagger\cdot O\cdot U_{d+1}^\dagger\ket{\iwk {d+2}}  \bra{\iwk{d+2}}U_{d+1} \cdot O\cdot U_d\cdots \cdot U_2\cdot O  } \\
&=\tfrac{1}{NW} \Tr\pbra{U_1 \pbra{\underset{ \iwk{1}\in \start}{\sum}  \ket{\iwk{1}}\bra{\iwk{1}}}  U_1^\dagger\cdot O \cdots  O \cdot U_{d+1}^\dagger\pbra{\underset{\iwk{d+2}\in \final}{\sum}   \ket{\iwk {d+2}}  \bra{\iwk{d+2}}} U_{d+1}\cdot O \cdots U_2\cdot O }.
\end{align*}
We will further simplify this expression by introducing $M\times M$ matrices $V_1,\ldots,V_{2d}$ as follows. Let $V_1=\sum_{\iwk{1}\in \start} U_1\ket{\iwk{1}}\bra{\iwk{1}}U_1^\dagger$. For $t\in[2,d]$, let $V_{t}:=U_{t}^\dagger$. Let $V_{d+1}= \sum_{\iwk{d+2}\in \final} U_{d+1}^\dagger\ket{\iwk{d+2}}\bra{\iwk{d+2}}U_{d+1}$ and for $t\in[d-1]$, let $V_{d+1+t}:=U_{d-t+1}$. This allows us to express $f(x)$ as 
\begin{align*} f(x)&=(NW)^{-1}\cdot \Tr\pbra{ V_1\cdot O\cdot V_2\cdot O\cdots   V_{2d} \cdot O }.\end{align*}
This gives us the desired expression. Finally we observe that $\|V_t\|_\op \le 1$ for all $t$, and $V_1$ is (up to multiplication by unitary matrices) equal to a diagonal matrix with at most $|\start|= NW$ non-zero entries of value 1, hence $\|V_1\|_\frob \le \sqrt{NW}$. 
\end{proof}

\paragraph*{$\hBQP$ algorithms.}

\begin{proof}[Proof of~\Cref{claim:acceptance_probability_hbqp}]
Let $\final$ be the accepting pairs of initial and final states of a $\hBQP$ algorithm and $U_1,\ldots,U_{d+1}$ be unitary operators as in~\Cref{def:hbqp_algorithm} and~\Cref{fig:hbqp}. For $\iw{1},\iw{d+1}$, we use $F_{\iw{1},\iw{d+1}}$ to denote 1 when $(\iw{1},\iw{d+1})\in \cF$ and 0 otherwise.  It is fairly straightforward to see that the acceptance probability $f(x)$ of the algorithm is given by
\begin{align*} \begin{split}f(x)&:=M^{-1} \sum_{\iw{1},\iw{d+2}\in[M]}F_{\iw{1},\iw{d+2}}\cdot \abs{\bra{\iwk{d+2}} U_{d+1}\cdot O \cdot U_d\cdots O\cdot  U_1 \ket{\iwk{1}}}^2  \\
&=M^{-1}\sum_{\iw{1},\iw{d+2}\in[M]}F_{\iw{1},\iw{d+2}}\cdot \bra{\iwk{1}} U_1^\dagger\cdot O \cdots O\cdot  U_{d+1}^\dagger  \ket{\iw{d+2}}\bra{\iwk{d+2}}  U_{d+1}\cdot O\cdots O\cdot  U_1 \ket{\iwk{1}}  \end{split} 
\end{align*} 
as desired.
 \end{proof}

\begin{comment}
    
\begin{proof}[Proof of~\Cref{claim:acceptance_probability_bqp}]
    Let $U_0,\ldots,U_{d}$ be the $M\times M$ unitary matrices applied by the algorithm and $F\subseteq [M]$ be the set of accepting final states as in~\Cref{def:bqp_algorithm} and~\Cref{fig:bqp}. Let $\Pi_F$ be the $M\times M$ diagonal matrix whose $i$-th entry is $0$ if $i\notin F$ and $1$ otherwise. Let $I_1=\ket{0\ldots 0}$.  Observe that the acceptance probability of the algorithm on input $x$ is precisely
\begin{align*}
    f(x)&:= \braket{ I_1\gap  U_0^\dagger \cdot O\cdots U_{d-1}^\dagger \cdot O\cdot U_{d}^\dagger \cdot \Pi_F\cdot U_{d}\cdot O\cdots   O\cdot U_0\gap  I_1}
\end{align*}
where $O=O_x\otimes \bI$. Let $v=U_0\ket{0,\ldots,0}$. Define matrices $V_i$ for $i\in[2d+1]$ as follows. For $i\in[d-1],$ $V_i:=U_i^\dagger$, $V_{d}=U_{d}^\dagger\cdot \Pi_F\cdot U_{d},$ and for $i\in[d-1]$, $V_{2d-i}=U_i$. Observe that $\|V_i\|_\op\le 1$ for all $i\in[2d-1]$, furthermore, 
\begin{align*}
    f(x)&:= \braket{v\gap   O\cdot V_1\cdots   V_{2d-1}\cdot O\gap v}.
\end{align*}
This completes the proof. 
\end{proof}
\end{comment}

\subsection{Fourier Coefficients of Quantum Algorithms}

\label{sec:appendix_fourier_coefficients}

\paragraph*{$\DQC{k}$ Algorithms.}

\begin{proof}[Proof of~\Cref{claim:dqck_fourier_coefficients}]
From~\Cref{claim:dqck_acceptance_probability}, the acceptance probability $f(x)$ of a $d$-query $\DQC{k}$ algorithm is given by $f(x)$ where 
\begin{align}\begin{split} f(x)&=(NW)^{-1}\cdot \Tr\pbra{ (O_x\otimes \bI)\cdot V_1 \cdots (O_x\otimes \bI)\cdot V_{2d} }\\ 
&= (NW)^{-1} \sum_{\iwk{1},\ldots,\iwk{2d}\in [M]} \prod_{t\in[2d]}\pbra{V_t[\iwk{t}\gap \iwk{t+1}] \cdot x_{i_t}}\end{split} \label{eq:dqc1_1}
\end{align}
with the convention that $\iwk{2d+1}=\iwk{1}$. 
We now replace $x$ by $\rho(x)$ in~\Cref{eq:dqc1_1} to obtain 
\begin{align}\label{eq:dqc1_2}
f(\rho(x))&=(NW)^{-1} \sum_{\iwk{1},\ldots,\iwk{2d}\in [M]} \prod_{t\in[2d]}\pbra{V_t[\iwk{t}\gap \iwk{t+1}] \cdot \rho(x)_{i_t}}
\end{align}
Since the coordinates in $L$ are unfixed and the rest are fixed, \[\rho(x)_{i_t}=\begin{cases} x_{i_t} & \text{if } i_t \in L\\  \rho_{i_t} &\text{if } i_t\notin L\end{cases}.\] 
In particular, 
\[ \prod_{t\in[2d]}\rho(x)_{i_t}=\pbra{\prod_{\substack{t\in[2d]\\\text{with }i_t\notin L}}\rho_{i_t}}\cdot \pbra{\prod_{\substack{t\in[2d]\\\text{with }i_t\in L}}x_{i_t}} \]
Substituting this in~\Cref{eq:dqc1_2}, we get
\begin{align}\label{eq:dqc1_3}
f(\rho(x))&=(NW)^{-1} \sum_{\iwk{1},\ldots,\iwk{2d}\in [M]} \prod_{t\in[2d]}\pbra{V_t[\iwk{t}\gap \iwk{t+1}]}\cdot \pbra{\prod_{\substack{t\in[2d]\\\text{with }i_t\notin L}}\rho_{i_t}}\cdot \pbra{\prod_{\substack{t\in[2d]\\\text{with }i_t\in L}}x_{i_t}}.
\end{align}
To simplify this expression and get rid of the $\rho_{i_t}$, we will define a $M\times M$ diagonal matrix $D^\rho$ and $M\times M$ unitary matrices $V^\rho_t$ for $t\in[2d]$ as follows. For $\iwk{}\in[M]$, define $D^\rho$ to be a diagonal matrix whose $\iwk{}$-th diagonal entry is $\rho_i$ if $i\notin L$ and $1$ otherwise. Define $V_t^\rho=D^\rho\cdot V_t$ for all $t\in[2d]$. Observe this allows us to simplify~\Cref{eq:dqc1_3} and obtain
\begin{align*} 
f(\rho(x))&=(NW)^{-1} \sum_{\iwk{1},\ldots,\iwk{2d}\in [M]} \pbra{\prod_{t\in[2d]}V_t^\rho[\iwk{t}\gap \iwk{t+1}]}\cdot \pbra{ \prod_{\substack{t\in[2d]\\\text{with }i_t\in L}} x_{i_t}}
\end{align*}
From here, we see that the only non-zero Fourier coefficients correspond to $S\subseteq L$ and satisfy the defining equation as in~\Cref{claim:dqck_fourier_coefficients}. The bounds on the norms of $V_t^\rho$ follow immediately from the corresponding bounds on $V_t$ from~\Cref{claim:dqck_acceptance_probability} and the fact that $\|D^\rho\|_\op \le 1$. \end{proof} 

\paragraph*{$\hBQP$ Algorithms.}

\begin{proof}[Proof of~\Cref{claim:fourier_coefficients_hbqp}]

 Recall from~\Cref{claim:acceptance_probability_hbqp}
 that the acceptance probability of a $d$-query $\hBQP$ algorithm is given by $f(x)$ where
\begin{align}\label{eq:hbqp_0_acceptance_probability} 
f(x)&:=M^{-1}\sum_{\iw{1},\iw{d+2}\in[M]}F_{\iw{1},\iw{d+2}}\cdot \bra{\iwk{1}} U_1^\dagger\cdot O \cdots O\cdot  U_{d+1}^\dagger  \ket{\iw{d+2}}\bra{\iwk{d+2}}  U_{d+1}\cdot O\cdots O\cdot  U_1 \ket{\iwk{1}} .\end{align}
To simplify notation, for all $t\in[d+1]$, we define $V_t:=U_t^\dagger$ and $V_{2d+3-t}=U_t$. Substituting this in~\Cref{eq:hbqp_0_acceptance_probability}, we get
\begin{align} \label{eq:hbqp_1}f(x):=M^{-1} \sum_{\iw{1},\iw{d+2}\in[M]} F_{\iw{1},\iw{d+2}}\underset{\substack{\iw{2},\ldots,\iw{d+1}\in[M]\\\iw{d+3},\ldots,\iw{2d+2}\in[M]}}{\sum} \pbra{\prod_{t\in[2d+2]} V_t[\iw{t}\gap\iw{t+1}] }\cdot \pbra{ \prod_{t\in[2d+2]\setminus\{1,d+2\}}x_{i_t}}.\end{align}
Substituting $\rho(x)$ in place of $x$ in~\Cref{eq:hbqp_1}, we get 
\begin{align}\label{eq:hbqp_2} 
\begin{split}f(\rho(x))
&=M^{-1} \sum_{\iw{1},\iw{d+2}\in[M]} F_{\iw{1},\iw{d+2}} \underset{\substack{\iw{2},\ldots,\iw{d+1}\in[M]\\\iw{d+3},\ldots,\iw{2d+2}\in[M]}}{\sum}  \pbra{ \prod_{t\in[2d+2]}  V_t[\iw{t}\gap\iw{t+1}]}\\
&\cdot \pbra{\prod_{\substack{t\in[2d+2]\setminus\{1,d+2\}\\\text{with }i_t\notin L}} \rho_{i_t} } \cdot \pbra{\prod_{\substack{t\in[2d+2]\setminus\{1,d+2\}\\\text{with }i_t\in L}} x_{i_t} }. 
\end{split}\end{align}
As in the proof of~\Cref{claim:dqck_fourier_coefficients}, we will simplify this expression by defining $D^\rho$ to be a diagonal matrix whose $\iw{}$-th diagonal entry is $\rho_i$ if $i\notin L$ and $1$ otherwise and let $V^\rho_1=V_1$, $V^\rho_{d+2}=V_{d+2}$ and let $V^\rho_t=D^\rho\cdot V_t$ for $t\neq 1,d+2$. This allows us to simplify~\Cref{eq:hbqp_2} as 
\begin{align*} f(\rho(x))
&=M^{-1} \sum_{\iw{1},\ldots,\iw{2d+2}\in[M]} F_{\iw{1},\iw{d+2}}  \prod_{t\in[2d+2]}  V^\rho_t[\iw{t}\gap\iw{t+1}] \cdot \pbra{\prod_{\substack{t\in[2d+2]\setminus\{1,d+2\}\\\text{with }i_t\in L}} x_{i_t} }. 
\end{align*}
From here, we see that only Fourier coefficients with $S\subseteq L$ are non-zero and are given by the defining equation in~\Cref{claim:fourier_coefficients_hbqp}. The norm bounds on $V_t^\rho$ follow immediately from the corresponding bounds in~\Cref{claim:acceptance_probability_hbqp}. This completes the proof.
\end{proof}

\end{document}